\documentclass[sn-mathphys,Numbered]{sn-jnl}

\usepackage{graphicx}%
\usepackage{multirow}%
\usepackage{amsmath,amssymb,amsfonts}%
\usepackage{amsthm}%
\usepackage{mathrsfs}%
\usepackage[title]{appendix}%
\usepackage{xcolor}%
\usepackage{textcomp}%
\usepackage{comment}
\usepackage{float}
\usepackage{manyfoot}%
\usepackage{booktabs}%
\usepackage{algorithm}%
\usepackage{algorithmicx}%
\usepackage{algpseudocode}%
\usepackage{cleveref}
\usepackage[section]{placeins}
\usepackage{listings}%

\DeclareMathOperator*{\argmin}{arg\,min}
\usepackage{tikz-cd}

\def\bB{\mathbb B}

\newcommand{\pluseq}{\mathrel{+}=}

\usepackage{subcaption}

\crefname{equation}{Eq.}{Eqs.}
\crefname{figure}{Fig.}{Figs.}
\crefname{section}{Section}{}

\newtheorem{theorem}{Theorem}
\newtheorem{proposition}[theorem]{Proposition}%
\newtheorem{definition}{Definition}%

\begin{document}

\title[Article Title]{Consensus ranking by quantum annealing}


\author[1]{\fnm{Daniele} \sur{Franch}}\email{daniele.franch@unitn.it}
\author[1]{\fnm{Enrico} \sur{Zardini}}\email{enrico.zardini@unitn.it}
\author[1,3]{\fnm{Enrico} \sur{Blanzieri}}\email{enrico.blanzieri@unitn.it}
\author[2,3]{\fnm{Davide} \sur{Pastorello}}\email{davide.pastorello3@unibo.it}

\affil[1]{\orgdiv{Department of Information Engineering and Computer Science}, \orgname{University of Trento}, \orgaddress{\street{via Sommarive 9}, \city{Povo}, \postcode{38123}, \state{Trento}, \country{Italy}}
\\}
\affil[2]{\orgdiv{Department of Mathematics}, \orgname{University of Bologna}, \orgaddress{\street{Piazza di Porta San Donato 5}, \city{Bologna}, \postcode{40126}, \state{Bologna}, \country{Italy}}
\\}
\affil[3]{\orgdiv{TIFPA-INFN}, \orgaddress{\street{via Sommarive 14}, \city{Povo}, \postcode{38123}, \state{Trento}, \country{Italy}}}


\abstract{Consensus ranking is a technique used to derive a single ranking that best represents the preferences of multiple individuals or systems. It aims to aggregate different rankings into one that minimizes overall disagreement or distance from each of the individual rankings. Kemeny ranking aggregation, in particular, is a widely used method in decision-making and social choice, with applications ranging from search engines to music recommendation systems. It seeks to determine a consensus ranking of a set of candidates based on the preferences of a group of individuals.
However, existing quantum annealing algorithms face challenges in efficiently processing large datasets with many candidates.
In this paper, we propose a method to improve the performance of quantum annealing for Kemeny rank aggregation. 
Our approach identifies the pairwise preference matrix that represents the solution list and subsequently reconstructs the ranking using classical methods. This method already yields better results than existing approaches. Furthermore, we present a range of enhancements that significantly improve the proposed method's performance, thereby increasing the number of candidates that can be effectively handled.
Finally, we evaluate the efficiency of our approach by comparing its performance and execution time with that of KwikSort, a well-known approximate algorithm.}


\maketitle

\section{Introduction}\label{sec1}

The Kemeny ranking method \cite{kemeny1} is designed to minimize the total number of disagreements among individual rankings, aiming to find a ranking that closely aligns with all individual preferences by minimizing the required pairwise swaps. This approach offers not just a single solution but a range of potential solutions, empowering users to choose the one that best fits their subjective criteria, which may be challenging or impossible to formally model.
Throughout this paper, we use the terms {\em ranking} and {\em list} interchangeably to refer to the ordered sequence of items or candidates.
The Kemeny ranking is defined as the list that minimizes the cumulative Kendall-Tau distance across all lists in the dataset. The Kendall-Tau distance acts as a measure to quantify the variance between two rankings, assessing all inconsistencies in the relative order of candidate pairs between the rankings. This distance metric is widely used in ranking assessments due to its computational simplicity and its ability to capture the importance of the order of items in rankings.
By leveraging the Kemeny method and the Kendall-Tau distance, it becomes possible to effectively assess and compare rankings while accounting for disagreements between them. This approach provides a practical and insightful solution for ranking challenges, enabling users to make well-informed decisions tailored to their specific needs and preferences.

Quantum computing (QC) is a revolutionary computational paradigm that leverages the fundamental principles of quantum mechanics to process information. In recent years, significant strides have been made in the development of practical quantum computers, making them increasingly accessible in the market. This progress holds immense potential for transforming the fields of artificial intelligence and machine learning. QC's unique ability to provide efficient solutions to various search and optimization problems commonly encountered in these domains is particularly noteworthy.\\
In the realm of Kemeny ranking solutions, researchers have proposed leveraging quantum annealers to address the ranking aggregation problem \cite{fiergolla2023heuristicdiversekemenyrank}, showcasing the promise of this approach. However, these attempts face limitations in terms of performance. One key drawback is that the representation employed requires a significantly large number of qubits, which is challenging for the current quantum annealers. Additionally, it requires a high penalty coefficient term, further increasing the complexity and leading to suboptimal performance.

Given this context, we have chosen to adopt a hybrid approach to address the challenges encountered by the annealer. Our strategy involves the annealer primarily focusing on identifying the binary representation of the optimal list, while preprocessing and postprocessing (e.g., list reconstruction) are performed using classical methods. This approach aims to improve the overall performance and overcome the current limitations by reducing the complexity placed on the annealer and shifting it to classical computation.
Our approach starts by determining a preference matrix that quantifies the relative preferences between each pair of candidates. Each element in this matrix indicates how many times a candidate is preferred over another in the dataset. We then search for the optimal rank, represented as a binary preference matrix, that best aligns with the preferences in the matrix.
However, a significant challenge arises in dealing with cycles, which represent violations of the transitivity property in the binary preference matrix.\\
In some practical cases, a disagreement may arise among three candidates, creating a cyclic inconsistency, where there is no clear ordering among the three candidates. For example, $c1\prec c2\prec c3\prec c1$ is impossible to realize in a valid rank since $\prec$ is not a total order relation.\\
To address this challenge, we introduce a penalty term that discourages the existence of such cycles. This addition enables us to determine the optimal rank effectively, eliminating any cyclic inconsistencies. By combining the strengths of the quantum annealer and classical computation, we aim to find an effective solution that minimizes the Kemeny distance while handling cycle removal in the most efficient way.
Overall, the combination of quantum and classical techniques in our hybrid approach holds the potential to leverage the strengths of both paradigms, paving the way for more efficient and effective solutions to the Kemeny ranking problem.

\subsection{Structure of the paper}

The structure of the paper is organized into several sections. Section 2, titled ``Background," presents a summary of the key concepts and theoretical framework pertinent to our study, discussing the essential theory required to comprehend the subject. Section 3, ``Representation of the Kemeny Aggregation into a Quantum Annealer," describes the methodology employed and the enhancements achieved. Section 4, ``Kemeny Ranking in Multiagent Systems," examines the potential application of the proposed method for Kemeny ranking within decision-making in multiagent systems. Section 5, ``Methodology," elaborates on the research approach we adopted. Section 6, ``Results," reveals the primary findings, giving thorough explanations of the discoveries and significant observations, including results for each proposed variation and a comparison with KwikSort. Section 7, ``Related Work," examines a paper that suggested a similar approach but did not implement the improvements we propose. Section 8, ``Conclusions and Outlook," distills the main insights and contributions of the study, indicating possible areas for future research. Section 9, ``Acknowledgements," conveys appreciation to those who assisted and supported the research efforts.

\section{Background}
This section delves deeply into a variety of topics, providing crucial theoretical foundations necessary for understanding the proposed model. The topics covered include QUBO problems, quantum annealing employing D-Wave technology, Kemeny ranking, Kendall-Tau distance, as well as partial and weighted lists.

\subsection{QUBO problems}
Quadratic Unconstrained Binary Optimization (QUBO) problems are optimization problems of the form:
\begin{equation}
    \argmin_x x^TQx,
\end{equation}
\\
where $x$ is a binary vector, and Q is an upper triangular (or symmetric) matrix of real values.
If we suppose $x$ to be a $n\times 1$ vector and Q an upper triangular $n\times n$ matrix the QUBO problem can be reformulated by separating the diagonal elements of the matrix from the off-diagonal elements. This results in a formulation that is equivalent to the interaction form that appears in the quantum annealer:
\begin{equation}
    x^TQx=\sum_{i=1}^n Q_{ii}x_i+\sum_{i<j}^n Q_{ij}x_ix_j,
\end{equation}
where we exploit the fact that $x_i\in\mathbb{B}=\{0,1\}$ and so that $x_i^2=x_i$.
Although QUBO problems are typically unconstrained by name, we can still introduce constraints by penalizing configurations that do not adhere to them. This is done by adding penalty terms to the objective function that increase the energy of solutions that violate the constraints.
The key benefit of QUBO is that it can be directly mapped to the Ising model by simply changing the domain from \{0,1\} to \{-1,1\}, allowing us to leverage quantum annealers to solve QUBO problems.

\subsection{Quantum annealing}

Quantum annealing (QA) is a heuristic search used to solve optimization problems~\cite{kadowaki}. The solution of a given problem corresponds to the \emph{ground state} (the less energetic physical state) of a $n$-qubit system with energy described by a \emph{problem Hamiltonian} $H_P$, which is a self-adjoint linear operator on the Hilbert space $(\mathbb C^2)^{\otimes n}$.
The annealing procedure is implemented by a time evolution of the quantum system towards the ground state of the problem Hamiltonian. More precisely, let us consider the time-dependent Hamiltonian
\begin{equation}
H(t)=\Gamma(t) H_D+H_P\qquad t\in[0,\tau],
\end{equation}
where $H_P$ is the problem Hamiltonian, and $H_D$ is the \emph{transverse field Hamiltonian}, which gives the kinetic term that does not commute with $H_P$ allowing quantum superpositions improving the exploration of the solution landscape. $\Gamma:[0,\tau]\rightarrow [0,1]$ is a smooth decreasing function that attenuates the kinetic term, driving the system towards the ground state of $H_P$.
In practical terms, the problem Hamiltonian can be represented as a {\em quantum Ising spin glass}:
\begin{equation}\label{eq:Hp}
H_P =\sum_{i=1}^n h_i\sigma_z^{(i)} + \sum_{i<j}^nJ_{ij}\sigma_z^{(i)}\sigma_z^{(j)}
\end{equation}
where $\sigma_z^{(i)}$ denotes the linear operator defined on the $n$-qubit Hilbert space $(\mathbb C^2)^{\otimes n}$ which acts as the $Z$-Pauli matrix on the $i$-th qubit and as the identity elsewhere, while $J_{ij}$ and $h_i$ are real parameters that fully characterize the optimization problem by describing the interaction and local field terms, respectively. The transverse field Hamitlonian is $H_D=\sum_i d_i\sigma_x^{(i)}$, where $d_i\in\mathbb R$ and $\sigma_x^{(i)}$ acts as the $X$-Pauli matrix on the $i$-th qubits and trivially elsewhere.
The ground state of the problem Hamiltonian corresponds to the spin configuration $\mathbf{z}^*\in\{-1,1\}^n$ that minimizes the energy function:
\begin{equation}
E(\mathbf{J},\mathbf{b},\mathbf{z}) = \sum_{i} h_iz_i + \sum_{i<j}^n J_{ij}z_iz_j, \hspace{0.3cm} \mathbf{z}=(z_1,...,z_N) \in \{-1,1\}^n
\end{equation}
\\
While the problem Hamiltonian (\ref{eq:Hp}) represents a fully connected spin glass, achieving this level of connectivity in reality is infeasible for actual quantum annealers. In fact, the limited connectivity between qubits, along with the restricted number of available qubits, are significant constraints faced by current quantum annealing hardware. Addressing the existing challenges related to hardware limitations is necessary to fully exploit the advantages of quantum annealers and pave the way for solving more complex and practical optimization problems.

\subsection{Kemeny ranking method}

Kemeny ranking is a method used to compute a consensus ranking from individual preferences. It determines the ranking that minimizes the total number of pairwise disagreements among all individual rankings. This approach proves particularly useful for achieving a compromise or establishing a collective preference in the presence of diverse opinions. Kemeny ranking efficiently aggregates individual preferences, offering a solution that balances fairness and consensus.

Let us recall that a {\em total preorder} (or {\em weak order}) on a set $X$ is a binary relation $\prec$ on $X$ that is reflexive ($x\prec x$, $\forall x\in X$), transitive ($x\prec y \mbox{ and } y\prec z \Rightarrow x\prec z$, $\forall x,y,z\in X$), and strongly connected ($\forall x,y\in X$, $x\prec y$ or $y\prec x$). If a total preorder $\prec$ is also antisymmetric ($x\prec y$ and $y\prec x$ $\Rightarrow$ $x=y$) then it is said to be a {\em total order}.

\vspace{0.5cm}

\begin{definition}
    A {\bf ranking} (or a {\bf list}) over a set $C$ is a total order on $C$. The elements of a set on which a ranking is defined are called {\bf candidates}. If $\prec$ is a ranking over $C$ and $c\prec c'$ then $c$ is said to be {\bf preferred} to $c'$. \\
    A {\bf ranking with ties} over a set $C$ is a total preorder on $C$.
    
\end{definition}

\vspace{0.5cm}
\noindent
The set of all the rankings over $C$ is denoted by $\Pi_C$. A ranking $\prec$ over $C$ can be also identified to the totally ordered set $(C,\prec)$ itself. While a totally preordered set desxcribes a ranking where ex aequo are possible, that is $x\prec y$ and $y\prec x$ with $x\not =y$ is allowed. In the following, we consider ranking {\em without} ties.

By definition, a ranking $\prec$ on a finite set $C$ can be equivalently expressed by a bijective function $\pi:C\rightarrow \{1,...,n\}$, where $|C|=n$, such that $c_i\prec c_j$ if and only if $\pi(c_i)\leq \pi(c_j)$. In case of a ranking with ties, the associated function is not injective.  Let us denote the ranking univocally represented by the function $\pi$ as $\prec_\pi$.

The Kendall tau distance quantifies the total pairwise disagreement between rankings defined over the same set.

\vspace{0.5cm}

\begin{definition}\label{def:KT}
Let $\pi_1$ and $\pi_2$ two rankings over the finite set $C=\{c_1,...,c_n\}$, the {\bf Kendall tau (KT) distance} between $\pi_1$ and $\pi_2$ is defined by 
\begin{equation}
    KT(\pi_1,\pi_2) := \sum_{i<j} \overline{K}_{i,j}(\pi_1,\pi_2)
\end{equation}
where
\begin{equation}
    \overline{K}_{i,j}(\pi_1,\pi_2):=\left\{\begin{array}{cc}
         0& \mbox{\em if } (c_i\prec_{\pi_1} c_j \wedge c_i\prec_{\pi_2} c_j)\vee (c_j\prec_{\pi_1} c_i \wedge c_j\prec_{\pi_2} c_i)  \\
         & \\
         1& \mbox{\em otherwise}\hspace{6cm}
    \end{array}\right.
\end{equation}
in other words, the term \(\overline{K}_{i,j}(\pi_1,\pi_2)\) is 0 when elements $c_i$ and $c_j$ maintain the same relative order in both rankings, and 1 otherwise.
\end{definition}

\vspace{0.5cm}

\noindent
In definition \ref{def:KT}, the elements of $C$ are arbitrarily labelled to provide the analytic definition of the KT distance, let us remark that these labels do note relate to any ranking defined over $C$ a priori. 
The goal of the Kemeny ranking method is finding an optimal ranking which aggregates the rankings from a set \(\Pi\subset\Pi_C\) as a minimizer of the cumulative KT distance to all other rankings in the set.

\vspace{0.5cm}

\begin{definition}\label{def:opt}
Let \(\Pi\) be a set of rankings over a finite set $C$, a {\bf Kemeny ranking} for $\Pi$ is $\pi_o\in\Pi_C$ such that:
    \begin{equation}\label{eq:opt}
    \pi_o := \argmin_{\pi\in\Pi_C} \left( \sum_{\widetilde\pi\in\Pi} KT(\widetilde\pi, \pi) \right).
\end{equation}
\end{definition}

\vspace{0.5cm}

Let us remark that the optimization in (\ref{eq:opt}) is over $\Pi_C$ then, in general, $\pi_o$ does not belong to $\Pi$.
The Kemeny ranking is highly valued for its theoretical properties, particularly its ability to minimize disagreements and ensure that the consensus ranking closely aligns with individual rankings. 
    It is axiomatically attractive, as it satisfies several desirable properties such as neutrality, consistency, and the {\em Condorcet criterion}, which means it will rank the alternative that wins all pairwise comparisons (known as the Condorcet winner) in the top position \cite{Gpu}.
    Kemeny ranking has an interpretation as a maximum likelihood estimator, making it well-suited to epistemic social choice, which assumes that there is a "ground truth" or correct ranking that the method is trying to approximate \cite{rausr}.
    The method is justified in contexts where there is an assumption of a correct order that the rankings are trying to estimate, such as in expert rankings or when aggregating individual judgments to arrive at a collective decision \cite{rausr}.
However, Kemeny ranking is computationally complex (NP-hard), making it challenging to compute for a large number of candidates. It is also coNP-complete \cite{conp} to verify if a given ranking is a Kemeny ranking.
In summary, Kemeny ranking is an effective tool for deriving consensus rankings by minimizing pairwise disagreements, using the Kendall tau distance to measure dissimilarities. It is applicable in voting systems and decision-making processes, providing an efficient method for preference aggregation that respects the majority consensus.

\subsection{Partial and weighted lists}

In the context of Kemeny ranking and decision theory, a partial list is a ranking that represents preferences over a subset of the available alternatives. In other words, instead of ordering all the alternatives under consideration, a partial list orders only some of them, capturing individual preferences for those specific alternatives. In particular, the $k$-top list is a partial list formed by the first $k$ alternatives referring to a given ranking. 

\vspace{0.3cm}

\begin{definition}
    A {\bf partial list} over a finite set $C$ is any ranking over a subset $D\subset C$. Given a ranking $\prec$ over $C$ and $k\leq |C|$, the {\bf $k$-top list} w.r.t. $\prec$ is the partial list $(D,\prec)$ where $D\subset C$, with $|D|=k$, and $d\prec c$ for any $c\in C\setminus D$.
\end{definition}

\vspace{0.3cm}

\noindent
Another relevant notion is that of weighted list that is a ranking in which pairwise comparisons between alternatives are associated with numerical weights. In practical situations, these weights reflect the relative importance of preferences expressed by an individual or a group of decision makers.

\vspace{0.3cm}

\begin{definition}
    A {\bf weighted list} over $C$ is a pair $(\prec, \omega)$ where $\prec$ is a ranking and $\omega:C\times C\rightarrow \mathbb R^+$ is a weight function.
\end{definition}

\vspace{0.3cm}

\noindent
In the definition above, the value $\omega(c_i,c_j)$ is the weight assigned to the preference $c_i\prec c_j$.
The weighted lists are used to refine the process of aggregation of preferences in the Kemeny ranking, allowing a more detailed and precise representation of the priorities of the decision makers. This method allows you to significantly influence the final consensus order, more accurately reflecting collective preferences.

\section{Representation of the Kemeny aggregation into a quantum annealer}
In this section, we present the QUBO formulation for our Kemeny ranking approach and its representation into a quantum annealer. Rather than having the annealer directly output the optimal list, we choose to output the binary configuration which represents the optimal list in terms of a strictly upper triangular binary matrix, which is then used to derive the optimal list through classical computation.

\subsection{QUBO formulation of the Kemeny ranking problem}\label{sec:QUBO}

Suppose we start from a set of rankings $\Pi$ over a fnite set $C$;
the first step is to construct the {\em pairwise comparison matrix} $W=\{w_{ij}\}$ which indicates how many times a candidate is preferred over any other in the dataset:
\begin{equation}\label{eq:w}
    w_{ij}:=\sum_{\pi\in\Pi}(c_i\prec_{\pi} c_j),
\end{equation}
where $(c_i\prec_{\pi} c_j)=1$ if $c_i$ is preferred to $c_j$ in the list $\pi$ (i.e. $c_j\prec_{\pi} c_i$ is true), and $(c_i\prec_{\pi} c_j)=0$ otherwise. In other words, $w_{ij}$ is the number of rankings in $\Pi$ where $c_i$ is preferred to $c_j$.
Next, we define the following cost function to minimize:
\begin{equation}\label{eq:C_1}
    \mathcal C(X):=\sum_{i<j}^n(w_{ji}-w_{ij})x_{ij}=\sum_{i<j}^nb_{ij}x_{ij},
\end{equation}
where $b_{ij}=(w_{ji}-w_{ij})$ are biases forming the {\em bias matrix} $B=\{b_{ij}\}$ and the binary variables $\{x_{ij}\}_{i<j}$ defines a $n\times n$ binary matrix $X$ that is {\em strictly upper triangular}, i.e. $x_{ij}=0$ for any $i\geq j$. Let us denote the set of $n\times n$ strictly upper triangular binary matrices by $\mathcal M_{n,\bB}$.

\vspace{0.3cm}

\begin{definition}\label{def:X}
Given a ranking $\pi$ over a set $C=\{c_1,...,c_n\}$ we say that $X\in\mathcal M_{n,\bB}$ {\bf represents} $\pi$ if $x_{ij}=1$ when $c_i\prec_\pi c_j$ and $x_{ij}=0$ otherwise for any $i< j$. 
\end{definition}

\vspace{0.3cm}

\noindent
Conversely, any $X\in\mathcal M_{n,\bB}$ induces a strongly connected relation $\prec$ on a set $C=\{c_1,...,c_n\}$, defined by:
\begin{equation}\label{eq:X}
\mbox{$c_i\prec c_j$ if $x_{ij}=1\quad$ and $\quad c_j\prec c_i$ if $x_{ij}=0,\quad$ for $i< j$.}    
\end{equation}
In general, we have the following elementary technical result providing a necessary and sufficient condition such that $X\in\mathcal M_{n,\bB}$ represents a ranking.
\vspace{0.3cm}

\begin{proposition}\label{prop:X}
Let be a finite set $C$, with $|C|=n$ and $\prec$ be the strongly connected relation on $C$ represented by $X\in\mathcal M_{n,\bB}$ as in (\ref{eq:X}). Then, $\prec$ is a ranking if and only if $x_{ij}=x_{jk}=1$ implies $x_{ik}=1$ for any $i<j<k$.
\end{proposition}

\begin{proof}

$\prec$ is reflexive since we have $x_{ii}=0$ that means $c_i\prec c_i$ for any $i=1,...,n$, observe that also the property $x_{ii}=1$, for any $i$, would imply the reflexivity as well but it is excluded since $X$ has zero diagonal. 
If $i\not =j$ we have that $c_i\prec c_j$ ($x_{ij}=1$) and $c_j\prec c_i$ ($x_{ij}=0$) are mutually exclusive then $\prec$ is antisymmetric.
The property $x_{ij}=x_{jk}=1$ implies $x_{ik}=1$ is equivalent, by (\ref{eq:X}), to $c_i\prec c_j $ and $c_j\prec c_k$ imply $c_i\prec c_k$ that is transitivity.  
\end{proof}

\vspace{0.2cm}

\begin{definition}\label{def:cycles}
     If $X\in\mathcal M_{n,\bB}$ represents a non-transitive relation then the triple $(i,j,k)$ such that $x_{ij}=x_{jk}=1$ and $x_{ik}=0$ is called {\bf cycle} of $X$. If $X$ does not present cycles it is said to be {\bf acyclic}.
\end{definition}

\vspace{0.3cm}

\noindent
In other words, if $X\in\mathcal M_{n,\bB}$ admits cycles then it does not represent a ranking over a set of cardinality $n$, since the corresponding total relation fails transitivity. In fact, if $(i,j,k)$ is a cycle of $X$ then the corresponding total relation satisfies $c_i\prec c_j\prec c_k\prec c_i$. Thus, the set 
\begin{equation}
    \mathcal M_{n,\bB}^a:=\{X\in\mathcal M_{n,\bB}\, :\, \mbox{$X$ is acyclic}\},
\end{equation}
is in bijective correspondence to the set of rankings over a set of cardinality $n$.

Now, let us clarify the meaning of the cost function (\ref{eq:C_1}), the values assumed by the biases $b_{ij}$ entails three main possibilities:
\begin{enumerate}
    \item $b_{ij}>0$: implies that $c_j$ is preferred to $c_i$ most of the times.
    Minimization requires $x_{ij}=0$; 
    \item $b_{ij}<0$: implies that $c_i$ is preferred to $c_j$ most of the times.
    Minimization requires $x_{ij}=1$;
    \item $b_{ij}=0$: indicates the absence of a clear preference. For minimization, the value of $x_{ij}$ is irrelevant and will be randomly assigned. This situation can occur when we consider a data set with an even number of votes, or when we deal with partial lists.
\end{enumerate}

More precisely, we have the following result relating an optimal ranking in the sense of definition \ref{def:opt} to a minimum of the cost function (\ref{eq:C_1}).

\vspace{0.3cm}

\begin{theorem}\label{thm:Xo}
    Let $\Pi$ be a set of rankings over the finite set $C$ with $|C|=n$, and $\pi_o\in\Pi_C$ be the ranking represented by the acyclic matrix $X_o\in\mathcal M_{n,\bB}^a$. Then, $\pi_o$ is a Kemeny ranking for $\Pi$ if and only if:
    \begin{equation}\label{eq:Xo}
        X_o=\argmin_{X\in\mathcal M^a_{n,\bB}} \mathcal C(X).
    \end{equation}
    where $\mathcal C:\mathcal M^a_{n,\bB}\rightarrow\mathbb R$ is the cost function defined in (\ref{eq:C_1}).
\end{theorem}
\begin{proof}
    Let us consider the cost function of the optimization problem defined by (\ref{eq:opt}):
    \begin{equation}\label{eq:costKT}
        F_\Pi(\pi):=\sum_{\widetilde\pi\in\Pi} KT(\widetilde\pi,\pi)=\sum_{\widetilde\pi\in\Pi} \sum_{i<j}^n \overline{K}_{ij}(\widetilde\pi,\pi)=\sum_{i<j}^n\sum_{\widetilde\pi\in\Pi}\overline{K}_{ij}(\widetilde\pi,\pi).
    \end{equation}
Taking an arbitrary ranking $\pi\in\Pi_C$ and assuming that $c_i\prec_\pi c_j$ for some $i,j\in\{1,...,n\}$, by definition of Kendall tau distance the term $\sum_{\widetilde\pi\in\Pi}\overline{K}_{ij}(\widetilde\pi,\pi)$ is the number of rankings in $\Pi$ such that $c_j$ is preferred to $c_i$, i.e. $c_i$ and $c_j$ are not ordered like in $\pi$.\\
Now, consider the optimization of $\mathcal C$ which is equivalent to the minimization of:
    \begin{equation}\label{eq:C2}
        \mathcal C_2(X)= \sum_{i<j}^n(b_{ij}x_{ij}+w_{ij}),
    \end{equation}
because of the mere addition of the constant $\sum_{i<j} w_{ij}$ to $\mathcal C$. Taking $X$ which represents $\pi$, we have $x_{ij}=1$ since $c_1\prec_\pi c_j$, therefore $b_{ij}x_{ij}+w_{ij}=w_{ji}$ that is the number of rankings in $\Pi$ such that $c_j$ is preferred to $c_i$ by (\ref{eq:w}). Therefore:
\begin{equation}
    \sum_{\widetilde\pi\in\Pi}\overline{K}_{ij}(\widetilde\pi,\pi)=b_{ij}x_{ij}+w_{ij}\qquad \forall i,j\in\{1,...,n\},
\end{equation}
then $F_\Pi(\pi)=\mathcal C_2(X)$. In view of the arbitrary choice of $\pi$ in $\Pi_C$, we can conclude that $\pi_o$ is an optimum of the cost function $F_\Pi$ defined in  (\ref{eq:costKT}) if and only if its representing matrix $X_o$ optimizes $\mathcal C_2$ and, equivalently, $\mathcal C$.
\end{proof}

\noindent
By theorem \ref{thm:Xo}, the problem of finding a Kemeny ranking for a set of lists $\Pi$ consists in solving the constrained optimization problem (\ref{eq:Xo}), where the constraint is posed by the requirement that the matrices are acyclic, whose  unconstrained version is the optimization over $\mathcal M_{n,\bB}$ that is trivial since any bias $b_{ij}$ forces the binary variable $x_{ij}$ to be either 0 or 1. In order to apply the constraint of absence of cycles, we incorporate in the cost function of the optimization over $\mathcal M_{n,\bB}$ a penalty term for any cycle:
\begin{equation}
    \mathcal P_{(i,j,k)}(X):=P_{ijk}\underbrace{(x_{ik}+x_{ij}x_{jk}-x_{ij}x_{ik}-x_{jk}x_{ik})}_{\text{$c_{ijk}$}},
\end{equation}
where $P_{ijk}\in\mathbb R^+$ is the penalty coefficient specific for the cycle $(i,j,k)$ and, by definition \ref{def:cycles}, we have $c_{ijk}=1$ if $(i,j,k)$ is a cycle, $c_{ijk}=0$ otherwise.
In practice, the values of $P_{ijk}$ must be sufficiently large to ensure that the penalty term effectively prioritizes cycle removal over the gain given by the linear term. For preventing any potential cycle that may arise, a penalty term is added for each possible cycle resulting in the following contribution to the cost function:
\begin{equation}
    \mathcal C_c(X)=\sum_{i<j}^n\sum_{k<j}^nP_{ijk}(x_{ik}+x_{ij}x_{jk}-x_{ij}x_{ik}-x_{jk}x_{ik}),
\end{equation}
which vanishes if and only if $X$ is acyclic. Then, the total cost function is:
\begin{equation}\label{eq:Ctot}
    \mathcal C_{tot}(X)=\sum_{i<j}^nb_{ij}x_{ij}+\sum_{i<j}^n\sum_{k<j}^nP_{ijk}(x_{ik}+x_{ij}x_{jk}-x_{ij}x_{ik}-x_{jk}x_{ik}).
\end{equation}
The binary variables $x_{ij}$, corresponding to the elements of $X$ above the diagonal, can be rearranged into a binary vector $(z_1,...,z_{\frac{n(n-1)}{2}})$ 
obtaining the QUBO formulation of the optimization of $\mathcal C_{tot}$ over $\mathcal M_{n,\bB}$ which corresponds to the constrained optimization of $\mathcal C$ over $\mathcal M_{n,\bB}^a$ that is equivalent to the problem of finding a Kemeny ranking by minimizing the cumulative Kendall tau distance by \autoref{thm:Xo}.

\subsection{Choice of $P_{ijk}$}
The choice of $P_{ijk}$ is crucial in penalizing the occurrence of the cycle $(i,j,k)$. For simplicity, let us assume that $P_{ijk}$ is constant across all cycles and denote it as $P$. Therefore, we need to determine an appropriate value for $P$ that strikes a balance between the effectiveness of cycle penalization and the performance of the annealer.
Ideally, we want $P$ to be as small as possible since larger values can negatively impact the annealer's efficiency. Initially, one might think to set the strict requirement that $P$ must be greater than the number of votes in the dataset ($P > |\Pi|$) since it is the maximum value a bias can have.
Instead, we believe that by carefully analyzing the cycles and exploring how they can be concatenated, it is possible to decrease this value.
For instance, by considering a single cycle featuring an odd number of votes, the request can be diminished to $P > |\Pi|-2$ (the proof can be found in Appendix \ref{secA1}). This relaxation of the constraint is feasible due to the fact that if one of the biases happens to have a value equal to $|\Pi|$, it becomes impossible for the remaining two to take values that would allow the triplet to form a cycle.
However, this reduction is not as straightforward when dealing with a dataset containing an even number of votes.\\
While $P > |\Pi|-2$ guarantees cycle removal, it's not always strictly needed. So, while the requirement ensures certainty, many situations allow flexibility without compromising cycle elimination.

In fact, to further optimize and potentially reduce the value of $P$, we can perform additional calculations to determine the maximum possible value of the bias coefficients. Then, we will select the penalty coefficient as the smaller value between this maximum and the theoretical required value. This approach will ensure that the value of $P$ is as small as possible while still being sufficiently high to remove all cycles:
\begin{align}
P_{even} &> \text{min}(\text{max}(|b_{ij}|),|\Pi|)\\
P_{odd} &> \text{min}(\text{max}(|b_{ij}|),|\Pi|-2) \label{minmax}
\end{align}

\noindent
where $P_{even}$ stands for the penalty coefficient for datasets with an even number of votes, while $P_{odd}$ represents the penalty coefficient for datasets with an odd number of votes.
This additional process carries a time complexity of $\mathcal{O}(n^2)$, with $n$ being the number of candidates.
We believe there is potential to further improve these results and reduce $P$ even more.
However, these efforts will become unnecessary with the introduction of one of the two improvements described below. Specifically, the \textsf{iterative method} outlined in Section \ref{im}, which involves assigning a custom value to $P_{ijk}$ for each cycle instead of treating it as a constant, addresses this issue.

\subsection{Hybrid quantum-classical model for ranking aggregation}
In the following, we call {\em base model} (BM) the overall quantum-classical procedure given by the following steps: 
\begin{enumerate}
    \item initialization of the cost function (\ref{eq:C_1});
    \item formulation of the QUBO problem (\ref{eq:Ctot});
    \item quantum annealing to solve the QUBO problem;
    \item reconstruction of the Kemeny ranking.  
\end{enumerate}
Moreover, we should consider also the embedding procedure of the considered QUBO into the sparse annealer architecture. However, let us assume for now that the quantum architecture is connected enough so that the QUBO is directly representable.  
The process of list reconstruction, which involves identifying the ranking represented by the binary matrix $X\in\mathcal M_{n,\bB}^a$, is carried out classically using a straightforward method. 

\vspace{0.3cm}
\begin{definition}
Let $\pi$ be a ranking over a finite set $C=\{c_1,...,c_n\}$, and $X\in\mathcal M_{n,\bB}^a$ be the matrix representing $\pi$. The {\bf score} of $c_i\in C$ with respect to $\pi$ is the integer number:
\begin{equation}\label{listrec}
    V_{c_i}:=\sum_{i<j}x_{ij}+\sum_{k<i}(1-x_{ki}).
\end{equation}

\end{definition}

\vspace{0.3cm}

In practice, for each candidate, we calculate a score by counting the number of elements with a value of 1 in the row of $X$, and the number of elements with a value of 0 in the column corresponding to the candidate $c_i$. 
The sum of these counts provides the score $V_{c_i}$ as the number of candidates such that $c_i$ is preferred over. 
In other words, the score of $c_i$ can be alternatively defined as the cardinality:
\begin{equation}
    V_{c_i}=\left|\left\{c\in \pi\, :\, c_i\prec c  \right\}\right|.
\end{equation}
After calculating the scores for each candidate, we arrange them in decreasing order based on these scores. Ideally, each candidate score should be distinct; however, in cases of approximate solutions where not all cycles are removed, it is conceivable to have two or more candidates sharing the same score.

\begin{figure}[!h]
    \centering
    \includegraphics[scale=0.5]{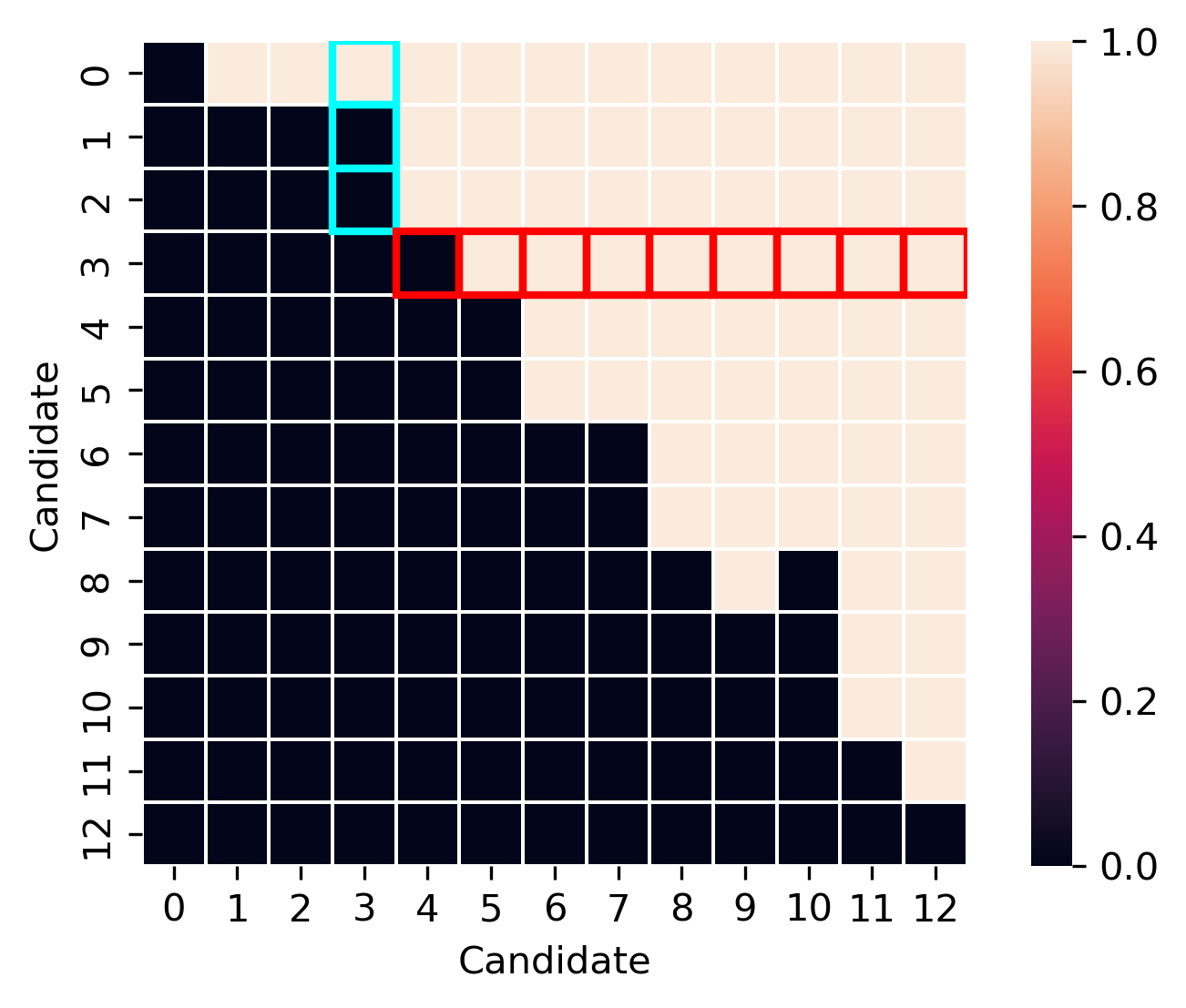}
    \caption{Graphical representation of the matrix elements contributing to the $V_{c_3}$ score. The matrix elements from the first sum are highlighted in red, while those from the second sum are shown in cyan.}
    \label{rossoeciano}
\end{figure}
To address this scenario, rather than considering all possible degeneracies and calculating the best list classically, we take a more conservative approach: when faced with candidates having identical scores, we randomly choose the order for those candidates. This ensures that we do not overly boost the performance of the model.

\vspace{0.3cm}

\begin{proposition}\label{prop:complexity}
Let $\Pi$ be a set of rankings over a finite set $C$ with $|C|=n$. The classical part of BM that returns a Kemeny ranking for $\Pi$ operates in time $\mathcal{O}\left(|\Pi|n^2\right)$.
\end{proposition}

\begin{proof}
    The construction of the comparison matrix as defined in (\ref{eq:w}) requires all the possible comparisons among the elements of $C$ for any ranking in $\Pi$ that is: 
    \begin{equation}
        N_W=|\Pi|\binom{n}{2}=|\Pi|\frac{n(n-1)}{2}.        
    \end{equation}
    Once initialized the quantum annealer with the bias matrix $b_{ij}=w_{ij}-w_{ji}$, it returns a global optimum of (\ref{eq:Ctot}), in time $t_{ann}$, that corresponds to a matrix $X_o\in\mathcal M^a_{n,\bB}$ representing a Kemeny ranking for $\Pi$ by Theorem \ref{thm:Xo}. In order to reconstruct the ranking, the computation of the score $V_{c_i}$ requires $N_V=\mathcal O(n)$ operations for any $i=1,...,n$, then all the scores are obtained in $\mathcal O(n^2)$ operations. Finally, the ordering of the scores in decreasing order require $N_o=\mathcal O(n^2)$ operations by brute force or $\mathcal O(n\log n)$ as a limit complexity. Therefore: $N_W+N_V+N_o=\mathcal O(|\Pi|n^2)$.      
\end{proof}

\vspace{0.3cm}

\noindent
In Proposition \ref{prop:complexity}, we assumed that the annealer architecture is connected enough to allow the direct representation of the QUBO problem without any embedding procedure. The annealing time $t_{ann}$ could be estimated by the adiabatic theorem under the assumption that the quantum annealer operates adiabatically, however in the real quantum machine the cooling is characterized by decoherence and energy dissipation. In practice, $t_{ann}$ is a user-specified parameter and the returned ground state is approximated, in this sense proposition \ref{prop:complexity} provides the effective time complexity of the base model as implemented in the present work up to the time complexity of the embedding which depends on a built-in function.


By construction, the model is able to accommodate partial lists and weighted lists.
Regarding partial lists, we focus on $k$-top lists, wherein the candidates present in the list are given a definite preference over those that are missing. The method slightly differs from the complete list procedure when calculating the elements of the pairwise comparison matrix $W$, presenting three potential scenarios: both candidates are present, one of the two candidate is missing, and both candidates are missing. In the first two cases, we can still determine an order, while in the third case, no order can be established, resulting in no contribution being made to the matrix elements. Once this step is completed, the procedure remains unchanged.

In the case of weighted lists, we introduce the ability to assign different weights to entire lists, enabling prioritization of one list over another. This is reflected in the contribution to the $w_{ij}$ matrix, which is also weighted accordingly. Furthermore, we can assign individual weights to each position within a list, providing even greater control over the order and significance of the candidates.
The model adeptly handles partial lists, ensuring order despite missing candidates.


\subsection{Enhancements}

In this section, we explore two approaches to optimize the performance of the proposed model. The first method, which we refer to as the \textsf{iterative method}, addresses the issue of penalizing an excessive number of cycles. The maximum number of cycles subject to penalties, denoted as $N_{max}$, is calculated using the formula $N_{max} = \frac{n(n-1)(n-2)}{3!}$, where $n$ represents the number of candidates in the dataset. This formula comes from combinatorics, where $n(n-1)(n-2)$ counts the number of ways to choose three distinct candidates from a total of $n$, and the division by $3!$ accounts for the fact that the order of selection does not matter when forming a cycle.
Despite this theoretical upper bound, the actual number of cycles observed during the solution process is typically much lower. This suggests that some of the penalization terms, which currently constrain the model's performance, may not be necessary for achieving an optimal solution. Thanks to our method of representing the list as a binary matrix, we can identify the cycles present in the annealer's output, offering an opportunity for further enhancement.
Thanks to the representation of the list as a binary matrix, we can identify and penalize cycles in the annealer's output, allowing for iterative refinement until convergence.
The second method, that we called ``pair removal", involves strategically excluding specific pairs from the problem formulation, with the understanding that their values can be inferred at the end of the annealing process.
By carefully selecting and removing certain pairs, we can reduce the overall complexity of the problem and potentially enhance the performance.

\subsubsection{Iterative method}\label{im}
This procedure is designed to precisely target and penalize only those cycles directly involved in the problem, applying the least possible weight to eliminate a cycle without compromising the discovery of the optimal solution. A cycle $(i,j,k)$ is directly involved in the problem when, over the considered set of rankings, $c_i$ is preferred to $c_j$ most of the times ($b_{ij}<0$), $c_j$ is preferred to $c_k$ most of the times ($b_{jk}<0$) and $c_k$ is preferred to $c_i$ most of the times ($b_{ik}>0$). An elementary example is given by the dataset $\Pi_C=\{\pi_1,\pi_2,\pi_3\}$ over $C=\{c_1,c_2,c_3\}$ such that: 
$$c_1\prec_{\pi_1} c_2\prec_{\pi_1} c_3\quad ,\quad c_2\prec_{\pi_2} c_3\prec_{\pi_2} c_1\quad ,\quad c_3\prec_{\pi_3} c_1\prec_{\pi_2} c_2,$$
in this case the dataset encodes the cycle $(1,2,3)$.
Through the selection of the cycles directly engaged in the problem and the assignment of penalties, we aim to keep the problem as simple as possible for the annealer and enhance its efficiency. In fact, the magnitude of the penalty coefficients plays a crucial role in influencing the efficiency of the annealing process.
To achieve this precise targeting, several runs of the annealer are often required. This iterative approach is essential for identifying all the relevant cycles in the datataset and for determining an appropriate value for the penalty coefficients.
To begin the procedure, we construct the matrix $\Omega$ by applying the following transformation to the bias matrix $B$ introduced in (\ref{eq:C_1}):
\begin{equation}\label{omega}
    \omega_{ij}:=\Theta(-b_{ij})
\end{equation}
where $\Theta$ is the Heaviside step function, which is defined as $\Theta(x)=0$ for $x<0$, $\Theta(x)=1$ for $x>0$, and $\Theta(0)=1/2$. Therefore,  $\omega_{ij}=1$ when $c_i$ is preferred to $c_j$ the majority of times over the dataset.
A sufficient condition for the presence of cycles is:  
\begin{equation}\label{omega1}
    (\omega_{ij}=0 \land \omega_{ik}=1 \land \omega_{jk}=0)  \lor (\omega_{ij}=1 \land \omega_{ik}=0 \land \omega_{jk}=1),
\end{equation}
where $\wedge$ and $\vee$ denote the logical conjunction and disjunction respectively. 
Additionally, for datasets with an even number of votes, we could also consider cycles that might emerge due to the random values assigned to the qubits with no bias; in such cases, the cycle detection criteria changes to:
\begin{equation}\label{omega2}
(\omega_{ij} \neq 1 \land \omega_{ik} \neq 0 \land \omega_{jk} \neq 1) \lor (\omega_{ij} \neq 0 \land \omega_{ik} \neq 1 \land \omega_{jk} \neq 0)
\end{equation}
Here, the rationale behind setting $\Theta(0)=1/2$ becomes clear.
By incorporating this change, we may end up considering a higher number of cycles than necessary. However, the advantage is that it can potentially decrease the total number of runs required to achieve the optimal output.
We associate a penalty coefficient $P_{ijk}$ with each cycle $(i,j,k)$, starting from the smallest weight that could potentially remove the cycle ($0+\epsilon$ for datasets with an even number of votes or $1+\epsilon$ for datasets with an odd number of votes, where $\epsilon$ is a small value to ensure the penalty coefficient is higher than the bias) and it may be adjusted upwards during subsequent iterations as needed.
Upon completing the annealing process, we analyze the output solution represented by the binary matrix $X$.  Specifically, we search again for cycles, and this search leads to one of the following scenarios:

\begin{enumerate}
\item \textbf{Absence of Cycles}: The absence of cycles in the matrix $X$ indicates that all necessary cycles have been effectively penalized. In this case, we consider the solution to be successfully identified.
\item \textbf{Presence of Cycles}: When cycles are detected within the output matrix, a tailored response is necessary based on the nature of the cycles:
\begin{enumerate}
    \item \textbf{New Cycles}: These cycles have not appeared in earlier iterations of the process. For each new cycle detected, we introduce it into the penalization framework and assign the minimum penalty coefficient to it.
    \item \textbf{Previously Penalized Cycles}: If the identified cycles include some of those that have been penalized in previous iterations but have recurred, it indicates that the penalization was insufficient. For these cycles, we increase the penalty coefficient by a factor of 2. The rationale behind this increment is tied to the bias spectrum's step size of two, which originates from the constraint $w_{ij} + w_{ji} = n$.\\
    It is important to note that this principle does not apply to partial
    lists, since the constraint $w_{ij} + w_{ji} = n$ is no longer present, and therefore requires a different approach.
\end{enumerate}
\end{enumerate}
We iterate this procedure until we obtain an output without cycles, at which point we have converged to an optimal solution.

However, converging to the optimal solution becomes complex in challenging situations where the problem is difficult to solve, such as when there is a high number of candidates or cycles. 
In these situations, the annealer may fail to produce the optimal output due to qubits adopting incorrect values, leading to the appearance of irrelevant cycles. 
For the same reason we can obtain an output without cycles that does not represent the optimal solution, such as when a qubit, representing a pair of closely ranked elements in the optimal list, acquires the wrong value. Consequently, while an output without cycles is achieved, the two candidates end up being swapped, resulting in a suboptimal solution.
A significant benefit of the pairwise preference representation is its flexibility when applied with the \textsf{iterative method}. Unlike conventional approaches that require pursuing convergence, this method allows us to halt the process at any desired step and settle for an approximate solution. This flexibility is made possible by the list reconstruction step, which effectively handles scenarios involving reconstructing a list from a binary configuration with cycles.

Opting for this approximate strategy necessitates a knowledge that the number of iterations does not directly correlate with the quality of the approximate solution. In other words, the Kemeny ranking of the output list at the $m$-th step may be larger than that of the list obtained at the $(m-1)$-th step. This underscores the need for a nuanced understanding of the iterative process and its outcomes when pursuing approximate solutions.

In summary, the primary objective of the \textsf{iterative method} is to penalize only the cycles directly relevant to the problem using the smallest possible penalty terms. However, the implementation of this strategy often requires multiple runs of the annealer and classical computation to identify and update the cycles involved.
The classical computational demand of a single iteration is due to cycle detection which scales as $\mathcal{O}(n^3)$, where $n$ represents the total number of candidates considered, as provided by (\ref{omega1}) and (\ref{omega2}), since the complexity of executing the base model is quadratic in $n$. 

The iterative nature of this method involves starting from the $\Omega$ matrix and creating a set of pairs (cycle, penalty coefficient) = $((i,j,k),P_{ijk}) \in \{C_n\}$ at the beginning of each iteration. This set is continuously updated by including newly appeared cycles and increasing the value of $P_{ijk}$ as necessary. This iterative process ensures that the penalty coefficients are refined over successive iterations to focus on relevant cycles and optimize the overall solution.\\
Initially, $P_{ijk}$ is set to the minimum value required to eliminate a cycle: $\epsilon$ for datasets with an even number of votes and $1+\epsilon$ for those with an odd number. The coefficient $P_{ijk}$ may then be increased by 2 at each iteration as needed.\\
At iteration $n$, the cost function we aim to minimize can be expressed as:
\begin{equation}
    C(X)=\sum_{i<j}^nb_{ij}x_{ij}+\sum_{(i,j,k)\in \{C_n\}} P_{ijk}(x_{ik}+x_{ij}x_{jk}-x_{ij}x_{ik}-x_{jk}x_{ik}).
\end{equation}

In practice, a solution is considered achieved when a cycle-free output is obtained. However, with challenging datasets, cycle-free outputs may not necessarily represent optimal solutions due to inaccuracies in the annealer's output.
Note, however, that running the same iteration multiple times can produce outputs with different cycles. Consequently, the final set of cycles {$C_n$}, which when penalized allow us to find the optimal solution, may differ between runs of the iterative method, this happens because a cycle can have the same absolute value of a bias for two or more pairs present in the cycle. Furthermore, the number of iterations required to converge to the optimal solution can also vary from run to run. 
To effectively manage such challenges, employing measures such as running the annealer multiple times per iteration can be beneficial. Ultimately, the method can serve as an approximate solution approach by halting before eliminating all necessary cycles, thereby settling for a suboptimal solution.

\subsubsection{Pair removal}\label{pairremoval}
We have observed that the performance of the annealer is influenced by the number of qubits employed (see \Cref{pairremoval}). Therefore, our aim is to reduce the quantity of candidate pairs submitted to the annealer (each represented by a qubit) while preserving enough information to effectively solve the problem.\\
At the process's conclusion, our strategy involves reinstating the eliminated pairs by inferring their values through the transitive property. To determine which pairs to remove, we employed two methods: pair removal of high biases ($PRHB$) and pair removal based on the $\Omega$ matrix ($PR\Omega$).
The $PRHB$ technique focuses on the elimination of pairs characterized by high bias. This selection criterion is grounded in the assumption that pairs with high bias are likely to contain candidates that are significantly distant in the optimal rank. The primary advantage of employing the $PRHB$ method lies in its ability to reduce the overall size of the problem by excluding pairs with high bias, thus simplifying the problem. {Moreover, $PRHB$ allows the problem to be rescaled and potentially leading to better performance.}
However, a potential limitation of this approach is the inherent uncertainty regarding the actual spatial separation of the elements of the pair in the optimal solution. This uncertainty may pose challenges in accurately deducing the value of the elements, possibly affecting the overall quality of the solution. The search of the pairs with a bias above a threshold presents a complexity $O(n^2)$ then, after the annealing, the calculation of the contribution of the removed pairs is calculated in time $O(n)$, as a result $PRHB$ has a quadratic classical complexity.

In contrast, the $PR\Omega$ technique leverages the $\Omega$ matrix, which represents an approximate solution. By analyzing the list reconstructed with this matrix, it becomes possible to identify elements that are likely to be distant, w.r.t. an arbitrary scale over the lists, in the final solution. {The main benefit of the $PR\Omega$ method is its ability to increase confidence, particularly in the context of large lists, in removing pairs representing candidates that are far apart in the optimal list. This offers the significant advantage of having numerous pairs available in Eq. (\ref{eq:valueassignement}) that can collectively contribute to inferring the value of the missing pair.}
However, this method requires careful execution, as removing many pairs containing the same candidates could inadvertently lead to incorrect qubit value assignments. Such inaccuracies could propagate through the reconstruction process, leading to errors in the inferred values of subsequent pairs.
There are also certain constraints to be considered when removing qubit pairs. Specifically, we cannot eliminate pairs that are part of cycles, as this would result in the loss of qubit changes associated with the cycle's removal. The $PR\Omega$ method requires the construction of the approximated list from $\Omega$ in time $O(n^2)$, the recognition of the pair with elements far a part, with complexity $O(n^2)$ then, after the annealing, we need the values of the removed pairs that takes time $O(n)$. Therefore, the classical complexity of $PR\Omega$ is $O(n^2)$.

The procedure for pair removal is articulated as follows (see also Algorithm \autoref{alg:pair_removal} in \autoref{appC}):
\begin{enumerate}\setlength\itemsep{1em}

\item Find the set of pairs to remove, $R$, using one of the two techniques and remove those pairs from the cost function. The total cost function is then as follows:
\begin{equation}
    C_{pr}(X)=\sum_{i} \sum_{\substack{i < j \\ (i,j) \notin R}}
b_{ij}x_{ij}+\sum_{i} \sum_{\substack{i<j \\ (i,j) \notin R}} \sum_{\substack{j<k \\ (i,k) \notin R \\ (j,k) \notin R}}
P_{ijk}(x_{ik}+x_{ij}x_{jk}-x_{ij}x_{ik}-x_{jk}x_{ik})
\end{equation}
\item Run the quantum annealer on the modified cost function $C_{pr}$ and obtain the output
\item Reintegrate the removed pairs (qubits), inferring their value from the other pairs.
\end{enumerate}
The reconstruction process involves gathering information from the remaining qubits, leveraging the transitive property, and the process can be described as follows:
\begin{enumerate}
    \item Suppose the removed pair is $(a,b)$. To infer the pair value, we employ all the pairs that were not removed and that contain a candidate from our pair. Specifically, we consider all pairs of the form $(a,c)$ and $(b,c)$, where $c$ is another candidate in the dataset. For simplicity, we assume an ordered input $a \prec b \prec c$ (For cases where $a \prec c \prec b$ or $c \prec a \prec b$, refer to \cref{tab:AXB} in Appendix \ref{tt}).
    We then examine the states of the output qubits: $x_{ac}$ and $x_{bc}$. 
    There are four possible scenarios:\\
    \begin{enumerate}\setlength\itemsep{1em}
        \item ($x_{ac},x_{bc}$)=(0,0): $c\prec a,c\prec B$ $\Rightarrow$ no information about the order
        \item ($x_{ac},x_{bc}$)=(0,1): $c\prec a,b\prec c$ $\Rightarrow$ $b\prec a$
        \item ($x_{ac},x_{bc}$)=(1,0): $a\prec c,c\prec b$ $\Rightarrow$ $a\prec b$
        \item ($x_{c},x_{bc}$)=(1,1): $a\prec c,b\prec c$ $\Rightarrow$ no information about the order\\
    \end{enumerate}
Note that if either $(a,c)$ or $(b,c)$ is removed, no inference can be made: both pairs are required to add a contribution to the majority vote.
Based on these cases, we assign a value to the variable $R_{c}$:
\begin{itemize}\setlength\itemsep{1em}
    \item $R_{c}=0$ for (a) and (d)
    \item $R_{c}=-1$ for (b)
    \item $R_{c}=1$ for (c)\\
\end{itemize}
\item We assign a value to the pair:
\begin{equation}\label{eq:valueassignement}
    x_{ab}=\Theta\bigg(\sum_{c\notin \{a,b\}}R_c\bigg),
\end{equation}
where $\Theta(x)$ is the Heaviside step function with $\Theta(0)=0.5$.
This specific value of 0.5 indicates an undecided state for the qubit, suggesting that the majority vote resulted in a draw, and therefore, no preference can be determined yet.\\

\item Repeat points 1 and 2 for all missing pairs\\

{\item Repeat the process by incorporating the reconstructed pairs until either all removed pairs have been reconstructed or stalling is detected (i.e., when it is no longer possible to assign a value to any pair in the last step).\\
}

\item If the process fails to reconstruct all pairs, reintroduce those pairs (remove from $R$) and start the procedure again.\\
To achieve full pair reconstruction, the procedure may need to be executed multiple times.
\end{enumerate}
As discussed above, the pair removal adds a computational cost of $O(n^2)$
in exchange of a modest performance enhancement. Therefore, there may be a temptation to improve the performance further by removing more pairs. However, it is crucial to strike a balance in the number of pairs removed: excessive removal can lead to incorrect value inference and consequently compromise list reconstruction.
Additionally, this method should be complemented with techniques that can identify the cycles involved, such as the iterative approach. Although the benefits may not always be substantial, this method can help in complex scenarios, providing valuable assistance.
The key strength of the model lies in its ability to extract valuable information from its binary representation, such as cycle identification, while offering the flexibility to adjust the matrix according to specific requirements.
The following are the model's key capabilities:
\begin{enumerate}
\item \textbf{Information Extraction:} The binary representation of the model provides a comprehensive account of the current sequence of all pairs, thereby enriching the available information.
\item \textbf{Error Identification:} The model allows for the exact identification of errors within an output, enabling targeted penalization of only the cycles present in that particular configuration.
\item \textbf{Approximate Solutions:} The model supports approximate solutions, allowing for the removal of some cycles without requiring the complete elimination of all cycles. This settles for an approximate solution, which eventually can be refined later.
\item \label{refinement}\textbf{Refinement of Approximate Solutions:} The model can be adjusted to refine and enhance the quality of an approximate solution, such as one obtained via KwikSort \cite{KwikSort}. By reducing the number of disagreeing pairs between the annealer output without cycle penalization (\cref{omega}) and the binary representation of the approximate solution (changing the signs of certain elements in the bias matrix), the model aims to solve a simplified version of the minimization problem. This results in a solution that is at least as good as, if not better than, the approximate solution.

\end{enumerate}
In summary, the binary representation not only allows the extraction of critical information such as cycle identification, but also allows the bias matrix to be adjusted to meet specific needs, such as refining an approximate solution. These features make the model a versatile and invaluable tool in different scenarios.

\subsubsection{Challenges and mitigation strategies}\label{section:mitigation}
The process presents several challenges:
\begin{enumerate}
    \item Inability to find an embedding: due to the extensive penalty on cycles, it may be difficult to find an embedding.
    \item Lack of convergence: the process might fail to produce an output matrix that is completely free of cycles.
    \item Convergence to a suboptimal solution: we may find an output matrix without cycles that is not the optimal solution.
\end{enumerate} 
To address these challenges, we can employ the following strategies: 
\begin{enumerate}
    \item \label{embeddingdiscovery} \textbf{Facilitating embedding discovery:} To aid the discovery of an embedding, we can reduce the number of cycles we penalize. One effective method involves eliminating a cycle from the penalty list if all the pairs it comprises are already accounted for in at least $k$ other cycles. This strategy reduces the cycle count for embedding purposes while preserving crucial cycles. The decision-making process involves finding a balance between reducing the number of cycles and retaining the key pairs that drive cycle formation.
    \item \textbf{Addressing non convergence and suboptimal solutions:} These issues often arise when the problem is too complex. To tackle this problem, we can ``double-check" multiple times the cycles we penalize. By conducting multiple runs for each stage and considering only the cycles that consistently appear in all runs, we can filter out random cycles and retain only the relevant ones. This approach helps ensure that the cycles we penalize are those that consistently appear and are not merely the result of random fluctuations.
    \end{enumerate}

Overcoming challenges such as embedding discovery, non convergence, and suboptimal solutions involves strategic adjustments. By refining cycle reduction methods and validating penalized cycles rigorously, we aim to enhance efficiency and achieve more effective outcomes.

\section{Kemeny ranking in multiagent systems}

In multiagent systems, the choice of the optimal aggregation method is important in the context of collective decision-making, where agents must agree on a common decision despite potentially conflicting preferences. The exploration of voting theory and how it applies to multiagent systems has been extensive \cite{pitt2006voting, endriss2014social, dodevska2019computational, ultimomas}, providing foundational insights into achieving fair and efficient decision-making under these complex scenarios. Among different ranking strategies, Kemeny ranking is distinguished by its fairness and ability to minimize total disagreement. Other list aggregation techniques, like Borda count and plurality voting, tend to have biases that might skew the ultimate decision. For instance, Borda count can disproportionately emphasize mid-ranked preferences, whereas plurality voting may disregard the preferences of agents whose top choice differs from the majority. By assuring that the collective ranking is the nearest to all individual rankings when using pairwise comparisons, Kemeny ranking provides a more equitable method for merging diverse inputs.
Moreover, compiling a list of potential decisions strengthens the resilience of multi-agent systems. Should a primary decision prove unsuccessful or infeasible, the system can promptly switch to the next available option, maintaining operational continuity. This redundancy fosters adaptability to evolving circumstances or unforeseen obstacles, thereby enhancing collaborative decision-making and optimizing resource distribution. A predetermined list not only supports the exploration of diverse strategies, but also helps agents prioritize the most promising options.

In our dataset, each agent’s decision is represented as a ranked preference list. These lists serve as the basis for the aggregation process, utilizing Kemeny ranking to deliver a consensus that minimizes overall pairwise disagreement among the agents' preferences. Kemeny ranking has received limited attention for multi-agent systems, probably due its computational difficulty. However, its theoretical properties make it appealing; by employing the \textsf{iterative method}, we enable the computation of Kemeny rankings for multi-agent system decisions, offering a theoretically sound alternative.

We can also utilize different variations of Kemeny ranking depending on the context.
In some decision-making situations, narrowing the classification process to focus on key choices can be beneficial. This is where $k$-top list aggregation comes into play. Instead of evaluating all agents' complete rankings, $k$-top aggregation only considers the top-k decisions, capturing the most critical items. This method reduces complexity, making the problem more tractable while ensuring crucial decisions are properly aggregated. In large multi-agent systems, where the number of options can be daunting, emphasizing top-k decisions provides a feasible way to pinpoint the best choices without compromising the consensus-enhancing capability of Kemeny ranking. It is particularly suited to cases where only top choices are crucial, such as selecting job candidates or determining optimal policy alternatives.

An additional extension of Kemeny ranking involves utilizing weighted lists. Often, it is not appropriate for all agents' opinions to carry equal significance. Agents who are more reliable, knowledgeable, or trustworthy should have their preferences prioritized in the final ranking. Weighted Kemeny ranking achieves this by applying greater weights to these agents or particular choices within the ranking. For instance, if an agent possesses more experience or information, their rankings should be given precedence. This approach ensures that the most dependable participants influence the system's final decision. Weighted rankings are also useful for highlighting the significance of certain list positions, such as placing more emphasis on top-ranked choices.

Our base model, along with the \textsf{iterative method} when necessary, effectively addresses traditional Kemeny ranking and its extensions, including $k$-top lists and weighted rankings. These techniques, especially the \textsf{iterative method}, are particularly adept at solving large-scale optimization challenges. In systems with numerous agents and decision options, they provide a practical way to rapidly determine an optimal or nearly optimal ranking that would be computationally challenging to achieve using conventional methods. While our core model and the \textsf{iterative method} are exceptional in large-scale conditions, it is vital to recognize that classical methods for computing Kemeny ranking remain valuable for smaller systems with fewer choices. In these scenarios, their simplicity and directness often make them the preferred option. This distinction highlights the adaptability of our model to varying problem scales and complexities.

\section{Methodology}
In this section, we present the way we implemented the proposed model, the metrics employed and the datasets used to evaluate the performance of the various models in addition to the hardware employed to perform classical computations.

\subsection{Implementation}

We employed Python to interface with the D-Wave Ocean suite, enabling a direct link to the D-Wave quantum computer. Our approach involved formulating the problem as a Constrained Quadratic Model (CQM), requiring the definition of a cost function to represent the problem at hand. In order to facilitate processing on the D-Wave quantum computer, we converted the CQM into a Binary Quadratic Model (BQM) utilizing dedicated functions available in D-Wave's Python library, since this form is required for submission to the annealer. Opting for the CQM-to-BQM conversion method over directly constructing a binary quadratic matrix was based on empirical evidence showcasing its superior efficiency: this approach not only resulted in a marginal increase in the number of occurrences of the optimal solution, but also provided a more straightforward and comprehensible representation of the problem.

In essence, our implementation strategy was designed to enhance the efficiency and effectiveness of solving complex problems through quantum computing. By encoding the problem as a CQM and subsequently transforming it into a BQM, we were able to leverage the capabilities of the D-Wave quantum computer easily and more efficiently.
For the detailed pseudocode of the base model, the \textsf{iterative method}, and pair removal, refer to Appendix \ref{appC}.

\subsection{Metrics}

We employed three distinct metrics to evaluate the performance of the models: 
\begin{enumerate} \item \textbf{Accuracy:} This metric evaluates how closely the annealer's output $X$ matches the binary representation of a set of optimal lists $L$. Let $L = \{L_1, L_2, \ldots, L_t\}$ be a set of optimal lists, where each list $L_t = \{l_{ij}^{(t)}\}$ is represented as a matrix where $l_{ij}^{(t)} = 1$ if $i \prec j$ in the optimal solution and $l_{ij}^{(t)} = 0$ otherwise. The accuracy is defined as the minimum value of the accuracy term across all optimal lists in $L$. Specifically, accuracy is computed as follows:
\begin{equation} 
\text{accuracy} = 
\begin{cases}
0 & \text{if } \min_{t} \left( \sum_{i}\sum_{j>i} |x_{ij} - l_{ij}^{(t)}| \right) > 0 \\
1 & \text{if } \min_{t} \left( \sum_{i}\sum_{j>i} |x_{ij} - l_{ij}^{(t)}| \right) = 0 
\end{cases} 
\end{equation}
The accuracy metric is used to evaluate how closely the output generated by the annealer matches the most similar optimal solution obtained through classical methods. When possible we will prefer accuracy over alternative metrics, such as the Kendall-Tau distance, because it is not affected by the random selections made during list reconstruction in cases of outputs containing cycles (while the optimal solution could occasionally be derived from such outputs, accuracy ensures that the output of the annealer is exactly the binary representation of the optimal solution), ensuring a more reliable and robust evaluation.\\
\item \textbf{Number of occurrences:} This parameter, denoted as \textit{num\_occ}, indicates how many times a specific solution appears during the annealing process. This number can be used to quantify the confidence we have in a certain output.
In instances where multiple optimal solutions exist, we consider the total number of occurrences across all optimal solutions.

It should be noted that during the data collection process, \textit{num\_occ} was considered different from 0 only when the output was an optimal solution.\\
\item \textbf{Kendall-Tau distance:} This metric, along with its normalized version, is employed when the accuracy cannot be determined. This occurs in scenarios where it is not feasible to compute the solution using classical methods (brute force), thereby preventing a comparison between the classical computation and the output of the annealer.
The Kendall-Tau distance of the output is computed as the sum of the Kendall-tau distance between the output list and each list present in the dataset.
\end{enumerate}

\subsection{Hardware}
In the domain of computational research, the choice of hardware significantly impacts the efficiency and accuracy of computations. For our classical computation and data collection endeavors, we have relied on the following hardware setup:
\begin{itemize}
\item Processor: Intel(R) Core(TM) i7-4710MQ CPU @2.5GHz (8CPUs)\\
\item Memory: 8192MB RAM
\end{itemize} 
It is imperative to acknowledge that this hardware configuration is now outdated. Transitioning to modern hardware has the capacity to significantly improve the outcomes and subtly alter the conclusions of our study.

\subsection{Datasets}
To showcase the effectiveness of our method and assess its limitations, we employed two different types of datasets:

\begin{enumerate}\setlength\itemsep{0.5em}
    \item \textbf{Synthetic datasets}: These datasets were created by randomly permuting a starting list to generate individual votes.\\
    \item \textbf{Simplified synthetic datasets}: Similarly to the synthetic datasets, these datasets were also generated by randomly permuting a starting list. However, in this case, for each vote, the starting list was divided into 3 or more sublists of random length, and the permutation was performed within each sublist. This method introduces additional structure and reduces dataset complexity, typically resulting in fewer cycles.\\
\end{enumerate}

In the dataset, each candidate is represented by a number ranging from 0 to $n-1$, where $n$ denotes the number of candidates. From now on, we will refer to a larger dataset as one with a larger number of candidates. Each dataset consists of 11 votes, a specific number chosen to mitigate the likelihood of encountering degenerate solutions, a common complication with datasets containing an even number of votes. Moreover, it simplifies the iterative process by avoiding potential cycles caused by the random values assigned to the zero-bias qubits.

\section{Results}\label{sec2}

This section presents the results of testing the model and its enhancements on various datasets. Each result is derived from $n_{reads}=2500$ using the D-Wave Advantage System 4.1 sampler with default parameter settings, and consists of the average of the outcomes of 10 repeated runs.\\
It is important to clarify the terminology used when referring to the annealer outputs:
\begin{itemize}
    \item ``Output" refers to the result produced by the annealer, without considering whether it is the correct output or not;
    \item ``Solution" indicates the correct (expected) output that the annealer is supposed to return when given a problem;
    \item ``Optimal solution" is the optimal solution to the overall problem being solved.
\end{itemize}
This distinction is necessary because the annealer does not always return the expected output, especially for very difficult problems. Additionally, due to the iterative nature of the method, the annealer may find a solution to a specific subproblem that satisfies the constraints, but it may not be the optimal solution to the overall problem.

The first data we present is a performance comparison between our base model and another model \cite{fiergolla2023heuristicdiversekemenyrank} that uses a different list representation. This comparison allows assessing the impact of these different representations.
Following this, we evaluate the performance of the \textsf{iterative method} on the same datasets to measure the enhancement in solution accuracy. We then progress to datasets with a larger number of candidates to evaluate the method's effectiveness, enhancements, and limitations.
Next, we evaluate the effectiveness of the pair removal method across various datasets. This evaluation provides insights into the method's performance and its potential contributions to the overall effectiveness of the model.
Finally, we offer a comparison between the \textsf{iterative method} and a classical approach: KwikSort \cite{KwikSort}. The comparison evaluates the two methods as ``exact methods", i.e., their ability to find the optimal solution, and then as approximate methods. The first comparison assesses the execution time when both methods achieve the optimal solution. The second comparison evaluates the two methods on very large datasets, where the \textsf{iterative method} can only be used as an approximate method due to its inability to converge to an optimal solution.

The findings related specifically to partial and weighted lists are detailed in Appendix \ref{pawl}.

\subsection{Base model}
In this subsection, we first discuss the comparison between the base model and another quantum algorithm \cite{fiergolla2023heuristicdiversekemenyrank} which we will call the $n^2$-representation model, and then present the results of our analysis.
\subsubsection{Comparison base model vs $n^2$-representation model}
The main difference between the two methods lies in the representation of the list: while our model represents the list as a binary strictly upper triangular matrix, where each qubit represents the order of the candidates involved, the $n^2$-representation model uses a square matrix where each row represents a candidate and each column a position.
More in detail, the ranking of n candidates is represented by $n^2$ binary variables ($c_{i,j}$), where $c_{i,j}$ = 1 indicates that candidate i occupies position j, and $c_{i,j}$ = 0 indicates otherwise. 
To ensure that the solution is meaningful, the constraints that each candidate holds exactly one position and that each position is occupied by only one candidate are imposed.
The problem can be visualized as a two-dimensional grid where in each row and column, there is only one element set to 1 in the grid.\\ 
The cost function describing the problem is then:

\vspace{0.5cm}

\begin{equation}
\begin{aligned}
f(c_{1,1},...,c_{n,n}) &= \sum_i^n(row_i+col_i+\sum_j^n rank_{i,j}) \\
&=\scalebox{0.65}{$\displaystyle\sum_i\left(P\left(1-\left(\sum_j c_{i,j}\right)\right)^2 + P\left(1-\left(\sum_j^n c_{j,i}\right)\right)^2+\sum_j^n w_{ij}\sum_k^{n-1}\sum_{l>k}^nc_{i,k}c_{j,l}\right)$}
\end{aligned}
\end{equation}\\
\vspace{0.5cm}
\\
where the $row_i$ and $col_i$ terms ensure that each candidate occupies exactly one position and that each position is filled by only one candidate, with $P$ representing the penalty coefficient.

The $rank_{i,j}$ term instead, corresponds to the Kendall-Tau distance term for the pair (i,j): it is $w_{ij}$ when candidate i is ranked above candidate j, and 0 when candidate j is ranked higher than candidate i.
To ensure that $P$ always remains an unpayable cost, it must satisfy the condition: $P=n^2|\Pi|$, where $|\Pi|$ represents the total number of votes in the dataset.

Comparing the $n^2$-representation model with our own reveals significant differences:
\begin{enumerate}\setlength\itemsep{0.5em}
    \item \textbf{Number of qubits}: The $n^2$-representation utilizes more than twice the number of qubits. Although the number of qubits may not be a significant limitation, it certainly impacts performance.
    \item \textbf{Penalty coefficient}: The $n^2$-representation model requires a significantly higher value of $P$ due to its dependence on $n^2$. While the penalty coefficient is essential for producing accurate outputs and ensuring the model functions correctly, it also increases the computational load and limits performance. Comparing the two coefficients, $P=|\Pi|$ for our model's worst-case scenario and $P=n^2|\Pi|$ for the $n^2$-representation model, it is evident that as the number of candidates increases, the performance of the latter is significantly impacted. This dependence on a higher penalty coefficient makes it difficult for the $n^2$-representation model to achieve good results on datasets with a large number of candidates.
    \item \textbf{Output interpretation}: Unlike our model, this representation model does not mandate post-processing for final output comprehension, saving time and resources.\\
\end{enumerate}
In conclusion, the $n^2$-representation model trades off the capability to evaluate a large dataset in exchange for a simplified and intuitive interpretation of the output. The clear correspondence between rows, representing candidates, and columns, representing positions, allows a straightforward understanding of the output ranking.

\subsubsection{Results}
In our evaluation, we compared the performance of the two methods on datasets taken from the synthetic datasets set. The penalty coefficient $P$ for the base model was set to $P=|\Pi-2|+\epsilon$ to ensure cycle removal, given that the datasets employed have an odd number of votes. We evaluate the accuracy and the number of occurrences for datasets with different numbers of candidates, averaging the results over 10 repeated runs.
\begin{figure}[H]
    \centering
    \begin{subfigure}{0.48\textwidth}
        \centering
        \includegraphics[width=\linewidth]{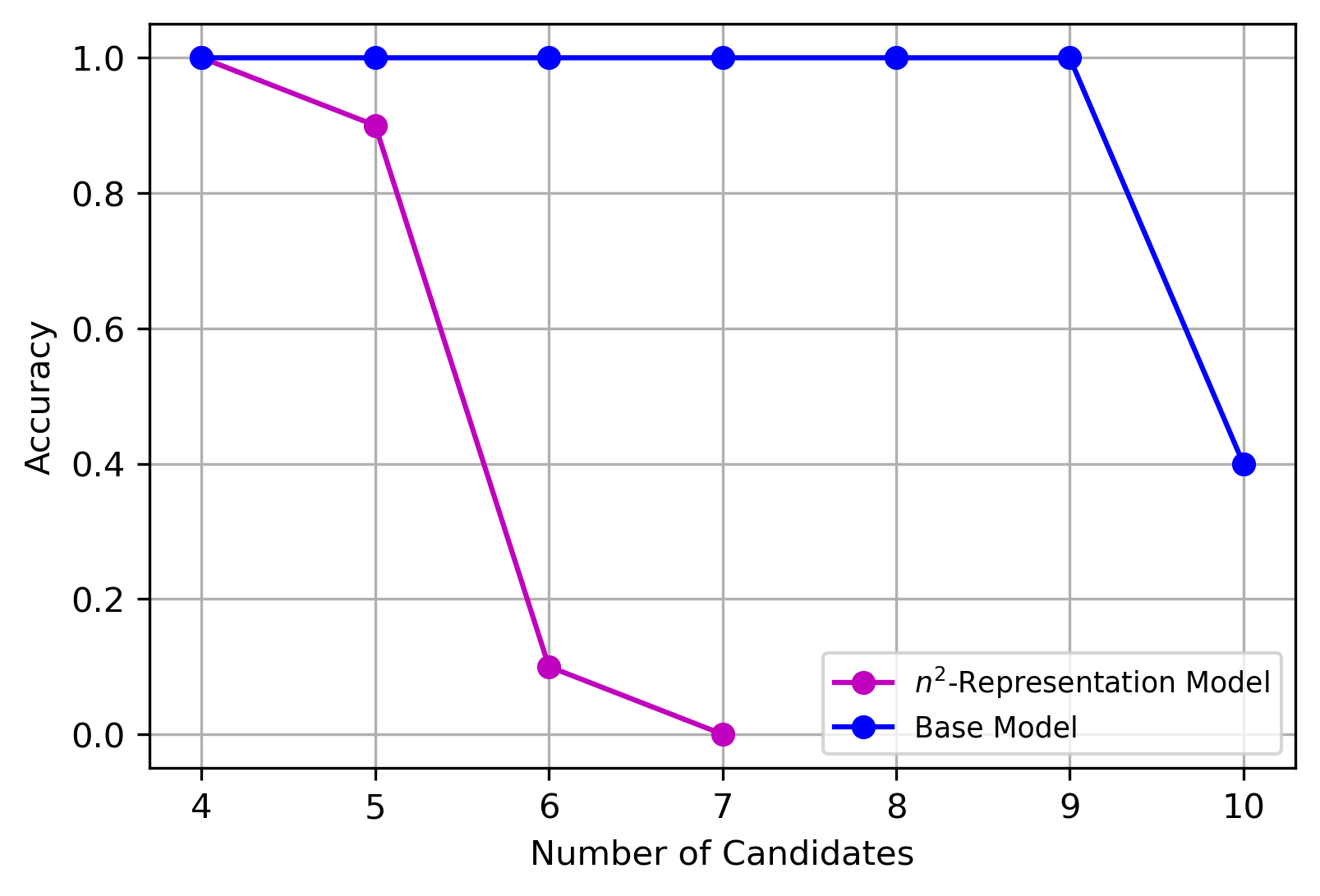}
        \caption{}
        \label{f:pp1}
    \end{subfigure}
    \hfill
    \begin{subfigure}{0.48\textwidth}
        \centering
        \includegraphics[width=\linewidth]{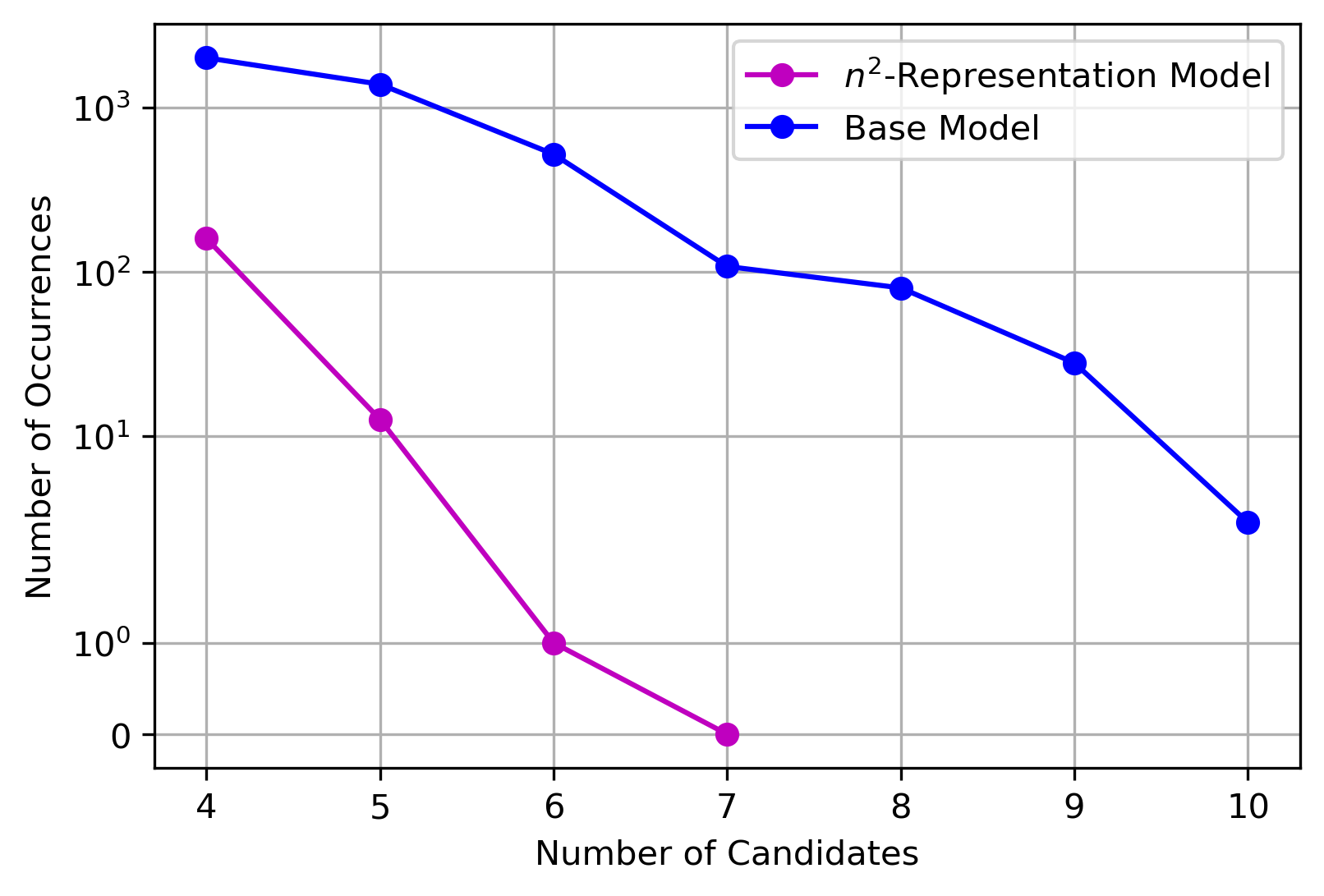}
        \caption{}
        \label{f:pp2}
    \end{subfigure}
     \caption{Comparative evaluation of the models accuracy (\cref{f:pp1}) and the number of occurrences (symlog scale) (\cref{f:pp2}) for the base model and the $n^2$-representation model across different candidate datasets.}
    \label{fig:comparisonbmn2}
\end{figure}
As shown in \cref{fig:comparisonbmn2}, the base model outperforms the $n^2$-representation model in terms of both accuracy and number of occurrences, allowing the handling of larger datasets.
Moreover, since the base model identifies the optimal solution more often, the number of reads required to find the optimal solution can be reduced, saving quantum annealing computation time.
Ultimately, the results demonstrate that the choice of representation significantly influences the efficiency of calculating the Kemeny ranking, primarily due to the effect of the associated penalty coefficient.\\

\subsection{Iterative Method}
The main advantage of the pair representation lies in the possibility to identify the cycles within an arbitrary output matrix, ensuring that penalties are imposed only to the cycles that truly exist.
The \textsf{iterative method} capitalizes on this by seeking out the cycles in the output of the annealer and applying the smallest effective coefficient as a penalty. For a detailed explanation of this method, see Section \ref{im}.
The advantages of this procedure are supported by the results in Appendix \ref{pu2}, which indicate that the required value for the penalty coefficient is often much smaller than the theoretical threshold. Moreover, the performance degrades as the penalty coefficient increases, highlighting the need for a method that maintains the penalty coefficient as small as possible.

The preliminary analysis entails a comparison between the iterative method and the base model, focusing on evaluating the frequency of solution occurrences in both models, using the same datasets as in prior evaluations.
\begin{figure}[H]
    \centering
    \begin{subfigure}{0.48\textwidth}
        \centering
        \includegraphics[width=\linewidth]{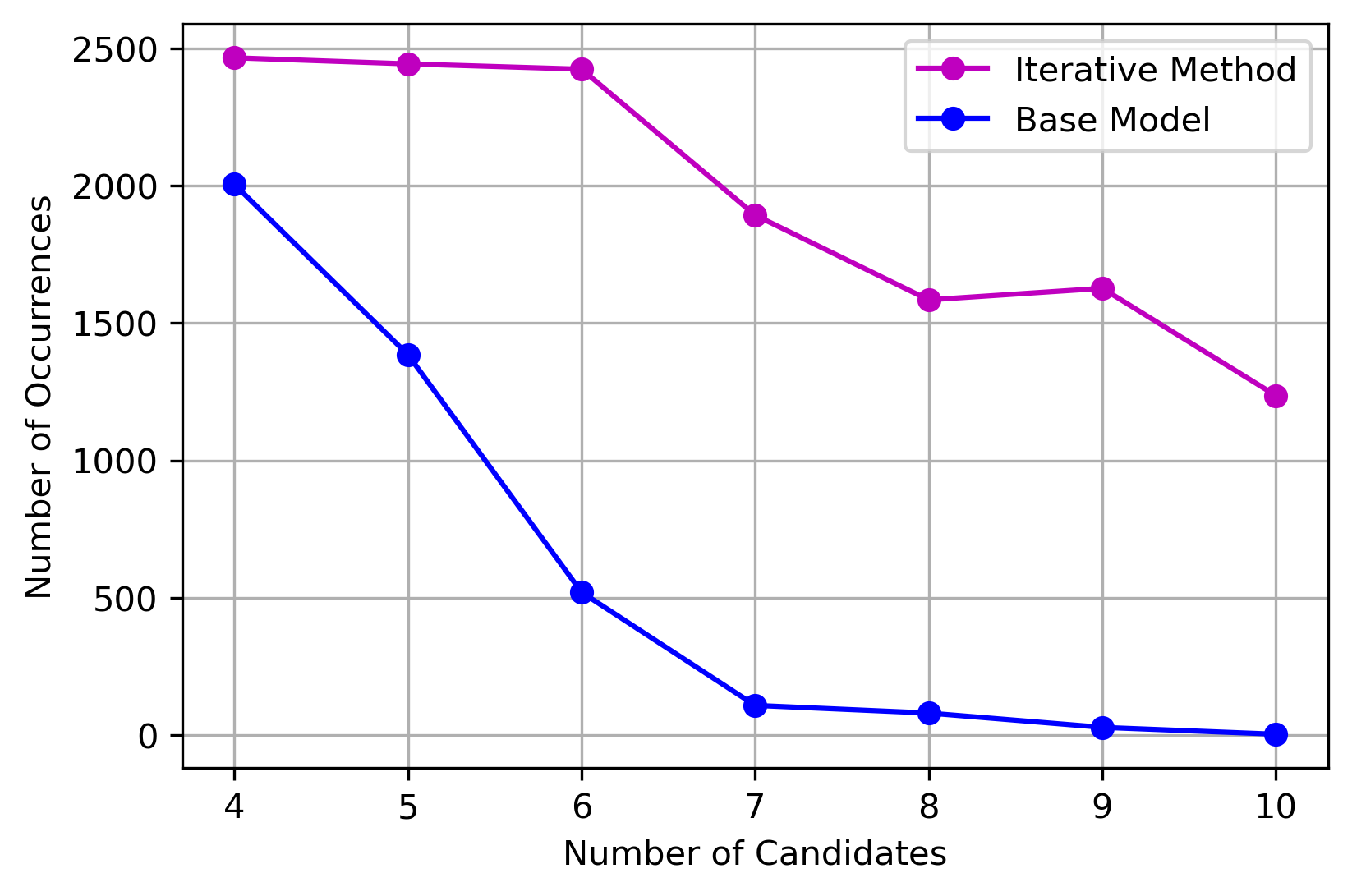}
        \caption{}
        \label{cit1}
    \end{subfigure}
    \hfill
    \begin{subfigure}{0.48\textwidth}
        \centering
        \includegraphics[width=\linewidth]{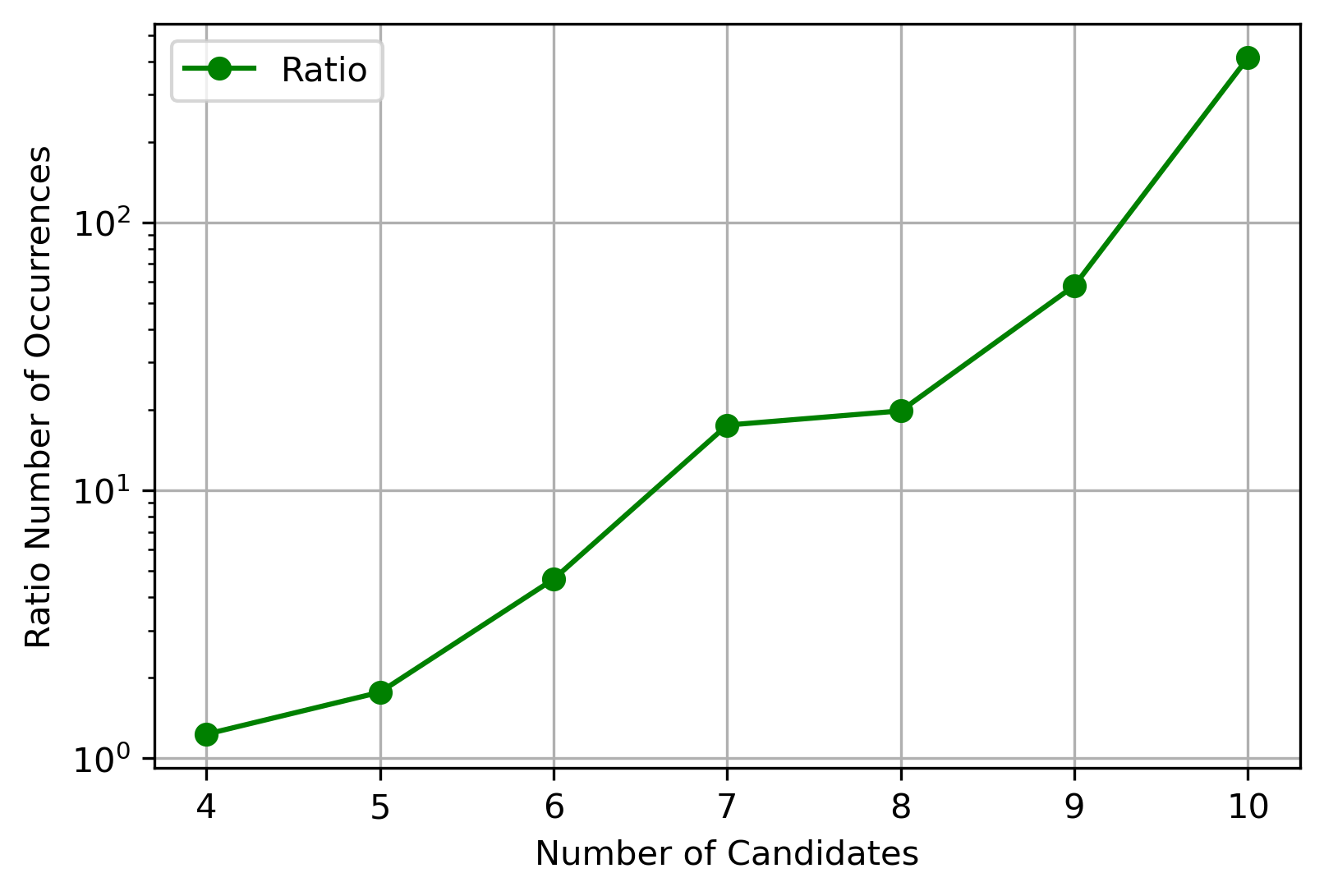}
        \caption{}
        \label{cit2}
    \end{subfigure}
    \caption{Comparative evaluation of the number of occurrences (\ref{cit1}) and their ratio (\ref{cit2}) between the base model and the iterative method across different candidate datasets. he values represent the average of 10 runs for each dataset, where each dataset corresponds to a different number of candidates.}
    \label{f}
\end{figure}
We observe a significant improvement in performance using the \textsf{iterative method} compared to the base model in terms of the number of occurrences (\cref{cit1}). This improvement is further highlighted in (\cref{cit2}), which indicates the ratio of the number of occurrences between the outputs of the two models, showing how much more frequently the \textsf{iterative method} produces successful results with respect to the base model.
This performance improvement can be attributed to the selective penalization of only the relevant cycles that arise in the problem, using the smallest effective coefficient, which leads to an overall boost in performance.

The second comparison regards the performance of the two methods as a function of the number of votes. For this specific task we created a dataset with 9 candidates and 1001 votes and selected subsets of votes for each data point in the plots.
For the \textsf{iterative method}, with datasets containing a high number of votes, finding the minimum value high enough to effectively penalize and remove the cycles would require too many iterations when proceeding in increments of 2.
Therefore, for this case, we decided to adopt a different approach: instead of assigning 1 as the initial penalty coefficient, we assign to each cycle an initial penalty coefficient equal to the minimum (plus $\epsilon$) of the three biases of the pairs involved in the cycle. This method allows for a faster achievement of the required penalty coefficient without requiring too many iterations.

\begin{figure}[H]
    \centering
    \begin{subfigure}{0.48\textwidth}
        \centering
        \includegraphics[width=\linewidth]{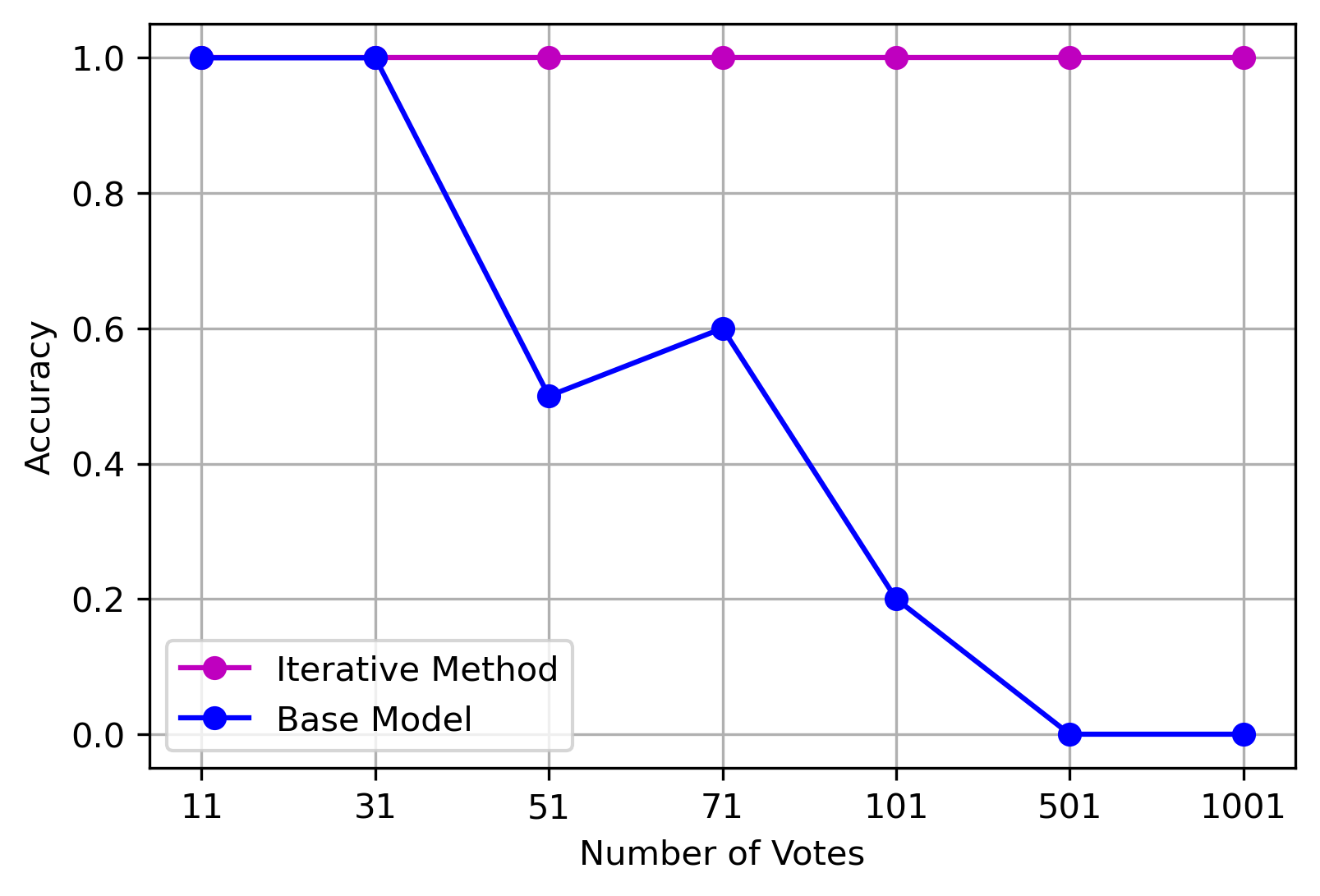}
        \caption{Accuracy comparison}
        \label{accuracyfv}
    \end{subfigure}
    \hfill
    \begin{subfigure}{0.48\textwidth}
        \centering
        \includegraphics[width=\linewidth]{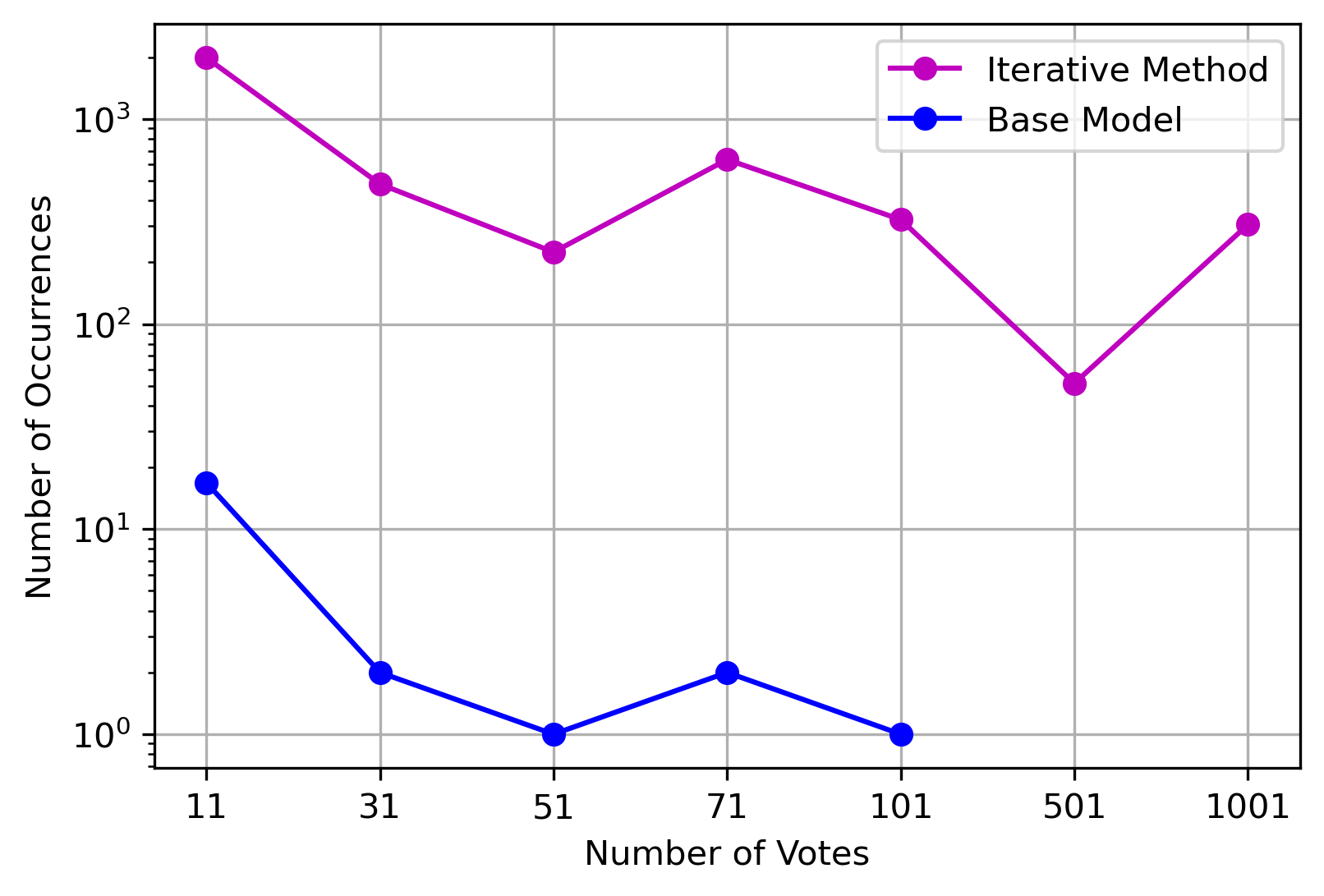}
        \caption{Number of occurrences comparison}
        \label{noccfv}
    \end{subfigure}
    \caption{Comparative evaluation of the accuracy (\ref{accuracyfv}) and the number of occurrences (\ref{cit2}) between the base model and the iterative method across different number of votes datasets. The values represent the average of 10 runs for each dataset, where each dataset corresponds to a different number of votes.}
\end{figure}
\begin{figure}[h]
    \centering
    \includegraphics[width=0.6\linewidth]{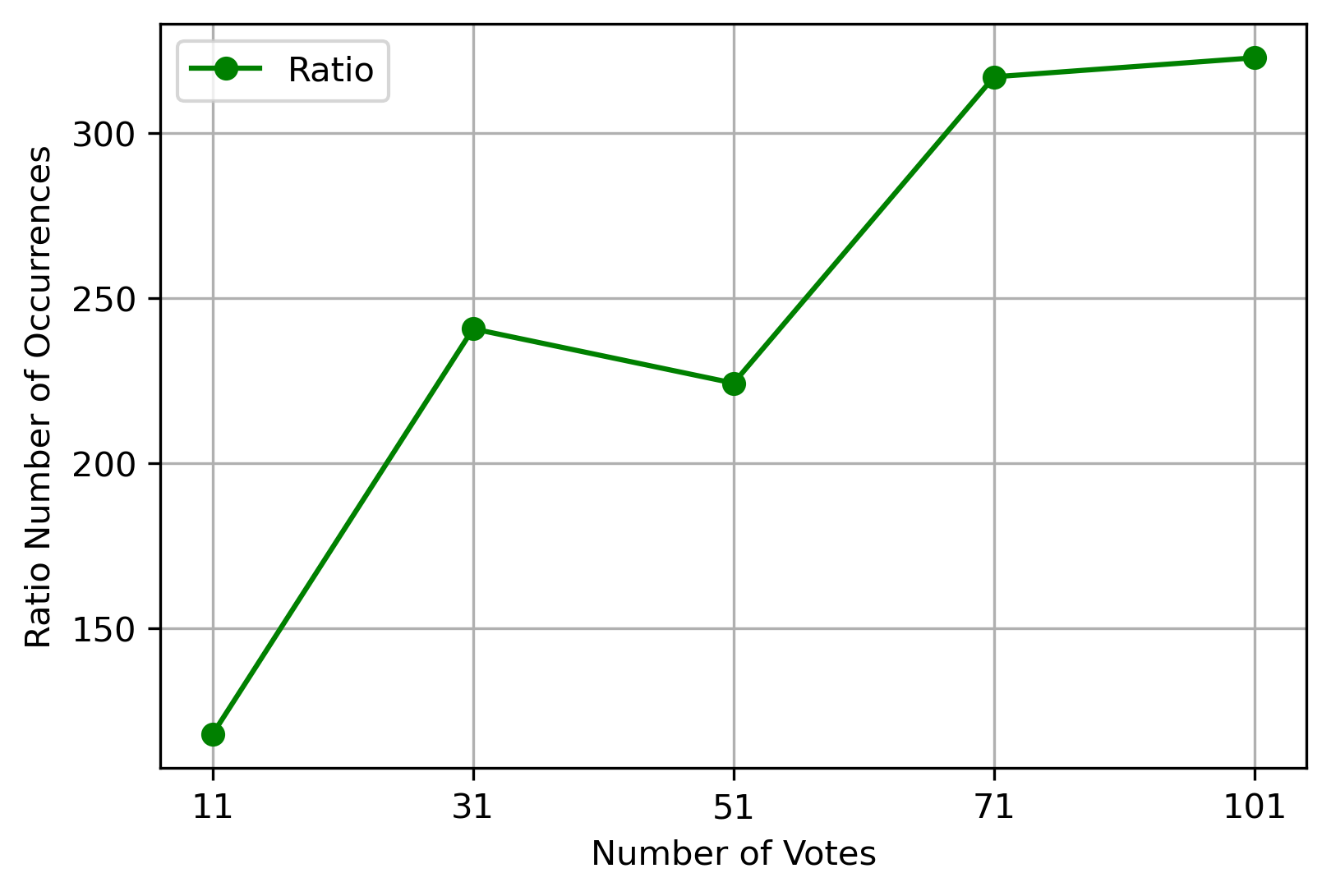}
    \caption{Ratio of the number of occurrences, calculated as the number of occurrences of the iterative method divided by the number of occurrences of the base model.}
    \label{ratiofv}
\end{figure}
\cref{accuracyfv} indicates that the accuracy of the base model decreases as the number of votes increases, whereas it remains consistently at 1 for the \textsf{iterative method}. This behavior is also observable in \cref{noccfv}, where the number of occurrences for the base model is generally decreasing, while for the iterative method, it still decreases but at a slower rate. As a result, the ratio between the two values is generally an increasing function (\cref{ratiofv}).
This is due to the penalty coefficient of the base model, which increases in direct proportion to the number of votes in the dataset, leading to a progressive deterioration in performance as the vote count rises. In contrast, while the penalty coefficients for the \textsf{iterative method} also generally increase, they do not grow at the same rate. This more stable behavior highlights the effectiveness of penalizing only the relevant cycles with appropriately tuned coefficients, allowing the \textsf{iterative method} to maintain better performance even as the number of votes increases.

Now, our goal is to evaluate whether the \textsf{iterative method} can perform well for significantly larger datasets. The primary factors affecting performance are still the number of cycles involved and the penalty coefficients required to eliminate them.
We evaluate the output across four datasets extracted from the simplified synthetic dataset set, which we label as \{$N, C_f$\}, where $N$ represents the number of candidates in the datasets and $C_f$ represents the closest integer to the average total number of cycles that appeared during the solution process over 10 runs (see \cref{average}). This labeling provides the elements for an approximate estimation of the complexity one might encounter in finding the solution.
Since we cannot directly compare the output with an actual solution, we use the normalized Kendall-Tau distance as a metric to evaluate the solution quality, averaging the results over the usual 10 runs for robustness.
Furthermore, in order to visualize the data from all four datasets on a unified plot, we subtract from each set of distances the minimum distance found for that set, so that the minimum distance becomes zero.
\begin{figure}[h]
        \centering
        \includegraphics[width=0.7\textwidth]{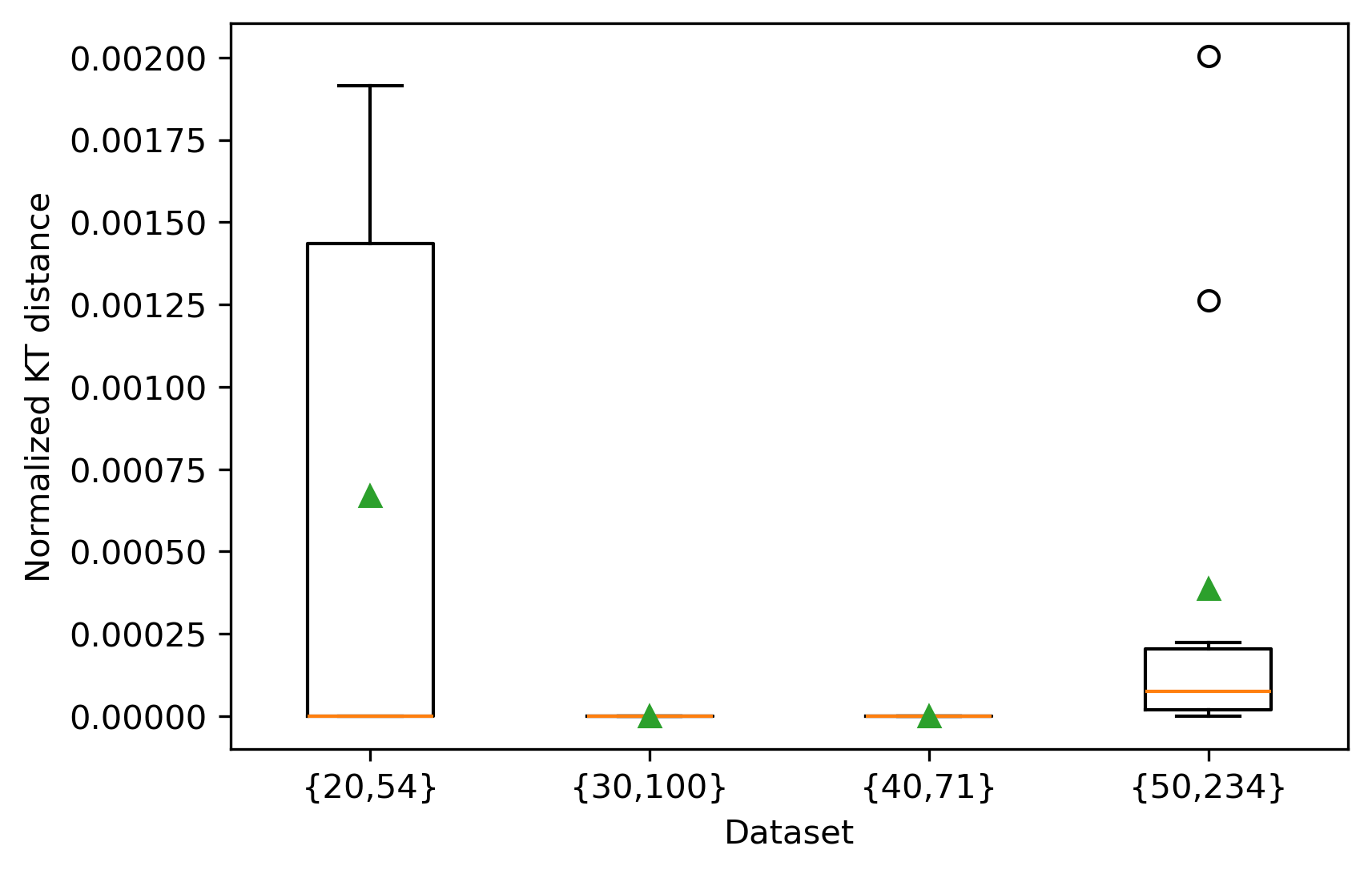}
        \caption{Boxplot comparison of baseline-adjusted normalized Kendall-Tau distances across four datasets (\{N, $C_f$\}) with varying numbers of candidates (N) and average penalized cycles ($C_f$).
        For the dataset with 20 candidates, the minimum distance was observed in 6 out of 10 runs, while for the datasets with 30 and 40 candidates, it was observed in all 10 runs. In contrast, for the dataset with 50 candidates, the minimum distance was achieved only 3 times.}
        \label{grafico2054}
\end{figure}
One significant discovery from these results (see \cref{grafico2054}) is that even when working with datasets containing 20 candidates, the best solution (which we have no way of claiming is optimal) is not always found consistently. This underscores the fact that, depending on the dataset's complexity, even a modest number of candidates can pose difficulties in pinpointing the best solution.\\
Another finding is that the ability to solve the problem is influenced not only by the number of cycles and candidates but also by the distribution of cycles and the penalty coefficients required to remove them. For instance, datasets with fewer candidates and cycles (e.g., \{20, 54\}) may yield best solutions less frequently compared to datasets with higher counts such as \{30, 100\} and \{40, 71\}.

It's worth noting that an initial observation of the plot might suggest that the algorithm performs better on the dataset \{50, 234\} with respect to \{20, 54\} due to the smaller box size. However, a closer examination of the data reveals that the best found solution is actually found twice as often in the dataset with 20 candidates. 
Moving forward with the analysis, in the dataset \{50, 250\}, the smallest distance recorded was 2328, while a commonly found output configuration without cycles had a distance of 2329. This scenario could potentially be attributed to either a chain break or the flip of a qubit representing two adjacent elements in the final list. This highlights the critical need for multiple validations, as even without cycles, there is no guarantee that the output found is the optimal solution in complex scenarios.
The difficulty level of each dataset varies, and while it can be roughly estimated by considering the number of candidates and cycles, it is evident that additional factors beyond these influence the optimization challenge presented to the annealer, such as the embedding, how the cycles are connected or the value of the penalty coefficients necessary to remove them.

The comparative analysis reveals improved performance compared to the base model, primarily due to the strategic penalization of only relevant cycles with reduced penalization terms. The \textsf{iterative method}'s efficacy is intricately linked to the nuances of the dataset. In fact, datasets with identical candidate and vote counts can yield significantly different results. In a way, it can be said that the \textsf{iterative method} takes advantage of the simplicity of the dataset.
Finally, in complex scenarios, multiple validations are essential, as an output without cycles does not guarantee that it is the optimal solution.
Despite these challenges, the method represents a significant advance in performance capabilities.

\subsection{Pair removal}\label{pairremoval}
To enhance the annealer's performance and address instances where it may not consistently yield the optimal solution, we adopt a strategic pre-processing step that simplifies the problem complexity. This simplification entails the temporary exclusion of certain pairs from the problem prior to submission to the annealer. After the annealing process, these excluded pairs are reintroduced, with their values inferred via the transitive property. Figure \ref{fig:prf} illustrates the annealer's output, highlighting the temporarily excluded pairs with a placeholder value of 0.5 (in red), and the configuration after deducing the values of these pairs.
\begin{figure}[H]
    \centering
    
    \begin{subfigure}{0.48\textwidth}
        \centering
        \includegraphics[width=\linewidth]{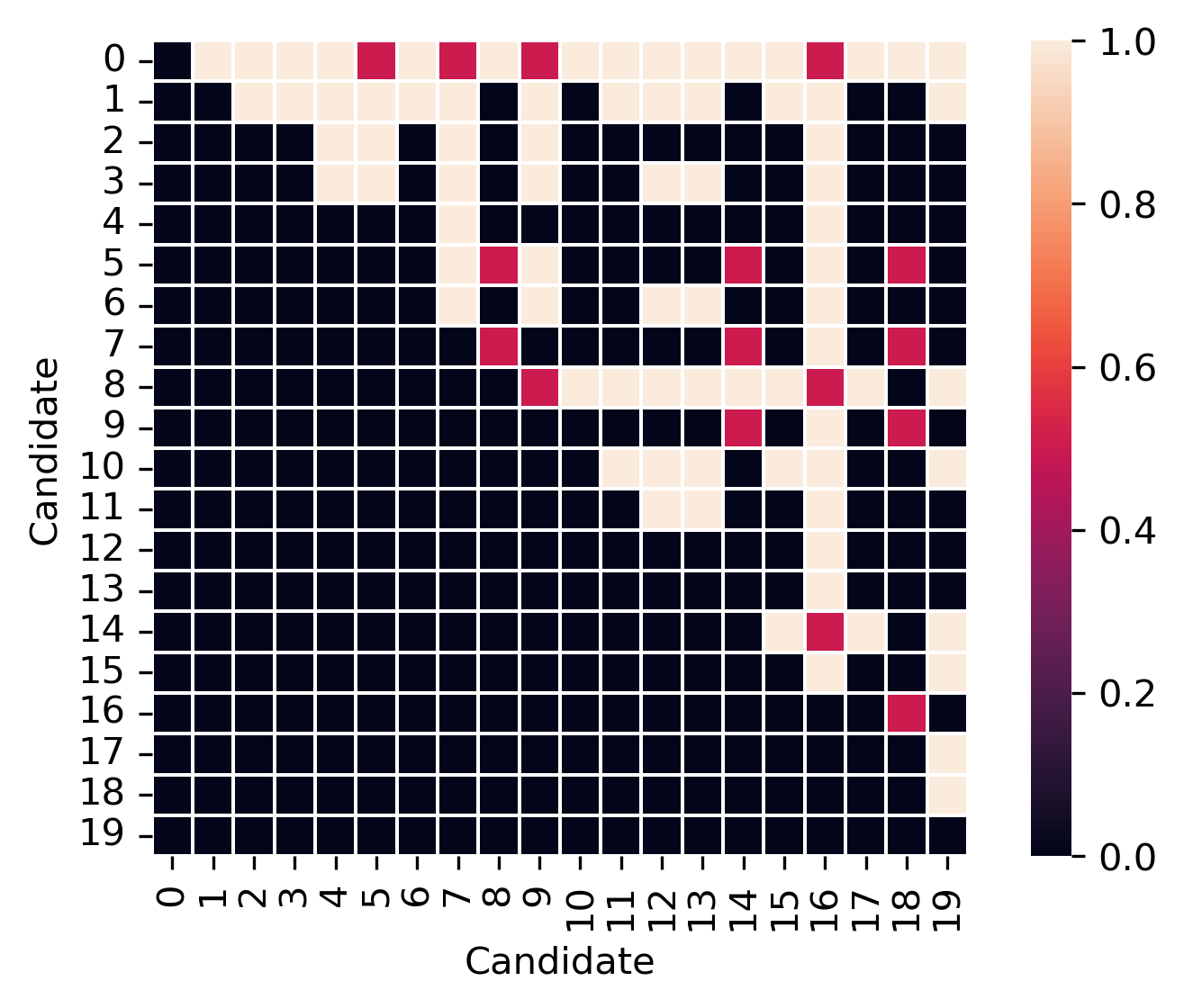}
        \caption{Annealer's output}
        \label{aopr}
    \end{subfigure}
    \hfill
    \begin{subfigure}{0.48\textwidth}
        \centering
        \includegraphics[width=\linewidth]{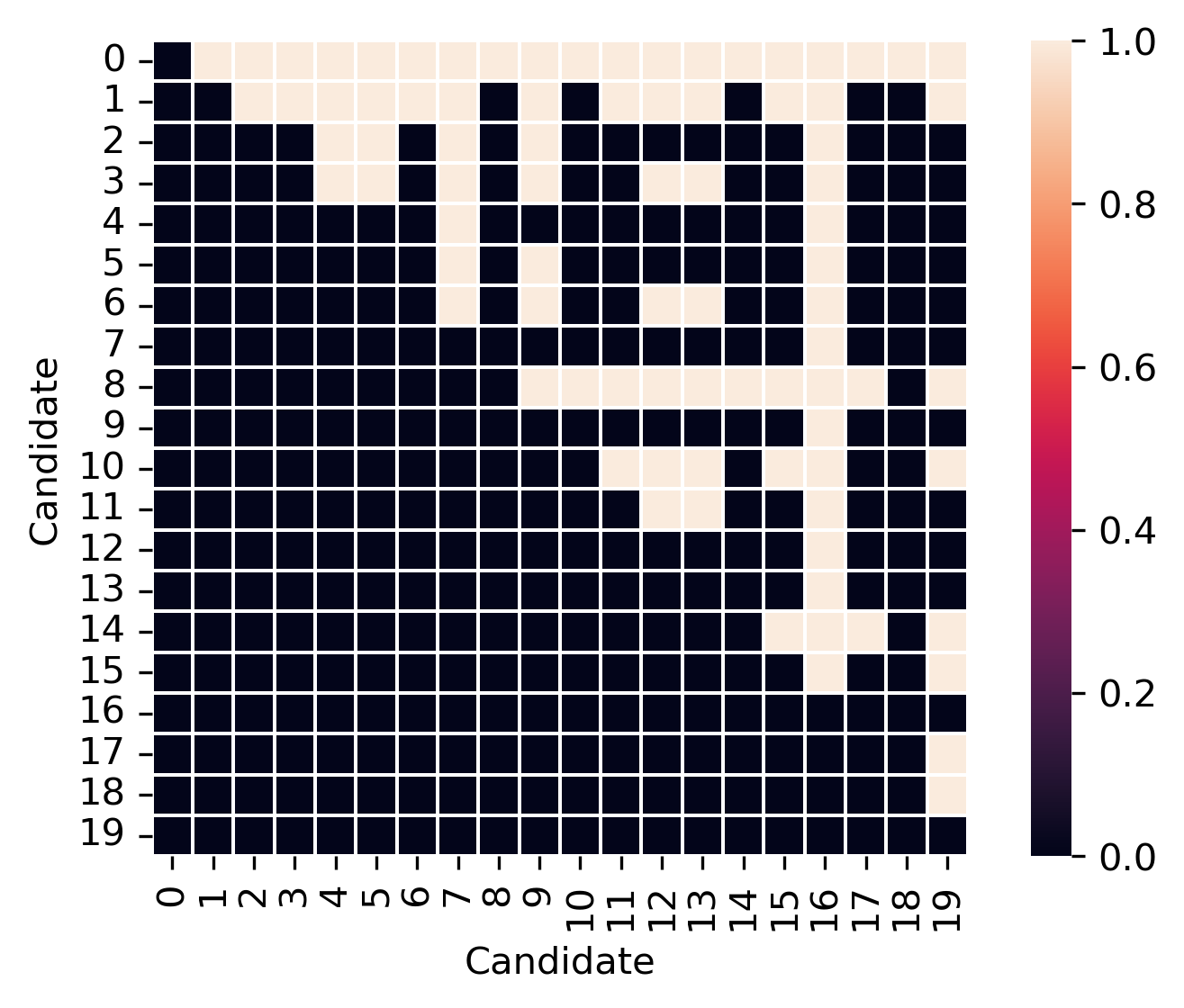}
        \caption{Output after reintroduction of removed pairs}
        \label{aoprf}
    \end{subfigure}

    \caption{Example of $PR\Omega$ output (\ref{aopr}) where removed pairs have a temporary value of 0.5 (red) and corresponding output after pair reintroduction (\ref{aoprf}).}
    \label{fig:prf}
\end{figure}
The first experiment evaluates the effectiveness of applying pair removal to the base model using the same dataset employed before to assess its performance. This approach enables us to measure how the pair removal strategy enhances the achievement of the optimal solution.
These tests involved the removal of an identical number of pairs in both the $PR\Omega$ and $PRHB$ techniques in order to provide a balanced comparative study.

\begin{figure}[H]
        \centering
        \includegraphics[width=0.7\textwidth]{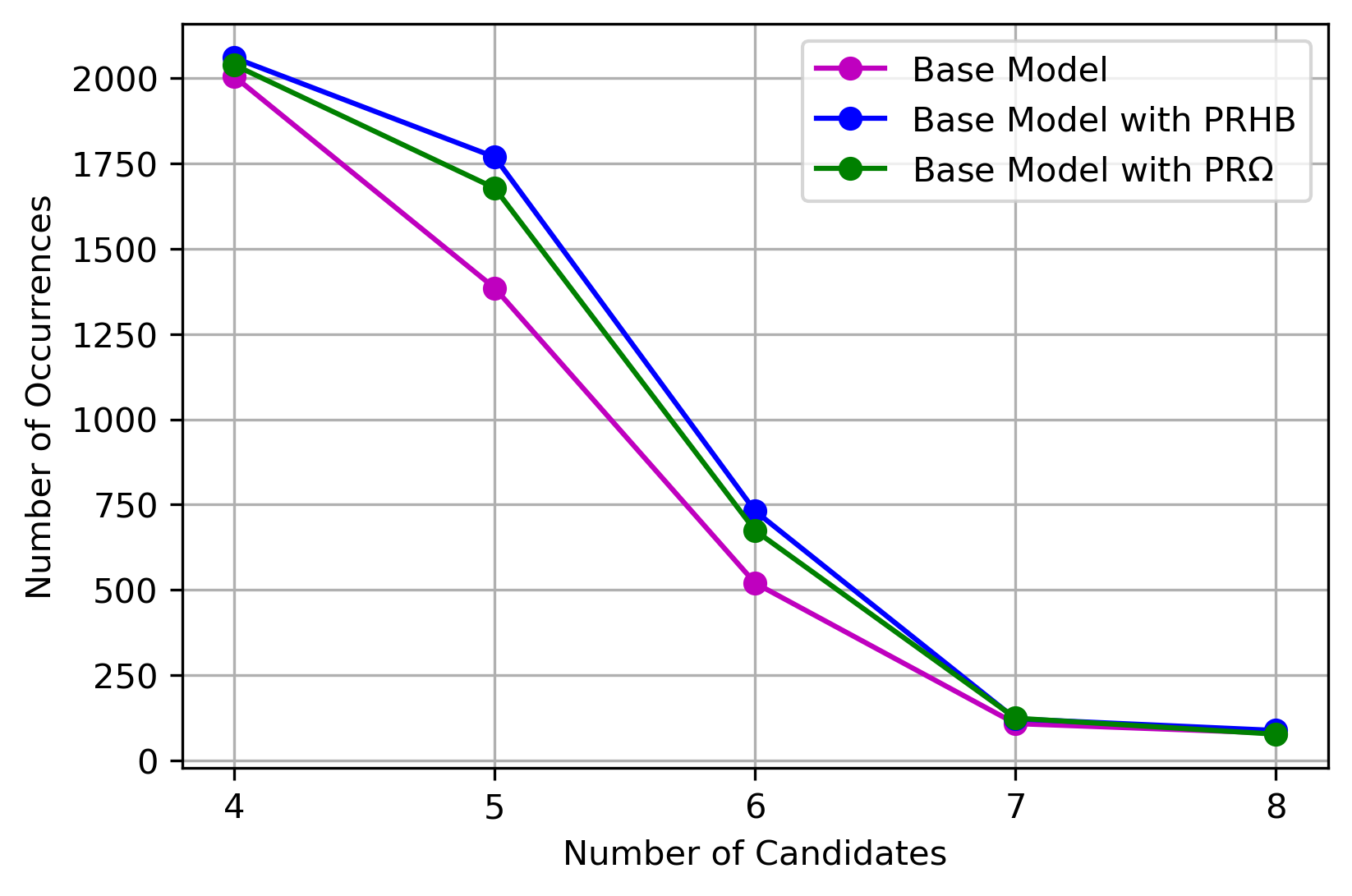}
        \caption{Comparison of the number of occurrences for the base model, the $PRHB$ method and the $PR\Omega$ method. Pairs were removed as follows: 2 for datasets with 4 and 5 candidates, 4 for datasets with 6 candidates, and 5 for datasets with 7 and 8 candidates.}
        \label{pr1}
\end{figure}
We found that, on average, the use of pair removal strategies outperforms the base model that does not incorporate pair exclusion (\cref{pr1}). 
Although the improvement is modest, given the ``limited" number of pairs removed, it proves to be helpful, reducing the complexity of the problem for the annealer.

We now move to the analysis of pair removal applied along with the \textsf{iterative method} on the previously examined dataset consisting of 20 candidates and 54 total cycles on average.
Recall that the \textsf{iterative method} was not always able to find the optimal solution for this dataset (\cref{grafico2054}). In an effort to address this challenge, we explore the removal of 16 pairs using both methods (\cref{pr2}).
\begin{figure}[H]
        \centering
        \includegraphics[width=0.7\textwidth]{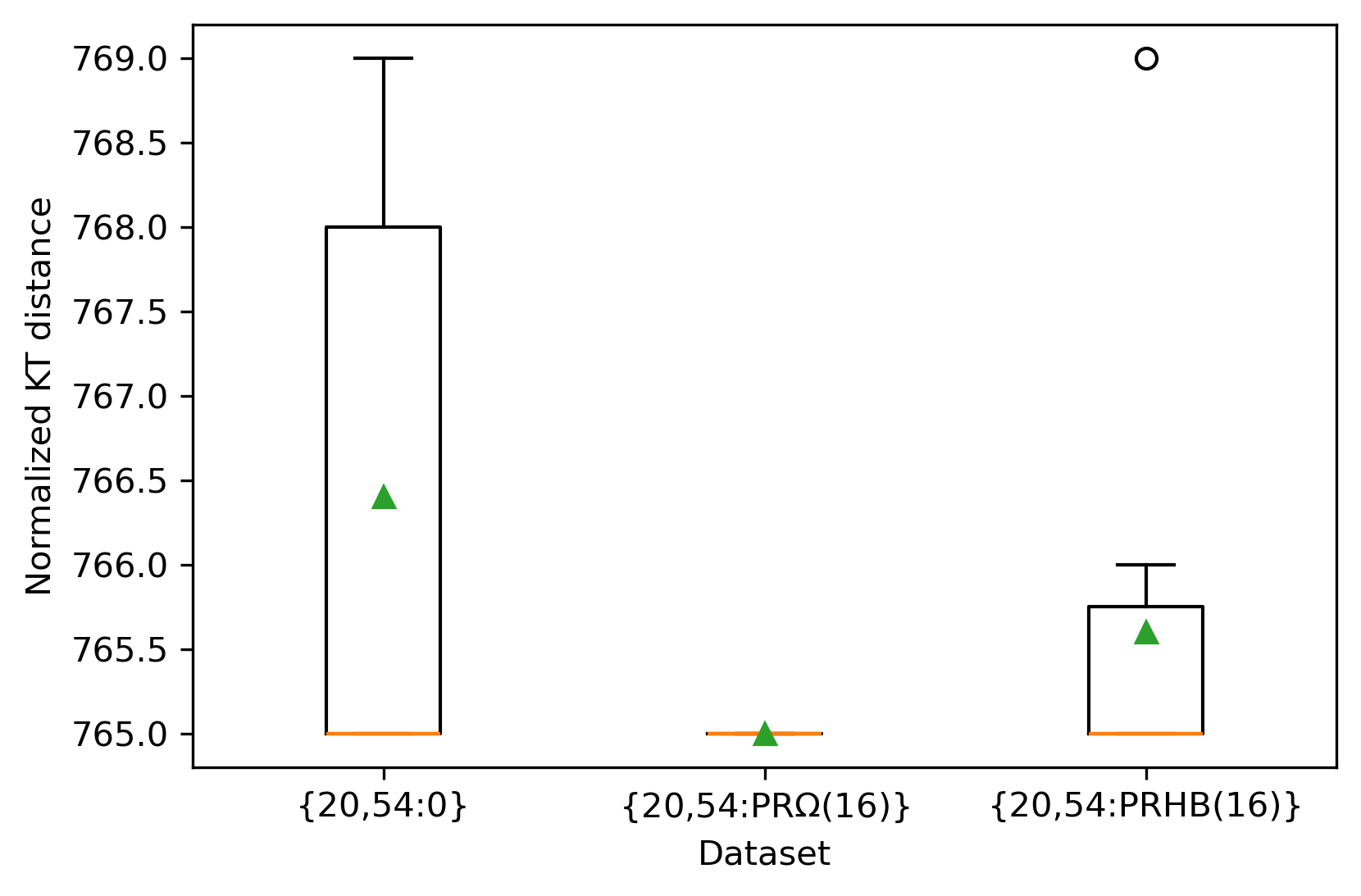}
        \caption{Boxplot comparison of Kendall-Tau distances for dataset \{20, 54\} for the iterative method, the iterative method with $PR\Omega$ and the iterative method with $PRHB$. The minimal distance was observed six times for the iterative method, always for the $PR\Omega$ method and seven times for the $PRHB$ method.}
        \label{pr2}
\end{figure}
Our findings show that the average distances between the output list and the dataset when pair removal is applied are generally lower than those obtained without this strategy, although all methods were able to achieve the same minimum distance. In particular, the $PR\Omega$ technique demonstrated superior performance due to its more consistent results and lower variability in distances.
While the benefits of pair removal are clear, its application must be approached with caution. Excessively or improperly removing pairs can lead to outputs that, while free of cycles, contain inaccurately valued pairs. Moreover, errors in early reconstruction steps can have cascading effects on the accuracy of pair values deduced in later steps.

In conclusion, our study validates the effectiveness of pair removal strategies in improving the quality of the solution, highlighting the importance of comparing and selecting the most suitable approach, as both $PR\Omega$ and $PRHB$ offer unique advantages and limitations.

\subsection{Iterative method vs KwikSort}

In our analysis so far, we have focused only on evaluating the solution quality of the \textsf{iterative method}, overlooking the time investment required to attain it. This section aims to provide a comprehensive performance comparison between the annealer-based method and traditional methods, factoring in both solution quality and time efficiency. By considering these aspects, users can make informed decisions tailored to their unique requirements. Our evaluation of classical methods includes a brute-force approach, where we computed the distance between each possible list and the dataset, alongside the KwikSort algorithm \cite{KwikSort}.
KwikSort, a variant of the QuickSort algorithm, is a comparison-based sorting algorithm that follows the divide-and-conquer paradigm. It works by selecting a ``pivot" element from the array and partitioning the other elements into two subarrays based on whether they are ranked higher or lower than the pivot in most of the input rankings. This process is recursively applied to the sub-arrays, ultimately leading to a sorted array. The efficiency of KwikSort, like QuickSort, relies on the choice of the pivot and can achieve an average case time complexity of $\mathcal{O}(n\log n)$ and a worst case time complexity of $\mathcal{O}(n^2)$ if poor pivot choices are made consistently.
In the context of Kemeny ranking, which seeks to find a consensus ranking that minimizes the total distance from individual rankings, KwikSort provides an efficient way to approximate the optimal Kemeny rank.\\
KwikSort was chosen for its ability to find a solution quickly, despite being an approximate method, because there are no methods with time complexities comparable to those of the annealer that guarantee an optimal solution.
Among the approximate methods, KwikSort appears to be the fastest, hence our decision to employ it.
Looking at the execution time of the classical methods considered (\cref{classicalmethods}),
\begin{figure}[b]
\centering
\includegraphics[scale=0.7]{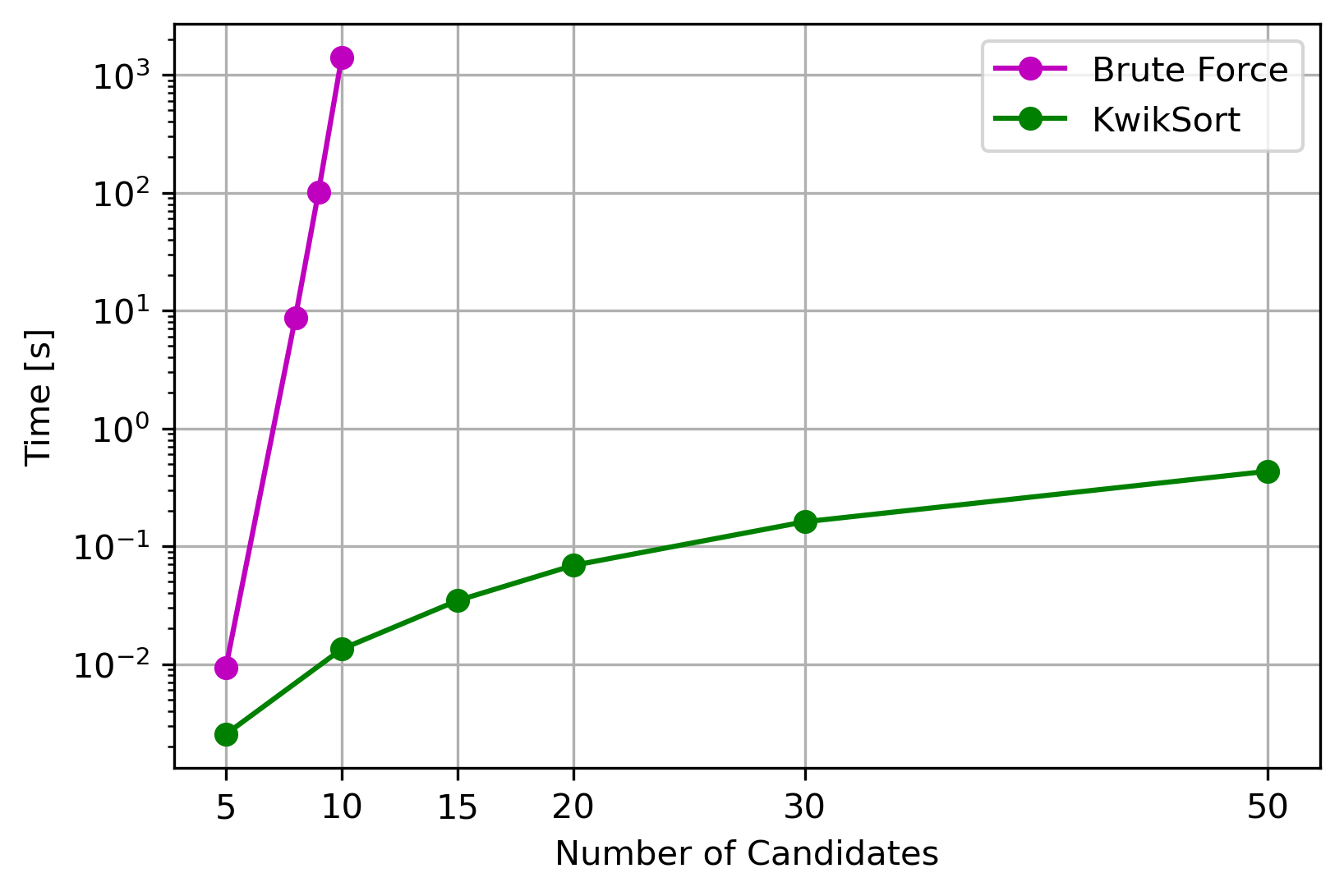}
\caption{Comparison of execution time as a function of the number of candidates in the dataset.}
\label{classicalmethods}
\end{figure}
It is evident that the brute force method, despite its ability to provide the optimal solution, becomes impractical when dealing with datasets of even moderate size. As a result, our analysis focuses exclusively on comparing the \textsf{iterative method} with KwikSort.

\subsubsection{Iterative method vs KwikSort: comparison as exact methods}

As previously mentioned, KwikSort is an approximate algorithm with a randomized output, which means that it returns a different result with each execution. By running KwikSort multiple times, it is possible to find the optimal solution, although this depends on the dataset (see Appendix \ref{ksnosol} for an example where this is not possible).
We aimed to leverage this property to compare KwikSort with the \textsf{iterative method} as exact methods.\\
It is important to note that this comparison may seem biased as KwikSort is an approximate method, but it serves the purpose of evaluating the \textsf{iterative method} performance in terms of speed against the fastest available method.
As a starting point, we computed the execution time for both methods in a single run (see \cref{KSvsA}). 
\begin{figure}[h]
    \centering
    \includegraphics[width=0.7\linewidth]{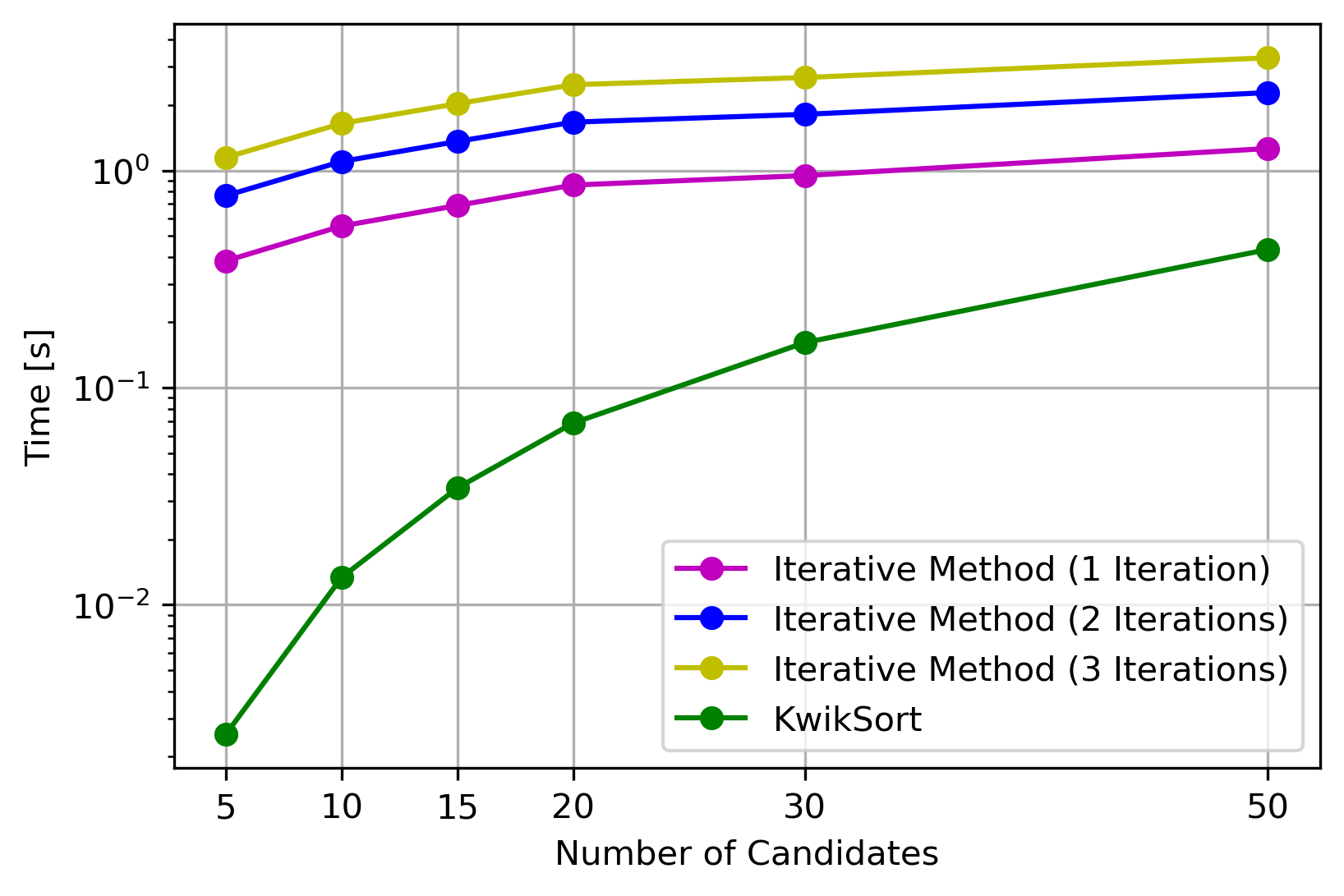}
    \caption{Comparison of execution times for a single run of KwikSort and the iterative method with different number of iterations.}
    \label{KSvsA}
\end{figure}
KwikSort, as expected, shows significantly superior performance in single run-times. We also would like to emphasize how the execution time of the \textsf{iterative method} is composed. Here, a substantial part is attributed to the embedding process, which constitutes approximately 80\% of the total time. Furthermore, as can be noted, the time required for the dataset with 30 candidates is greater than that for the dataset with 50 candidates. This shows how the scaling of the embedding time is not solely dependent on the dataset size but also depends on factors such as the number of cycles or the overall complexity of the dataset.

To compare the two methods as exact methods, we will therefore calculate the average number of attempts required by KwikSort to find a solution at least as good as that of the \textsf{iterative method} (see \cref{ntrials} in Appendix \ref{addfig}). Despite conducting 100,000 trials on the datasets consisting of 30 and 50 candidates each, KwikSort failed to find a solution of comparable quality. This outcome does not definitively rule out the possibility of KwikSort achieving a comparable solution, as exploring all pivot choices is required to exclude this, although it strongly suggests otherwise. Consequently, we will exclude these two datasets from the comparison.
The total time for KwikSort is then calculated as the average time per trial multiplied by the average number of trials:
\begin{equation}
    T_{KS} = \langle N_{trials} \rangle \cdot T_{trial}.
\end{equation}
Updating the timing for KwikSort, we obtain a new version of \cref{KSvsA} where now both methods achieve the same best solution (\cref{fig:timeoptsol}).
\begin{figure}[h]
    \centering
    \includegraphics[width=0.7\linewidth]{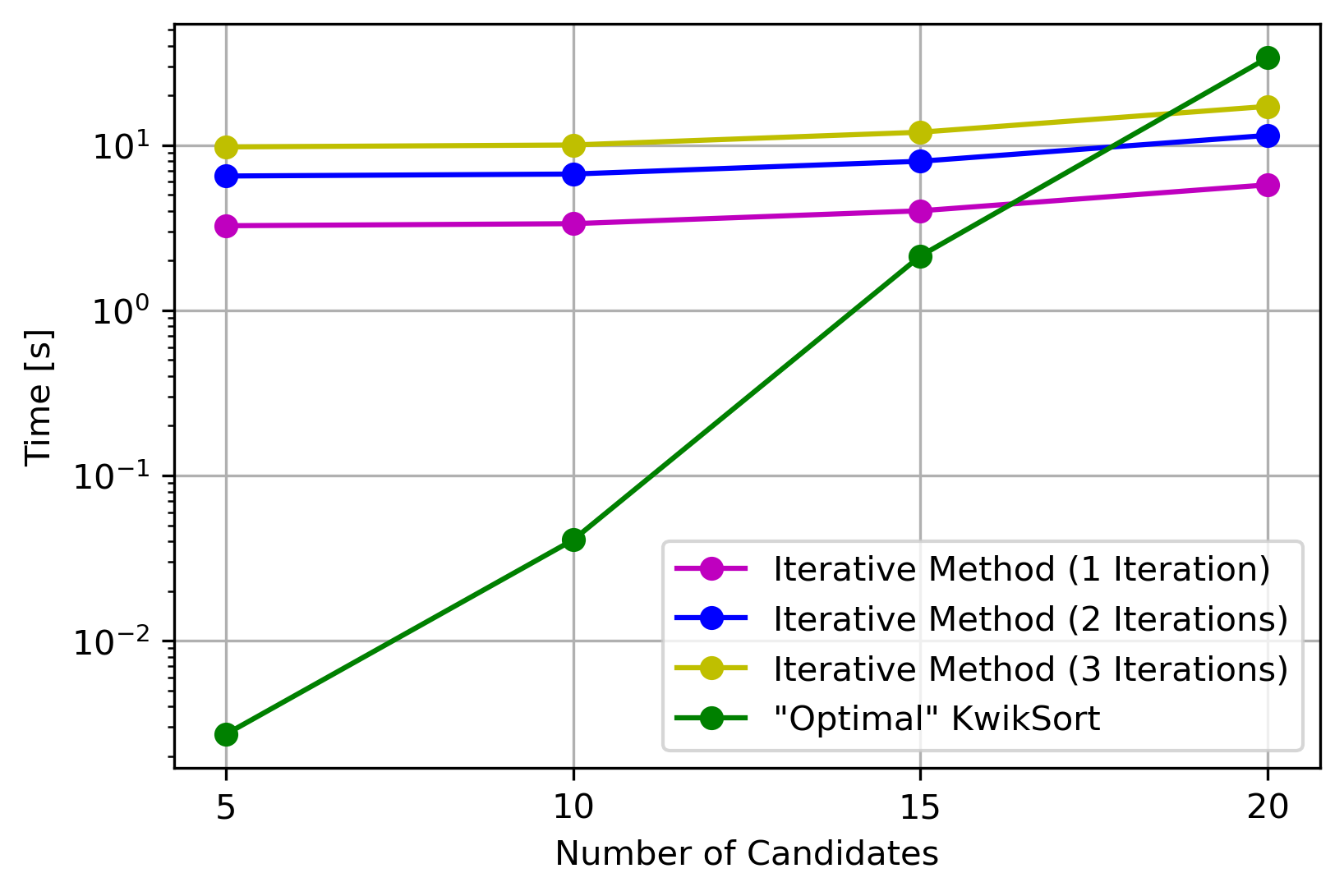}
    \caption{Comparison of execution times required to find the optimal solution for KwikSort and iterative method with different number of iterations.}
    \label{fig:timeoptsol}
\end{figure}
The results show that KwikSort, when used to find the optimal solution, takes longer than the \textsf{iterative method} only in the case of the 20 candidates dataset. However, this marks the beginning of a clear trend, where the difference in execution time between the two methods progressively increases as the dataset size increases. In fact, if we also consider the 30 and 50 candidates datasets, where KwikSort failed to find the optimal solution within the given number of trals, the trend becomes even more pronounced.\\
To further illustrate the performance difference, we compare the Kendall-Tau distances of the outputs of both methods, with KwikSort data obtained from 300 measurements per dataset. Figure \ref{resoconto} shows the growing discrepancy between the distance of the solutions of the two methods as the number of candidates increases.
\begin{figure}[H]
    \centering
    \includegraphics[scale=0.7]{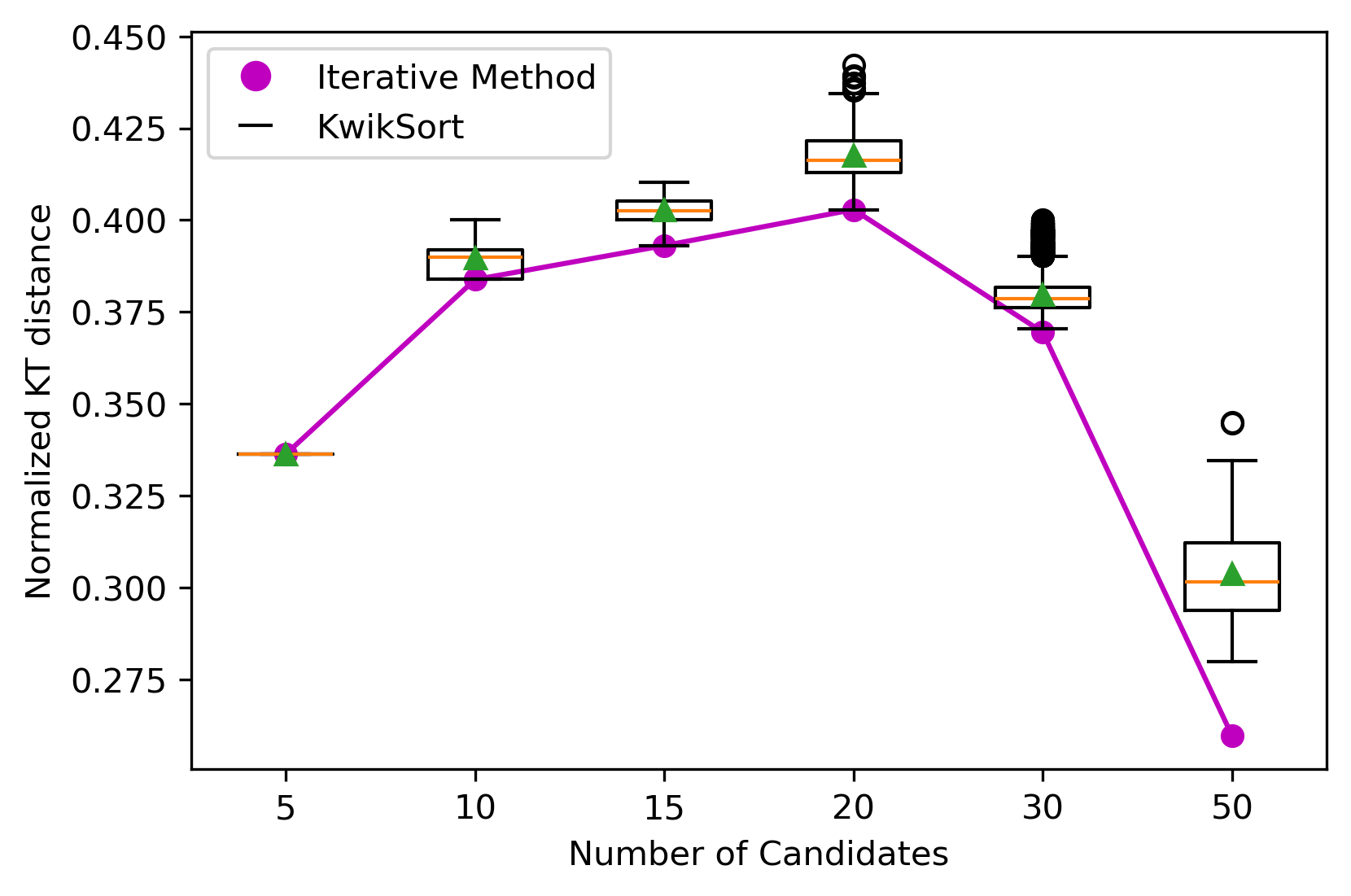}
    \caption{Normalized Kendall-Tau distance comparison between KwikSort (boxplots, 300 measurements) and the iterative method (line).}
    \label{resoconto}
\end{figure}
Although KwikSort excels in individual run-time, the \textsf{iterative method} consistently provides higher-quality solutions.

In summary, KwikSort shows superior speed compared to the \textsf{iterative method} for datasets containing up to 15 candidates, making it the better choice for smaller datasets.
However, as the number of candidates increases, the \textsf{iterative method} becomes more efficient and consistently delivers higher-quality solutions. Considering this, along with the inherent uncertainty in KwikSort's ability to find the optimal solution, the annealer-based method emerges as the preferred choice. In fact, the \textsf{iterative method} consistently delivers the best solutions, and while KwikSort may offer a slight time advantage for small datasets, this marginal benefit becomes irrelevant in light of the uncertainty surrounding its optimality.\\

\subsubsection{Iterative method vs KwikSort: comparison as approximate methods}

In this section, we focus on considering both the annealer and KwikSort as approximate solution methods. As previously mentioned, when tackling highly complex problems or when the number of iterations used is insufficient to address and eliminate all relevant cycles, the annealer's output represents an approximate solution.\\
Specifically, our attention is drawn to the former scenario: we aim to explore situations where the annealer can only be employed as an approximate method due to the dataset's complexity, and see how the quality of the solutions varies in comparison to KwikSort.\\
In this specific context, we are tasked with selecting a stopping criterion to determine when to halt the iterations. As a stopping criterion, we have chosen to halt after four cycle updates as a good compromise between discovering cycles (indicating proximity to the optimal solution) and avoiding an excessive number of penalized cycles in the final iteration. This balance helps maintain a reasonable level of confidence in the output, despite potential errors arising from the problem's complexity.\\
For each evaluation, we then executed four iterations, performing five separate evaluations for each dataset. In each evaluation, the solution is defined as the ranking with the minimal distance among the four iterations.
The datasets employed for this analysis are denoted as \{$N$, $C_{in}$\}, where $N$ represents the number of candidates, and $C_{in}$ denotes the number of cycles found in the initial configuration $\Omega$. The datasets used are the following: \{30, 300\}, \{40, 448\}, \{40, 800\}, \{40, 1124\} and \{50, 400\}.
The first four datasets are drawn from the synthetic dataset set, while the final one is sourced from the simplified synthetic set.

The first thing we had to tackle was a problem regarding two datasets: \{40, 800\} and  \{40, 1124\}. Finding an embedding with that number of cycles was not possible, so we opted to reduce the initial cycles using the strategy \textit{Facilitating embedding discovery} described in \cref{section:mitigation}.
Furthermore, as previously mentioned, it is quite unlikely that each iteration brings us closer to the optimal solution; hence, we track the best solution and keep its distance, regardless of whether it appeared in the first or last iteration.
For KwikSort, we conducted 10000 runs per dataset, collecting the distance for each output.

Of the five datasets used,  we will present plots for only two: one where the \textsf{iterative method} performs better, finding an output with a smaller Kendall Tau distance \{30, 300\}, and one where KwikSort performs better \{40, 1124\}.
First we present the case where the annealer performs better (\cref{fig:30300}).
\begin{figure}[H]
    \centering
    \begin{subfigure}[t]{0.48\textwidth}
        \centering
        \includegraphics[width=\linewidth]{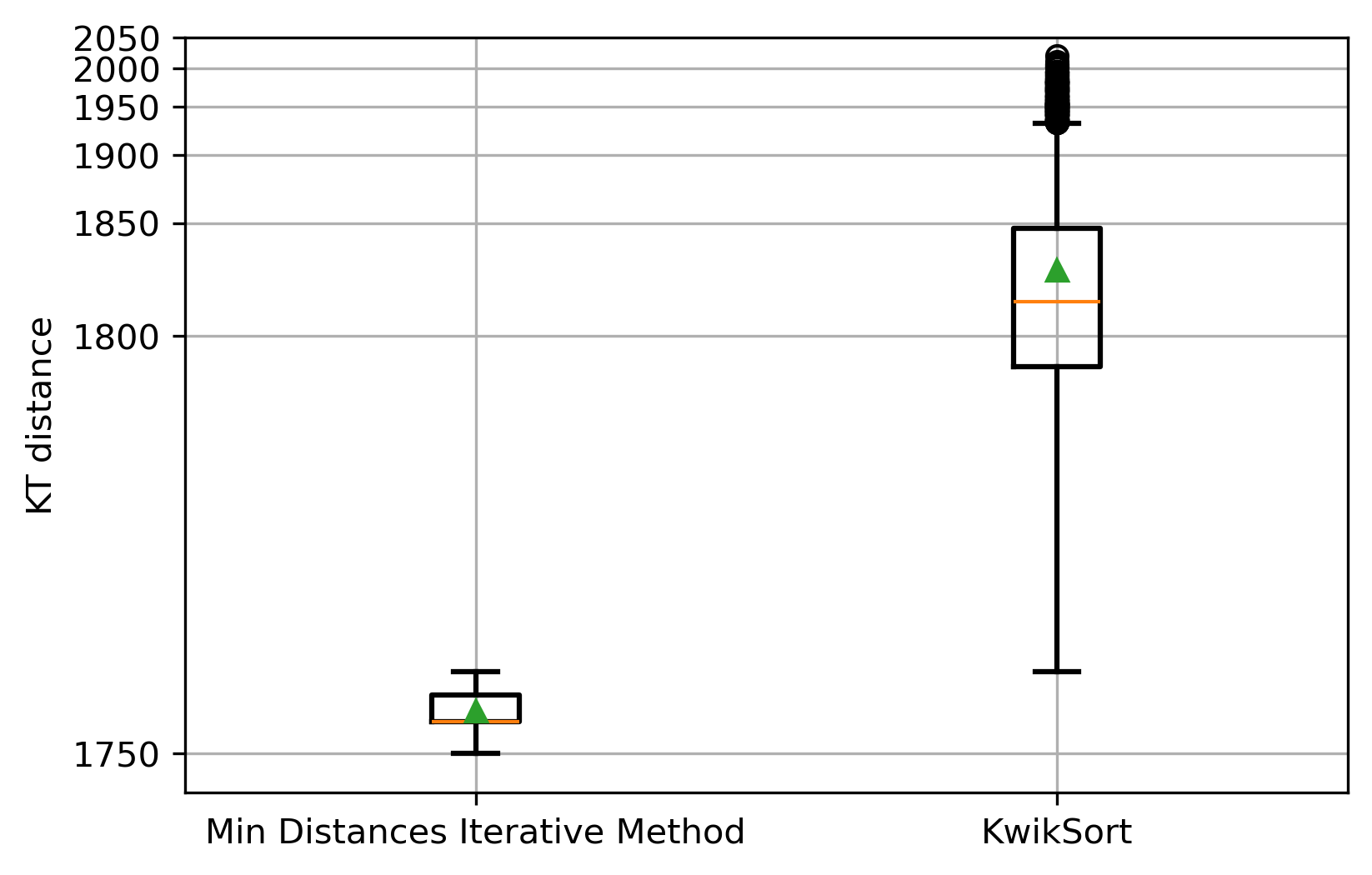}
        \caption{}
        \label{fig:303001}
    \end{subfigure}
    \hfill
    \begin{subfigure}[t]{0.48\textwidth}
        \centering
        \includegraphics[width=\linewidth]{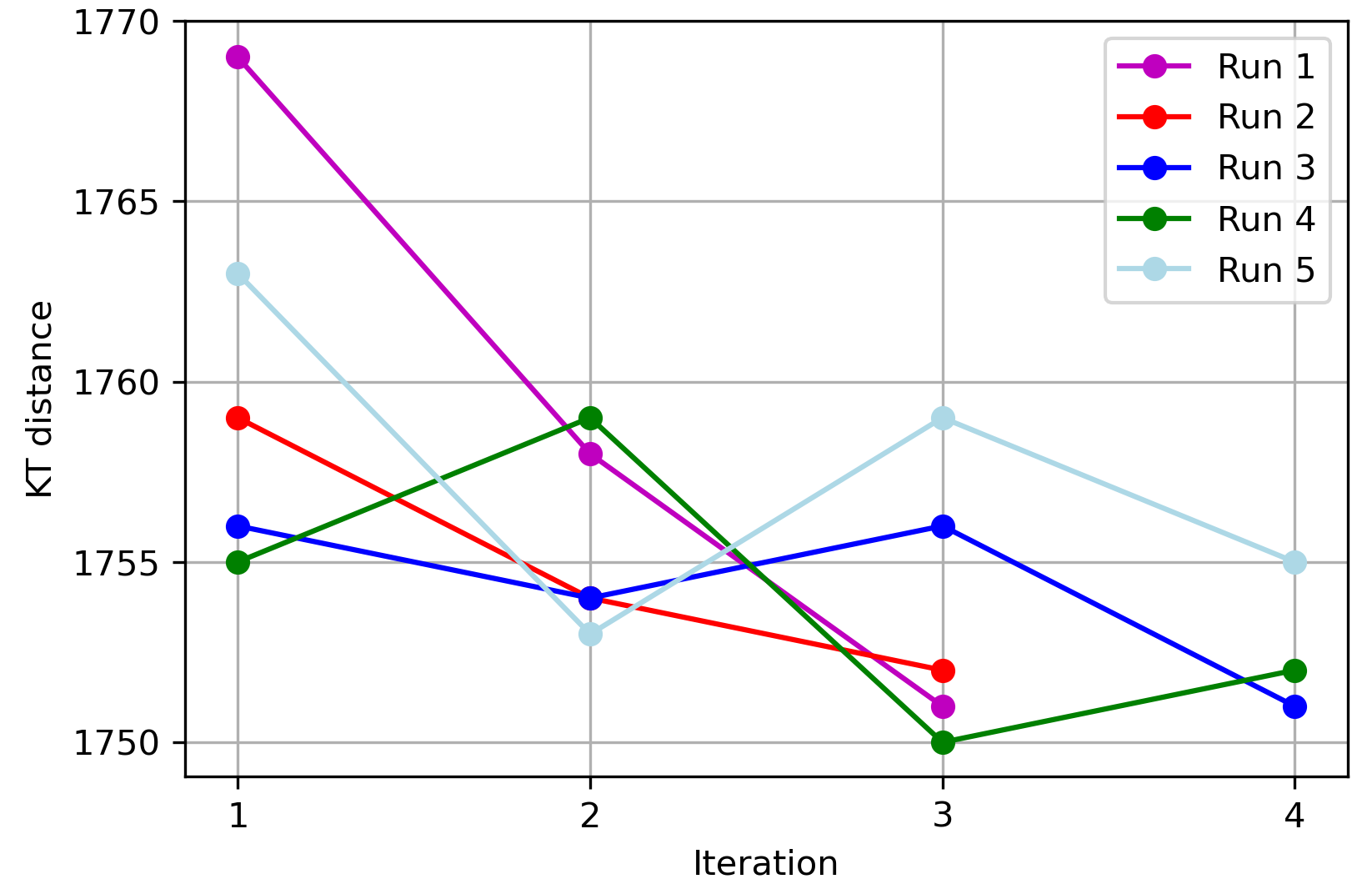}
        \caption{}
        \label{fig:303002}
    \end{subfigure}
    \caption{Comparison of the Kendall Tau (KT) distance between the KwikSort output and the iterative method's minimum output across four iterations on the dataset $\{30, 300\}$ (\ref{fig:303001}). The second plot (\ref{fig:303002}) shows the path of the five runs performed by the iterative method.}
    \label{fig:30300}
\end{figure}
The first observation is that the annealer significantly outperforms KwikSort both on average and in terms of the minimum value found. Another important observation from the second graph (\cref{fig:303002}) is that the first and the second runs converge to an output without cycles in only three iterations. However, it is important to acknowledge that the distances obtained in these cases are greater than the minimum distance found, which highlights once again that the absence of cycles does not necessarily guarantee an optimal solution for challenging datasets.

Now let us move to one case where the \textsf{iterative method} performs worse \{40, 1128\}.
Given the high number of initial cycles in this dataset, we were unable to find an embedding. To address this issue, we proceeded by removing cycles containing pairs that were already present in other cycles. 
In this specific example, we removed cycles containing pairs that appeared more than 6 times in other cycles, reducing the 1128 starting cycles to only 445, which were then used to initiate the annealing process.
\begin{figure}[H]
    \centering
    \begin{subfigure}[t]{0.48\textwidth}
        \centering
        \includegraphics[width=\linewidth]{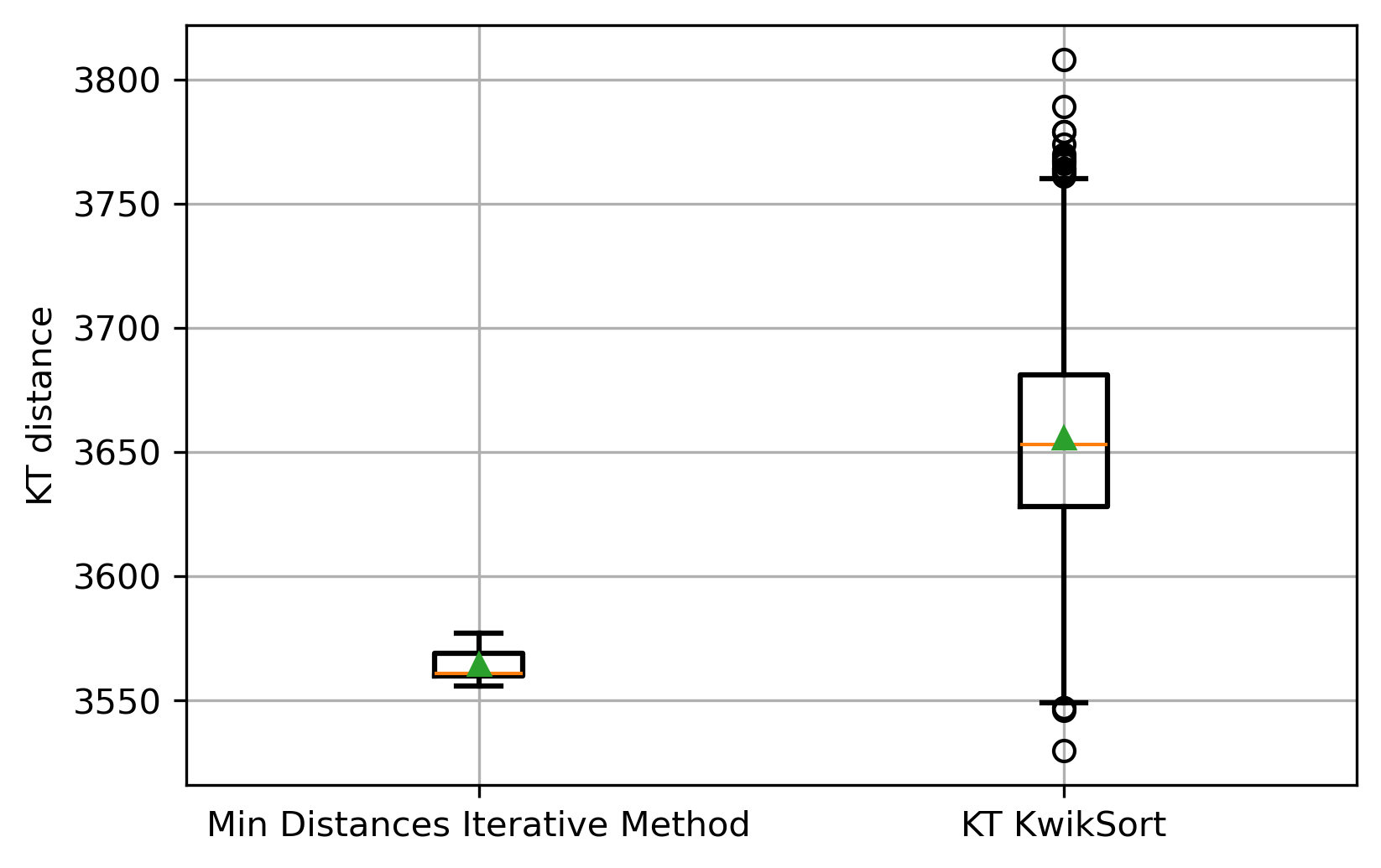}
        \caption{}
        \label{fig:4011281}
    \end{subfigure}
    \hfill
    \begin{subfigure}[t]{0.48\textwidth}
        \centering
        \includegraphics[width=\linewidth]{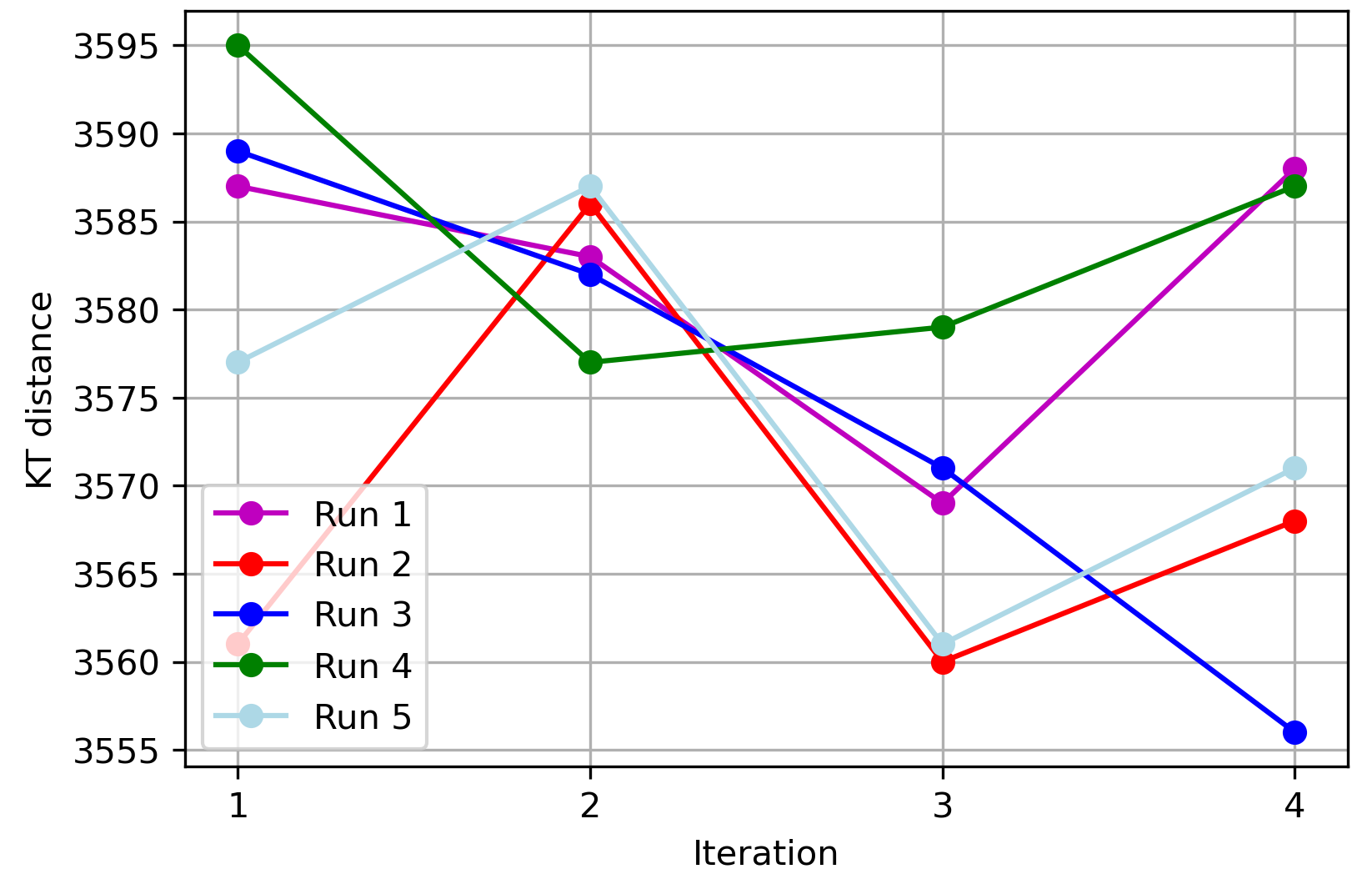}
        \caption{}
        \label{fig:4011282}
    \end{subfigure}
    \caption{Comparison of the Kendall Tau (KT) distance between the KwikSort output and the iterative method's minimum output across four iterations on the dataset $\{40, 1128\}$ (\ref{fig:4011281}). The second plot (\ref{fig:4011282}) shows the path of the five runs performed by the iterative method.}
    \label{fig:401128}
\end{figure}
We first notice (see \cref{fig:401128}) how, in this case, KwikSort occasionally finds better solutions than the \textsf{iterative method}. Thus, in cases like this, with a large number of cycles, KwikSort performs better in the sense that it is able to find solutions with a smaller Kendall-Tau distance. The average distance of KwikSort's outputs is still higher than that found through the \textsf{iterative method}. Therefore, one must consider the desired solution quality when choosing between the two methods.
For this reason, we present the comprehensive data we collected, including all 5 datasets, in \cref{tab:comparison}. This table provides four key metrics to facilitate the decision-making process. Here, $KT_{IM}$ refers to the Kendall-Tau distance of a solution produced by the \textsf{iterative method}, while $KT_{KS}$ denotes the Kendall-Tau distance of a solution produced by KwikSort. The metrics are as follows: the \textsf{iterative method} time $T_{IM}$, which is the time needed by the \textsf{iterative method} to produce an output; $KT_{KS} < \text{Avg}(KT_{IM})$, which is the average time needed by KwikSort to produce a solution better than the average solution of all runs of the \textsf{iterative method} (i.e., the average over the 20 data points); $KT_{KS} < \text{Avg}(\underset{\text{runs}}{\min}(KT_{IM}))$, where instead we consider the average of only the minimum Kendall-Tau distances obtained in each run and $KT_{KS} < \min(\underset{\text{runs}}{\min}(KT_{IM}))$, which represents the average time required by KwikSort to find a solution better than the best solution found by the \textsf{iterative method} in these 5 runs.
\begin{table}[h]
    \centering
    \resizebox{\textwidth}{!}{%
    \begin{tabular}{|c|c|c|c|c|}
    \hline
    Dataset & $T_{IM}$ [s] & $KT_{KS} < \text{Avg}(KT_{IM})$ [s] & $KT_{KS} < \text{Avg}\bigg(\underset{\text{runs}}{\min}(KT_{IM})\bigg)$ [s] & $KT_{KS} < \min\bigg(\underset{\text{runs}}{\min}(KT_{IM})\bigg)$ [s] \\
        \hline
    \{30, 300\}  &  103.57 & 1015.5 & - & -\\
        \hline
    \{40, 448\}  &  134.45   & 1980 & - & -\\
        \hline
    \{40, 800\}  &  216.83   & 220 & - & -\\
        \hline
    \{40, 1128\} &  214.55   & 26.92 & 70.525 & 157.24 \\
        \hline
    \{50, 400\}  &   148.4  & 13.87 & 195.2 & 1366.41 \\
        \hline
    \end{tabular}%
    }
    \caption{This table provides concise performance metrics to evaluate the efficiency and effectiveness of the iterative method and KwikSort, aiding in the selection of the most suitable approach.
    The `-' sign indicates that a numerical value corresponding to that cell does not exist.}
    \label{tab:comparison}
\end{table}
Upon reviewing the provided table, users can gain insights to help with selecting the most suitable approximate method based on their desired precision for the solution.
It is important to bear in mind the obsolescence of the hardware used, which means that these data may be significantly smaller when using a modern computer.

\section{Related work}
During our investigation, a related study was published proposing a foundational model similar to the one we have introduced as our base model. The method in question, detailed in Section 4.3 of the referenced paper \cite{krbase}, introduces a mathematical framework that aligns closely with our proposed base model. This coincidence highlights the relevance of the research within this domain.

However, it is crucial to emphasize that while the aforementioned study lays a significant groundwork by exploring the base model, it does not leverage the binary representation and does not extend its analysis to the \textsf{iterative method} that forms the cornerstone of our contribution. Our research diverges and advances significantly from this point, offering a novel approach that we propose as a key enhancement to the base model. We demonstrate that this iterative methodology not only complements but substantially improves the problem-solving capabilities of the base model.
The benefits of our \textsf{iterative method} are manifold, providing enhanced accuracy, efficiency, and applicability to a wider range of problems. Through rigorous testing and analysis, we have showcased how our iterative approach addresses the limitations of the base model, offering a more robust solution that could significantly contribute to the field.

In conclusion, although our work and the referenced study independently arrived at similar base models, our use of the binary representation for the \textsf{iterative method} represents a significant advancement. 

\section{Conclusions and outlook}

In this work, we proposed a novel approach to computing the Kemeny ranking using a quantum annealer, along with various methods to enhance its performance. The first improvement, the \textsf{iterative method}, leverages the pairwise preference representation to penalize only the cycles relevant to the problem. The second improvement, the pair removal technique, temporarily excludes some pairs from the problem, solves this simplified version, and then reintroduces the excluded pairs by inferring their values through the transitivity property.

The base model already outperforms the benchmarked model ($n^2$-representation model), but the real leap in performance comes from the \textsf{iterative method}, which allows addressing much more complex datasets. When used in conjunction with the pair removal method, performance improves even further, although marginally. \\
We compared the \textsf{iterative method} with KwikSort, both as exact methods and as approximate methods. For small datasets with less than 15 candidates, KwikSort performs better as an exact method, finding the solution faster than the \textsf{iterative method}. As an approximate method, KwikSort outperforms the \textsf{iterative method} in terms of finding a better solution as the dataset complexity increases, despite requiring more time. Users can thus choose the most suitable method based on their specific needs and the complexity of the dataset.

Future work involves leveraging Kemeny ranking in machine learning as both an ensemble method and for feature reduction, alongside exploring a modified version of the \textsf{iterative method}. This modified method would take an approximate solution as input and output a solution with a Kendall-Tau distance less than or equal to that of the input rank, thereby improving the approximate solution. \\
Concerning the usage as an ensemble method, instead of producing a single best class from each classification algorithm, the proposal is to generate a ranking of classes. These ranks can then be combined using the ranking model to derive a final, refined ranking, resulting in a more sophisticated and comprehensive classification process. Further refinement can be achieved by assigning weights that progressively diminish with each rank position, emphasizing the primary position's influence over subsequent ones and enhancing the precision of the classification outcome. Systematically exploring this weight space allows for the identification of optimal weights that yield the best classification results.

In conclusion, the introduction of this new approach significantly enhances the application of Kemeny ranking to large-scale problems, which was previously unthinkable due to the prohibitive computational challenges. This approach not only extends the applicability of Kemeny ranking but also paves the way for further research in this field. The potential implications of this advancement are vast, promising to unlock new opportunities and avenues for exploration in the realm of Kemeny ranking.

\section{Acknowledgments}
This work was funded by the National Recovery and Resilience Plan (PNRR), under Mission 4 "Education and Research"-Component 2, Investment 1.1 "Fund for the National Research Program, Projects of Relevant National Interest (PRIN)". Call: PRIN 2022 (D.D. 104/22), project title: ENGineering INtElligent Systems around intelligent agent technologies, CUP: E53D23007970006. The project was carried out at the department of information engineering and computer science (DISI) of the University of Trento. The authors gratefully acknowledge the Italian Ministry of University and Research (MUR) for supporting this research through the PRIN 2022 funding program. 

They also extend their gratitude to the J\"ulich Supercomputing Center (\url{https://www.fz-juelich.de/ias/jsc}) for funding this project by providing computing time on the D-Wave Advantage\texttrademark{} System JUPSI through the J\"ulich UNified Infrastructure for Quantum computing (JUNIQ).

DP is member of the ``Gruppo Nazionale per la Fisica Matematica (GNFM)'' of the ``Istituto Nazionale di Alta Matematica ``Francesco Severi'' (INdAM)''.

\begin{appendices}

\section{Proof: Maximum penalty weight for datasets with an odd number of votes} \label{secA1}

Let the total number of votes be denoted as $|\Pi|$, with $|\Pi|$ odd, and let us consider the triplet (a, b, c) as a representative example. For a cycle to form, we require either $b_{ab} < 0$, $b_{ac} > 0$, $b_{bc} < 0$ or $b_{ab} > 0$, $b_{ac} < 0$, $b_{bc} > 0$. Without loss of generality, let us assume the first case; the proof for the second case is symmetric.\\
For a penalty coefficient $P > |\Pi|$ to be necessary, one of the biases must be equal to $\pm|\Pi|$. We fix one of the biases, let us say $b_{ab} = (w_{ab} - w_{ba}) = -|\Pi|$.
This means that we need to satisfy:
\begin{equation}
  \begin{cases}
    b_{ab}=-|\Pi|\\
    b_{ac}>0\\
    b_{bc}<0
\end{cases}  
\end{equation}
Using the definition of bias and the fact that $|\Pi|$ is odd, the last two requirements can be decomposed as:

\[
\begin{tikzcd}[column sep=small]
b_{ac} > 0 \arrow[r,Rightarrow] & \begin{cases}
w_{ca} \geq (|\Pi|+1)/2\\
w_{ac} \leq (|\Pi|-1)/2
\end{cases} \\
b_{bc} < 0 \arrow[r,Rightarrow] & \begin{cases}
w_{cb} \leq (|\Pi|-1)/2 \\
w_{bc} \geq (|\Pi|+1)/2
\end{cases}
\end{tikzcd}
\]
Now we look at what are the possible lists present in the dataset.
Since $a \prec b$ is fixed for each list (due to $b_{ab}=-|\Pi|$), there are only three possible orders in the dataset for these three candidates:
\begin{enumerate}
\item $a \prec b \prec c$
\item $a \prec c \prec b$
\item $c \prec a \prec b$
\end{enumerate}
The contribution of each list to the $w$ matrix is:
\begin{enumerate}
\item $a \prec b \prec c$ $\xrightarrow{}$ $w_{ab} \pluseq 1$, $w_{ac} \pluseq 1$, $w_{bc} \pluseq 1$
\item $a \prec c \prec b$ $\xrightarrow{}$ $w_{ab} \pluseq 1$, $w_{ac} \pluseq 1$, $w_{cb} \pluseq 1$
\item $c \prec a \prec b$ $\xrightarrow{}$ $w_{ab} \pluseq 1$, $w_{ca} \pluseq 1$, $w_{cb} \pluseq 1$
\end{enumerate} Let $l_1$, $l_2$, and $l_3$ represent the coefficients of each list, indicating how many times they appear in the dataset. Then, $l_1 + l_2 + l_3 = |\Pi|$, i.e., the sum of the coefficients equals the total number of votes.\\
We rewrite the previous system of equations using $l_1, l_2$ and $l_3$:
\begin{equation}
\begin{cases}
l_1 + l_2 + l_3 = |\Pi|\\
l_1 + l_2 \leq \frac{|\Pi|-1}{2}\\
l_3 \geq \frac{|\Pi|+1}{2}\\
l_1 \geq \frac{|\Pi|+1}{2}\\
l_2 + l_3 \leq \frac{|\Pi|-1}{2}\\
\end{cases}
\end{equation} 
However, these conditions cannot be satisfied simultaneously. Therefore, a cycle cannot form with one of the biases equal to $|\Pi|$. The next possible value for the bias is $|\Pi|-2$. Using similar reasoning, we can show that a cycle is possible in this case. Hence, we have demonstrated that, as a penalty coefficient for a dataset with odd number of votes, $P > |\Pi|-2$ is sufficient.

\section{Transformation Table}\label{tt}
In the reconstruction of a pair process we considered only the case $a<b<x$. However, two additional scenarios are possible: $a < x < b$ and $x < a < b$. To handle these cases, one can use the following table, generated through the swap $x_{ij}\xrightarrow{}(1-x_{ji})$ and requiring that i$<$j, to compute the ``vote" coefficient $R_x$:
\begin{table}[ht]
\centering

\begin{tabular}{|c|c|c|c|}

\hline
\textbf{$a<b<x$} & \textbf{$a<x<b$} & \textbf{$x<a<b$} & \textbf{$R_{x}$} \\
\hline
(0,0) & (0,1) & (1,1) & 0\\
\hline
(0,1) & (0,0) & (1,0) & -1\\
\hline
(1,0) & (1,1) & (0,1) & 1\\
\hline
(1,1) & (1,0) & (0,0) & 0 \\
\hline
\end{tabular}
\caption{Transformation table for computing the coefficient $R_X$ in the scenarios $a<b<x$, $a<x<b$ and $x<a<b$.}
\label{tab:AXB}
\end{table}

\section{Pseudocode}\label{appC}

This section presents the pseudocode for the proposed methods. For datasets with an even number of votes, the pseudocode of the iterative method is the same as for the odd case, except that $P_{ijk}$ is initialized to $\epsilon$ rather than $1 + \epsilon$.\\
The algorithms below outline the base model, iterative method, and pair removal strategy.

\begin{algorithm}
\caption{Base Model}
\label{alg:base_model}
\begin{algorithmic}[1]
\Require Dataset of rankings $\Pi$
\State Compute the pairwise comparison matrix $W$ from $\Pi$ 
\State Compute the bias matrix $B$, where $b_{ij} = w_{ji} - w_{ij}$
\State Define the QUBO cost function:
        \begin{equation*}
            C(X) = \sum_{i, j \ : \ i < j}^n b_{ij} x_{ij} + P \sum_{i, j, k \ : \ i < j < k}^n {c}_{ijk}(X),
        \end{equation*}
        where:
        \begin{equation*}
            {c}_{ijk}(X) = x_{ik} + x_{ij} x_{jk} - x_{ij} x_{ik} - x_{jk} x_{ik}.
        \end{equation*}
\State Encode the problem into the quantum annealer
\State Execute quantum annealing to minimize the cost function
\State Retrieve the binary output matrix $X$
\State For each candidate $c_i$, compute its score:
        \begin{equation*}
           V_{c_i} = \sum_{j=i+1}^n x_{ij} + \sum_{j=1}^{i-1} (1 - x_{ji}) 
        \end{equation*}

\State Reconstruct the optimal ranking $\pi^*$ by sorting candidates in decreasing order of their scores\\
\Return $\pi^*$
\end{algorithmic}
\end{algorithm}

\begin{algorithm}
\caption{Iterative method for datasets with odd number of votes}
\label{alg:iterative_method}
\begin{algorithmic}[1]
\Require Dataset of rankings $\Pi$
\State Compute the pairwise comparison matrix $W$ and bias matrix $B$
\State Compute the $\Omega$ matrix
        \hfill $\triangleright$ $\Omega$ captures majority preferences between candidates
\State Initialize $C_n = \emptyset$ to store cycles and their penalties
\State Initialize $P_{ijk} = 1 + \epsilon$, where $\epsilon$ is a small positive value, for all cycles $(i, j, k)$ satisfying:
        \begin{equation*}
        (\omega_{ij}, \omega_{ik}, \omega_{jk}) = 
        \begin{cases}
        (0, 1, 0), \\
        (1, 0, 1).
        \end{cases}
        \end{equation*}
\State Add all initial cycles $(i, j, k)$ satisfying these conditions to $C_n$ with $P_{ijk} = 1 + \epsilon$
\State Define the QUBO cost function:
        \begin{equation*}
        C(X) = \sum_{i, j \ : \ i < j} b_{ij} x_{ij} + \sum_{i, j, k \in C_n} P_{ijk}{c}_{ijk}(X),
        \end{equation*}
        where:
        \begin{equation*}
        {c}_{ijk}(X) = x_{ik} + x_{ij} x_{jk} - x_{ij} x_{ik} - x_{jk} x_{ik}.
        \end{equation*}
\Repeat
    \State Execute quantum annealing to minimize the cost function
    \State Retrieve the binary output matrix $X$
    \State Detect cycles in $X$:
        \For{each detected cycle $(i, j, k)$}
            \If{$(i, j, k) \notin C_n$}
                \State Add $(i, j, k)$ to $C_n$
                \State Set $P_{ijk} = 1 + \epsilon$
            \Else
                \State Update $P_{ijk} \gets P_{ijk} + 2$
            \EndIf
        \EndFor
\Until{No cycles are detected in $X$}
\State Reconstruct the optimal ranking $\pi^*$ from the final $X$ using candidate scores:
        \begin{equation*}
        V_{c_i} = \sum_{j=i+1}^n x_{ij} + \sum_{j=1}^{i-1} (1 - x_{ji})
        \end{equation*}
\State Sort candidates in decreasing order of their scores to obtain the optimal ranking $\pi^*$\\
\Return $\pi^*$
\end{algorithmic}
\end{algorithm}

\begin{algorithm}
\caption{Pair Removal for the Base Model}
\label{alg:pair_removal}
\begin{algorithmic}[1]
\Require Dataset of rankings $\Pi$
\State Compute the pairwise comparison matrix $W$ and the bias matrix $B$
\State Compute the $\Omega$ matrix
\If{using $PR\Omega$} 
    \State Reconstruct a ranking from $\Omega$ by computing the score for each candidate:
        \begin{equation*}
           V_{c_i} = \sum_{j=i+1}^n x_{ij} + \sum_{j=1}^{i-1} (1 - x_{ji}) 
        \end{equation*}
        and sort candidates in decreasing order of $V_{c_i}$
    \State Remove pairs corresponding to candidates that are far apart in the reconstructed ranking
\ElsIf{using $PRHB$}
    \State Remove pairs with high bias values (e.g., $|b_{ij}|$ above a defined threshold)
\EndIf
\State Modify the cost function by excluding removed pairs:
        \begin{equation*}
        C_{\text{pr}}(X) = \sum_{(i, j) \notin R} b_{ij} x_{ij} 
        + P\sum_{(i, j, k) \notin R}{c}_{ijk},
        \end{equation*}
\Repeat
    \State Execute quantum annealing to minimize $C_{\text{pr}}(X)$
    \State Retrieve the reduced binary output matrix $X$, setting $x_{ij} = 0.5$ for removed pairs $(i, j) \in R$
    \State Infer values for removed pairs:
        \For{each removed pair $(i, j) \in R$}
            \State Compute the score for the pair:
                \begin{equation*}
                R_{ij} = \sum_{k \notin \{i, j\}} \text{vote}_{k}(i, j),
                \end{equation*}
                where:
                \begin{equation*}
                \text{vote}_{k}(i, j) = 
                \begin{cases}
                +1 & \text{if } x_{ik} = 1 \text{ and } x_{jk} = 0, \\
                -1 & \text{if } x_{ik} = 0 \text{ and } x_{jk} = 1, \\
                0  & \text{otherwise.}
                \end{cases}
                \end{equation*}
            \State Assign:
                \begin{equation*}
                x_{ij} = 
                \begin{cases}
                1 & \text{if } R_{ij} > 0, \\
                0 & \text{if } R_{ij} < 0, \\
                0.5 & \text{if } R_{ij} = 0 \text{ (inference cannot be made).}
                \end{cases}
                \end{equation*}
        \EndFor
    \If{any pair inference fails}
        \State Restart the process:
            \begin{itemize}
                \item Remove the unresolved pair $(i, j)$ from $R$.
                \item Recompute the cost function $C_{\text{pr}}(X)$ with the updated $R$.
            \end{itemize}
    \EndIf
\Until{all pairs are successfully inferred or no further resolution is possible}
\State Reconstruct the ranking $\pi$ from the final $X$ using candidate scores:
        \begin{equation*}
        V_{c_i} = \sum_{j=i+1}^n x_{ij} + \sum_{j=1}^{i-1} (1 - x_{ji})
        \end{equation*}
\State Sort the candidates in decreasing order of their scores to obtain $\pi$\\
\Return $\pi$
\end{algorithmic}
\end{algorithm}
\clearpage

\section{Results: base model performance for different values of the penalty coefficient}
\label{pu2}
In this section we want to present the data proving that the base model's performance is significantly influenced by the value of the penalty coefficient. 

In the context of the dataset featuring 10 candidates, an analysis of accuracy and number of occurrences reveals that a value of $P=2$ effectively removes all cycles, showcasing the fact that often the required value for the penalty coefficient is much smaller than the theoretical threshold (see \cref{fig:P2}).

\begin{figure}[H]
    \centering
    \begin{subfigure}{0.48\textwidth}
        \centering
        \includegraphics[width=\linewidth]{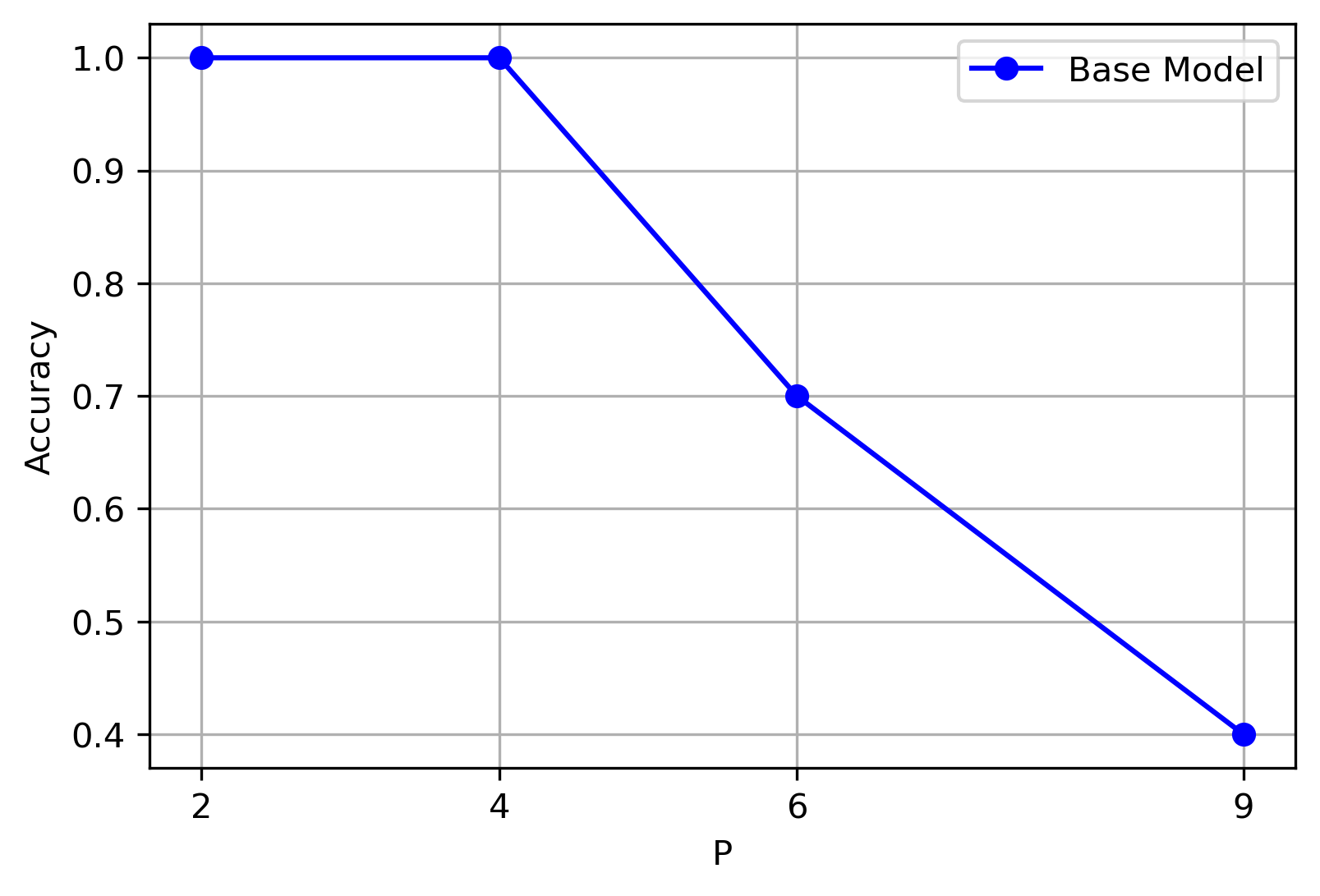}
        \label{f:pp3}
        \caption{}
    \end{subfigure}
    \hfill
    \begin{subfigure}{0.48\textwidth}
        \centering
        \includegraphics[width=\linewidth]{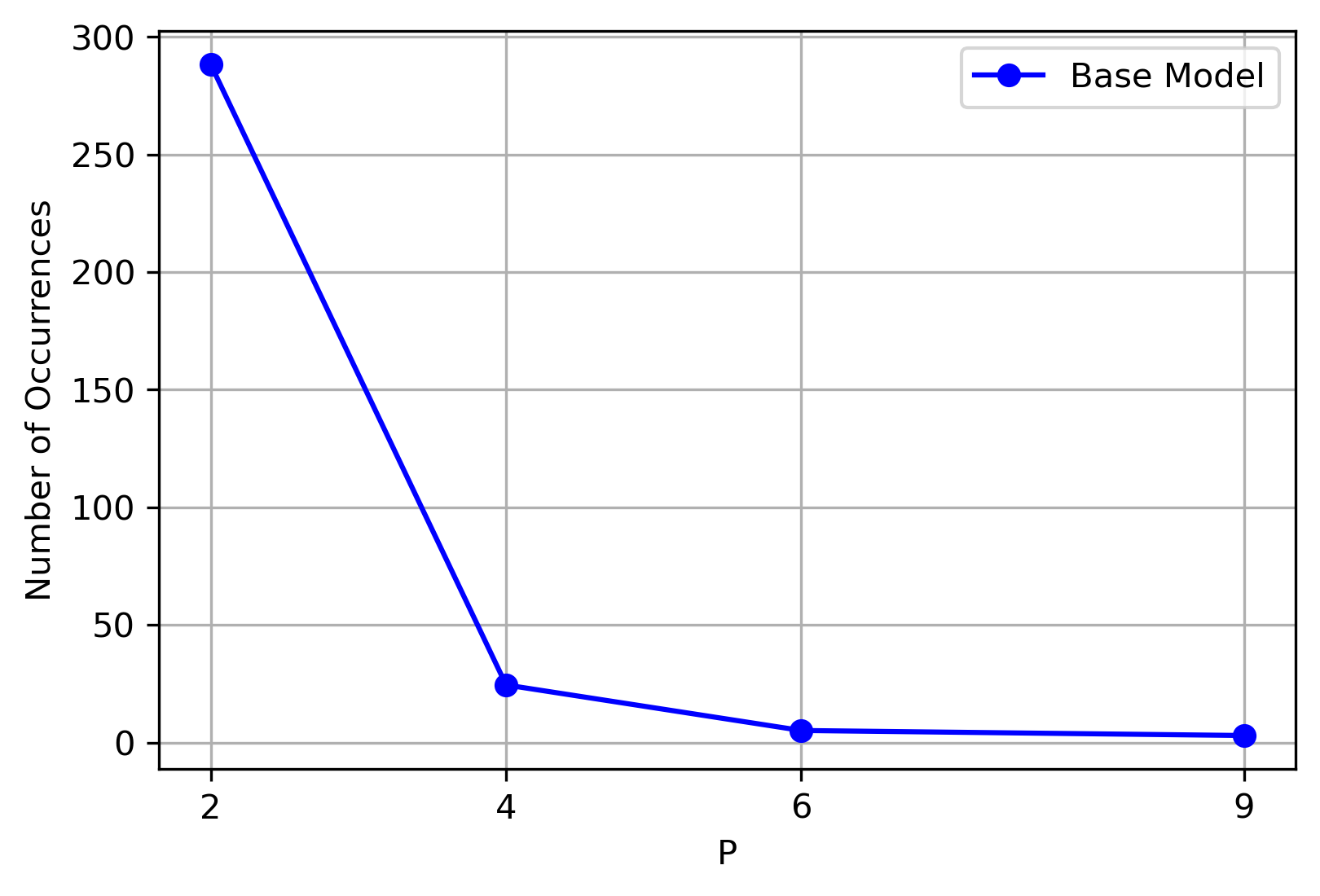}
        \label{f:pp4}
        \caption{}
    \end{subfigure}
    \caption{Behaviour of the accuracy (\cref{f:pp3}) and the number of occurrences (n\_occ) (\cref{f:pp4}) for the base model with respect to the penalty coefficient $P$ for the dataset containing 10 candidates.}
    \label{fig:P2}
\end{figure}
The results presented demonstrate the influence of $P$ on the model's performance. It is observed that as the value of $P$ increases, both the accuracy and the number of occurrences of the correct solution decrease. This suggests the strategic use of the pairwise preference representation to keep the value of $P$ as small as possible for each cycle and reduce the number of penalized cycles to only those that are truly necessary (\textsf{iterative method}). By adopting this approach, the efficiency can be enhanced through targeted cycle penalization, streamlining the computational process and yielding more accurate solutions.

\section{Example of an unsolvable dataset for KwikSort}\label{ksnosol}
Kwiksort employed as ``exact method", while effetive in many scenarios, can fail to converge to a solution for certain datasets.
Due to its inherent randomness, KwikSort relies on random pivot selections and partitioning, which can lead to suboptimal results. Even if all possible choices that the algorithm can explore were considered, the algorithm may still fail to identify the optimal solution. This occurs because the random nature of the pivot selection does not guarantee the exploration of the most efficient path, potentially leading to situations where the algorithm never reaches the best ordering, regardless of the available choices.

\begin{table}[!h]
    \centering
    \resizebox{\textwidth}{!}{ 
        \begin{tabular}{|c|c|c|c|c|c|c|c|c|c|c|}
            \hline
            \textbf{Vote 1} & \textbf{Vote 2} & \textbf{Vote 3} & \textbf{Vote 4} & \textbf{Vote 5} & \textbf{Vote 6} & \textbf{Vote 7} & \textbf{Vote 8} & \textbf{Vote 9} & \textbf{Vote 10} & \textbf{Vote 11} \\
            \hline
            0 & 2 & 0 & 3 & 2 & 2 & 0 & 3 & 1 & 2 & 4 \\
            2 & 0 & 1 & 0 & 1 & 3 & 1 & 4 & 4 & 3 & 3 \\
            1 & 4 & 4 & 2 & 3 & 0 & 2 & 0 & 3 & 0 & 2 \\
            3 & 1 & 2 & 1 & 0 & 1 & 4 & 2 & 0 & 4 & 0 \\
            4 & 3 & 3 & 4 & 4 & 4 & 3 & 1 & 2 & 1 & 1 \\
            \hline
        \end{tabular}
    }
    \caption{An example of an unsolvable dataset for KwikSort. Each column represents a vote, and each row indicates the rank of a candidate (0 being the highest and 4 the lowest).}
    \label{tab:my_label}
\end{table}

\section{Results of partial and weighted lists}\label{pawl}
To handle partial and weighted lists, we need to adapt both the distance measure and the way we construct the weight matrix $w_{ij}$, specifically for cases where some items are missing. 
Let $\{C\}$ represent the set of all candidates present in the dataset, with $n=|\{C\}|$ being the total number of candidates.
For partial lists, the term $(c_j \prec_{\pi_k} c_i)$ is defined as 1 if both candidates $c_i$ and $c_j$ are present in the list and $c_j$ is ranked higher than $c_i$, otherwise it is set to 0. In the case of $k$-top lists, if one candidate is present and the other is missing, the present candidate is preferred. However, if both are missing, no contribution is made to the pairwise preference matrix $W$.
We can also introduce a weight term $\alpha_{ij}$ to handle weighted lists (position weight), allowing different importance for each pair of candidates and a term $\beta_k$ to consider weighted lists. The element of the weight matrix $w_{ij}$ is then:

\begin{equation}
    w_{ij}=\sum_{k}^n \beta_k\alpha(i,j)(c_j \prec_{\pi_k} c_i)
\end{equation}
The same concept extends to a modified version of the Kendall-Tau distance:
\begin{equation}
   KT_{gen}(\pi_1, \pi_2) = \sum_{k}^n \beta_k\alpha(i,j)\cdot KT(i,j) 
\end{equation}
Here, $KT(i,j)$ represents the contribution to the Kendall-Tau distance for the pair $(i, j)$. The key points are:
\begin{itemize}
    \item For partial lists, both candidates $i$ and $j$ must be present in both lists for comparison.
    \item For $k$-top lists, the contribution is included if at least one candidate is present, and the present one is preferred over the missing one.
    \item $\alpha(i,j)$ serves as the weight for the pair in both the sum and distance.
    \item $\beta_k$ is the list weight
\end{itemize}
By handling missing items this way, we extend the standard Kendall-Tau distance and weight matrix $w_{ij}$ to effectively handle both partial and $k$-top lists.
The formula for the generalized Kendall-Tau distance can be expressed as:
\begin{equation}
   C_f = \sum_{i<j} b_{ij}x_{ij} + \alpha_{ij}w_{ij} =\sum_{i<j} \alpha_{ij}(w_{ij} - w_{ji})x_{ij} +\alpha_{ij} w_{ji}=\sum_{i<j}\alpha_{ij}(w_{ij}x_{ij}+w_{ji}(1-x_{ij}) 
\end{equation}
If $x_{ij} = 1$, then $C_f = \alpha_{ij} w_{ij}$, which means candidate $j$ is preferred to candidate $i$ weighted by $\alpha_{ij}$.
If $x_{ij}=0$, meaning candidate j is preferred to candidate i, 
This formulation respects the definition of generalized Kendall-Tau distance.
As weight values, we can choose:

\begin{itemize}
    \item \textbf{Position-based weights}: Items ranked higher are given more significance. The weight is defined as:
    \begin{equation*}
        w(i,j) = \frac{1}{(i+j)^p}
    \end{equation*}

    where $i$ and $j$ are positions of the items, and $p$ controls how the weight decreases.
    
    \item \textbf{Distance-based weights}: Larger rank discrepancies are penalized more. The weight is based on the absolute difference between ranks:
    \begin{equation*}
        w(i,j) = |i-j|
    \end{equation*}
\end{itemize}

We empirically computed the accuracy of the outputs from the annealer compared to the optimal solutions. For all three cases (partial lists, $k$-top lists, and weighted lists), we achieved an accuracy of 1. For weighted lists, we used datasets with different numbers of candidates, while for partial and $k$-top lists, we used one dataset with 9 candidates and randomly removed elements to simulate cases with different list lengths.

The results demonstrate that the annealer can consistently find the optimal solution with 100\% accuracy across partial, $k$-top, and weighted list cases. The experiments show that removing elements from the list (partial lists) or considering only the top candidates (k-top lists) does not affect the annealer's ability to maintain optimality. This indicates the robustness of the generalized Kendall-Tau distance when extended to handle missing items and weighted pairs. Both position-based and distance-based weighting schemes produced similar accuracy, indicating that the method is not sensitive to the choice of weight, further supporting its generalizability to different scenarios.
The empirical results achieved perfect accuracy (1) across the tested cases of partial, $k$-top, and weighted lists. This highlights the robustness of the generalized Kendall-Tau distance when adapted to handle missing items and weighted lists. The experiments involved varying the number of candidates and random removal of elements to create partial lists, simulating different list lengths. Despite these variations, the annealer consistently found the optimal solution, showing that the method is highly adaptable and reliable. Additionally, both position-based and distance-based weighting schemes performed similarly, suggesting that the algorithm can generalize well across different types of weighting functions.

Essentially, what we have demonstrated is that the quantum annealer is capable of solving not only standard ranking problems but also more complex scenarios such as partial lists, $k$-top lists, and weighted lists. The empirical results, with perfect accuracy across all cases, confirm the effectiveness and adaptability of the quantum annealer in handling these problems, making it a promising tool for future applications in decision-making processes that involve incomplete or weighted ranking data.

\section{Additional figures}\label{addfig}
This section presents supplementary figures to enhance understanding of the results, highlighting key aspects of algorithm performance and behavior.
\begin{figure}[!h]
    \centering
    \includegraphics[scale=0.65]{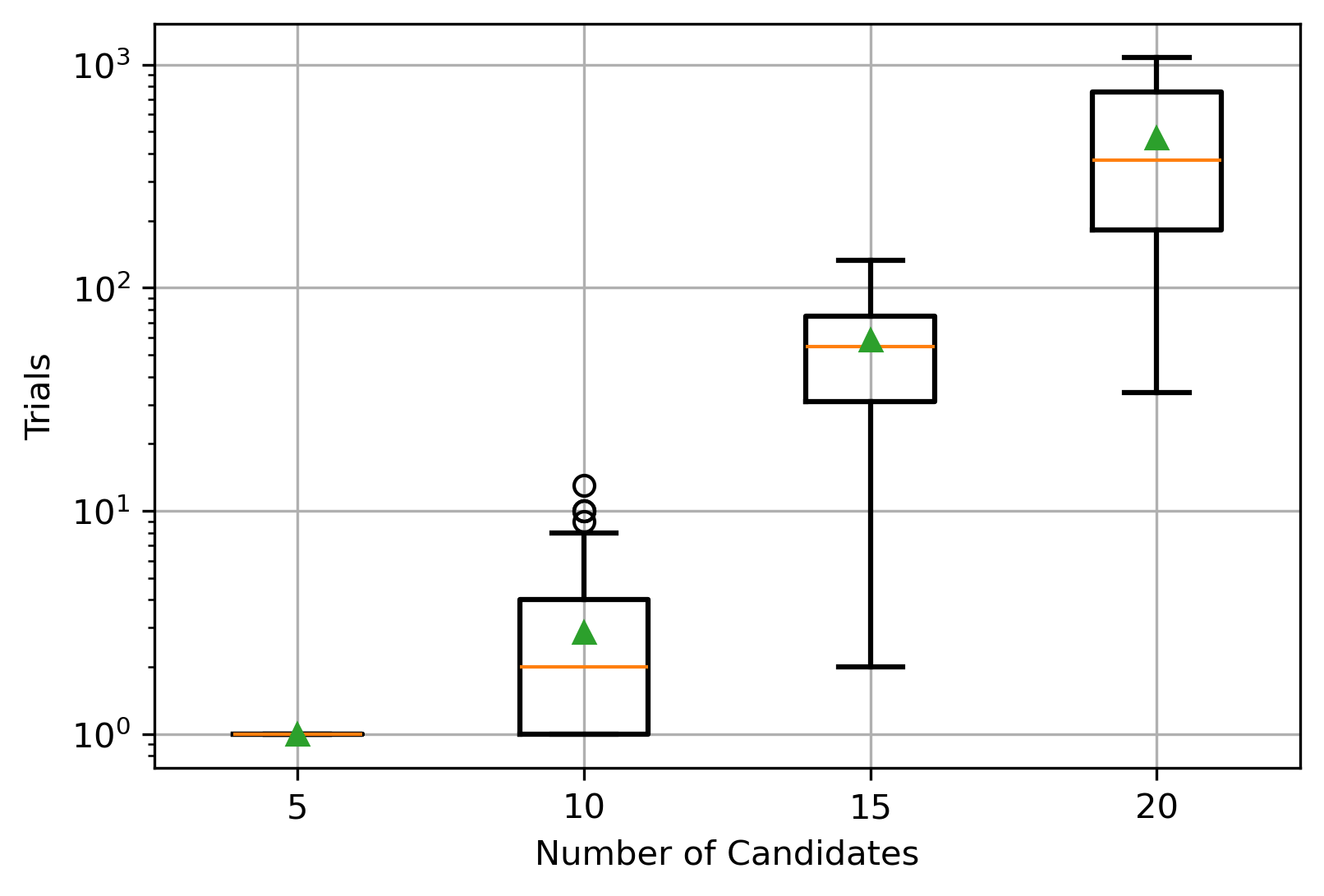}
    \caption{Boxplot of the number of trials required by KwikSort to identify the optimal solution across different datasets.}
    \label{ntrials}
\end{figure}
\begin{figure}[!h]
    \centering
    \includegraphics[scale=0.65]{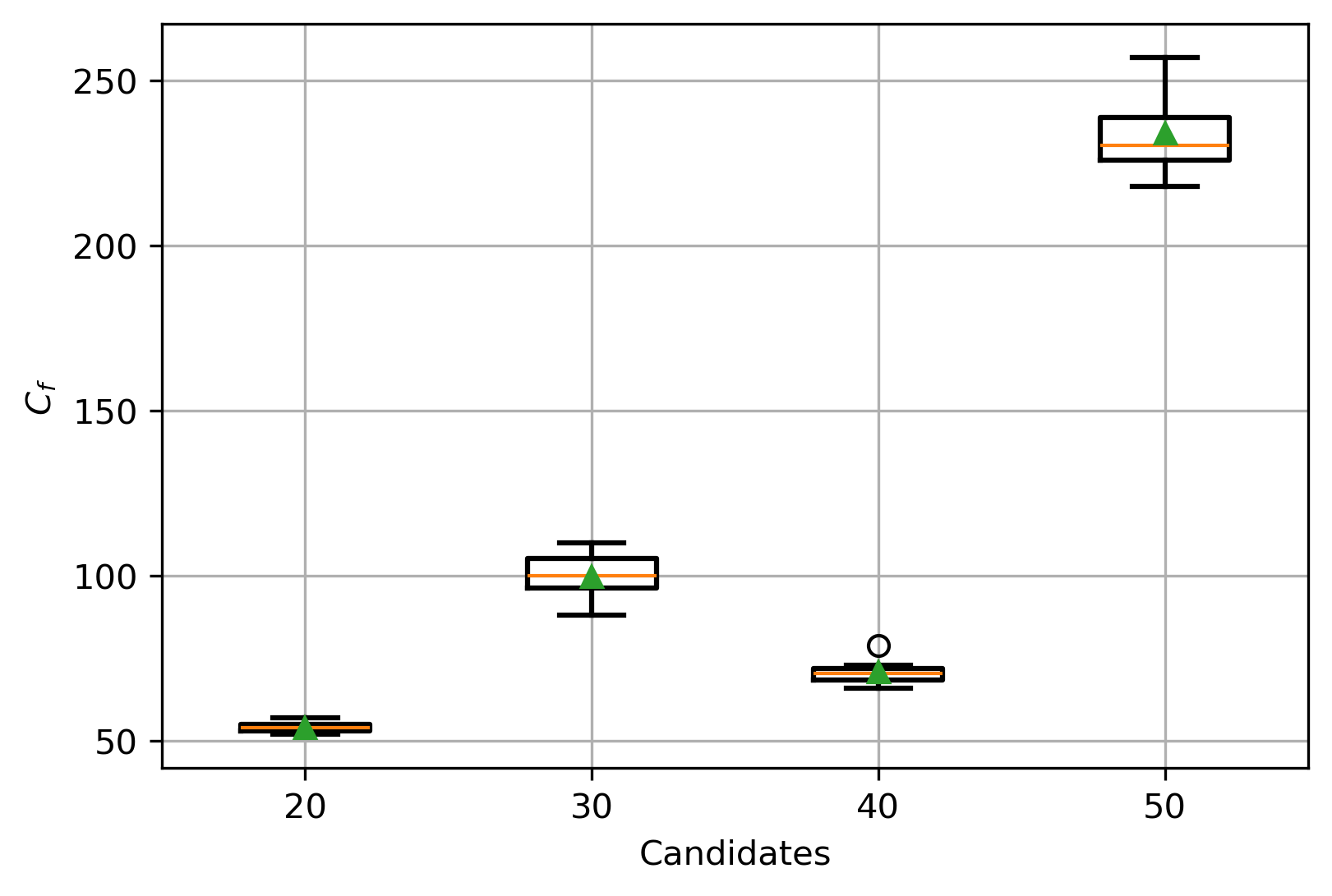}
    \caption{Boxplot of the total number of cycles that appeared during the solution process of the iterative method.}
    \label{average}
\end{figure}




\end{appendices}


\FloatBarrier

\bibliography{sn-bibliography-mia}

\end{document}